%% file: main.tex
\def\isarxiv{1}
 
\ifdefined\isarxiv
\documentclass[11pt]{article}

\usepackage[numbers]{natbib}

\else
\documentclass{article}
\usepackage{neurips_2022}
\fi

\usepackage{amsmath}
\usepackage{amsthm}
\usepackage{amssymb}
\usepackage{algorithm}
\usepackage{algpseudocode}
\usepackage{subfig}

\usepackage{graphicx}
\usepackage{grffile}
\usepackage{wrapfig,epsfig}
\usepackage{url}
\usepackage{xcolor}
\usepackage{epstopdf}
\usepackage{tabularx}
\usepackage{enumerate}

\usepackage{bbm}
\usepackage{dsfont}

\usepackage{lipsum} 
\usepackage[]{mdframed}

\allowdisplaybreaks

\ifdefined\isarxiv

\let\C\relax
\usepackage{tikz}
 \usetikzlibrary{arrows.meta, positioning}
\usetikzlibrary{quantikz2}
\usepackage{hyperref}
\hypersetup{colorlinks=true,citecolor=blue,linkcolor=red}

\usetikzlibrary{arrows}
\usepackage[margin=1in]{geometry}

\graphicspath{{./figs/}}

\else

\usepackage{microtype}
\usepackage{hyperref}
\definecolor{mydarkblue}{rgb}{0,0.08,0.45}
\hypersetup{colorlinks=true, citecolor=mydarkblue,linkcolor=mydarkblue}

\fi
\graphicspath{{./figs/}}

\newtheorem{theorem}{Theorem}[section]
\newtheorem{lemma}[theorem]{Lemma}
\newtheorem{question}[theorem]{Question}
\newtheorem{definition}[theorem]{Definition}

\newtheorem{conjecture}[theorem]{Conjecture}

\newtheorem{remark}[theorem]{Remark}

\newcommand{\ketbra}[2]{\ket{#1}\!\!\bra{#2}}

\newcommand{\wt}{\widetilde}

\newcommand{\eps}{\varepsilon}
\renewcommand{\epsilon}{\varepsilon}
\renewcommand{\phi}{\varphi}

\newcommand{\R}{\mathbb{R}}

\newcommand{\Z}{\mathbb{Z}}
\newcommand{\F}{\mathbb{F}}
\newcommand{\C}{\mathbb{C}}

\newcommand{\tr}{\mathrm{tr}}
\newcommand{\polylog}{\mathrm{polylog}}
\renewcommand{\i}{\mathbf{i}}
\renewcommand{\tilde}{\wt}

\newcommand{\poly}{\mathrm{poly}}

\DeclareMathOperator*{\E}{{\mathbb{E}}}

\newcommand{\SQIOP}{\mathsf{SQIOP}}
\newcommand{\QIOPEPR}{\mathsf{GQIOP}}
\newcommand{\QMA}{\mathsf{QMA}}

\makeatletter
\newcommand*{\RN}[1]{\expandafter\@slowromancap\romannumeral #1@}
\makeatother

\definecolor{b2}{RGB}{51,153,255}
\definecolor{mygreen}{RGB}{80,180,0}

\usepackage{lineno}

\begin{document}

\ifdefined\isarxiv

\date{}

\title{Quantum Interactive Oracle Proofs}
\author{
Baocheng Sun\thanks{\texttt{baocheng.sun@weizmann.ac.il}. Weizmann Institute of Science.
} 
\and
Thomas Vidick\thanks{\texttt{thomas.vidick@epfl.ch}. School of Computer and Communication Sciences and Center for Quantum Science and Engineering, École Polytechnique Fédérale de Lausanne, and Faculty of Mathematics and Computer Science, Weizmann Institute of Science.
}
}

\else
    
    \title{} 
    \maketitle 

\fi

\ifdefined\isarxiv
\begin{titlepage}
  \maketitle
  \begin{abstract}

\input{abstract}
  \end{abstract}
  \thispagestyle{empty}
\end{titlepage}

{\hypersetup{linkcolor=black}
\tableofcontents
}
\newpage
\setcounter{page}{1} 

\else

\fi

\input{intro}

\input{notation}

\input{clifford}

\input{formulation2}

\input{teleportation_approach}

\input{test_n_qubits}

\input{sqiop}

\newpage
\onecolumn
\appendix

\input{test1qubit}

\newpage
\ifdefined\isarxiv

\bibliographystyle{alpha}
\bibliography{ref}

\else
\bibliography{ref}
\bibliographystyle{alpha}
\fi

\end{document}

%% file: abstract.tex
We initiate the study of quantum Interactive Oracle Proofs (qIOPs), a generalization of both quantum Probabilistically Checkable Proofs and quantum Interactive Proofs, as well as a quantum analogue of classical Interactive Oracle Proofs. 

In the model of quantum Interactive Oracle Proofs, we allow multiple rounds of quantum interaction between the quantum prover and the quantum verifier, but the verifier has limited access to quantum resources. This includes both queries to the prover’s messages and the complexity of the quantum circuits applied by the verifier.
The question of whether QMA admits a quantum interactive oracle proof system is a relaxation of the quantum PCP Conjecture.

We show the following two main constructions of qIOPs, both of which are unconditional:
\begin{itemize}
\item We construct a qIOP for QMA in which the verifier shares polynomially many EPR pairs with the prover at the start of the protocol and reads only a constant number of qubits from the prover’s messages. 
\item We provide a stronger construction of qIOP for QMA in which the verifier not only reads a constant number of qubits but also operates on a constant number of qubits overall, including those in their private registers. However, in this stronger setting, the communication complexity becomes exponential. This leaves open the question of whether strong qIOPs for QMA, with polynomial communication complexity, exist.
\end{itemize}

As a key component of our construction, we introduce a novel single prover many-qubits test, which may be of independent interest.

%% file: intro.tex
\section{Introduction}
The classical PCP theorem \cite{alm+98,as98, dinur06} is one of the most significant breakthroughs in computational complexity theory from the early 1990s. Its impact extends beyond practical applications like succinct cryptographic arguments to foundational areas such as hardness of approximation \cite{fgl+91}, and it has profoundly influenced the development of coding theory.

The quantum analog of the PCP theorem, known as the quantum PCP conjecture (qPCP) \cite{an02, aalv09}, has been a subject of exploration for several decades. 
The qPCP conjecture holds significant implications for the foundations of quantum mechanics \cite{aav13}, as exemplified by its relationship with the NLTS (No Low-Energy Trivial States) theorem \cite{fh14}. Research into the qPCP conjecture has yielded both promising \cite{abn23} and negative results \cite{bh13, nn24}. A recent breakthrough \cite{abn23} proving the NLTS theorem has strengthened confidence in the plausibility of qPCP. 
However, the qPCP conjecture remains largely unexplored \cite{aav13}, sparking interest in natural and less restrictive questions derived from the qPCP conjecture.

We introduce a relaxation of the qPCP conjecture analogous to how classical PCPs were generalized to Interactive Oracle Proofs (IOPs) \cite{bcs16, rrr16}. 
Specifically, we introduce quantum Interactive Oracle Proofs (qIOP) and present two incomparable constructions, each using a different technical approach.
Our formulation of qIOP builds upon the quantum PCP conjecture, which is defined in terms of efficient proof verification \cite{aav13}: for any instance of a QMA language, there exists a polynomial-length quantum proof that can be verified by measuring only a constant number of qubits.
Our approach extends the proof system framework to allow multiple rounds of interaction while ensuring that the verifier reads only a constant number of qubits in total. This positions qPCP as a special case of quantum interactive oracle proofs, where restricting to a single round results in a qPCP. 

As a generalization of qPCP, the existence of qIOPs provides evidence supporting the validity of the qPCP conjecture.
Furthermore, qIOPs hold potential for developing quantum succinct arguments for QMA languages \cite{bcs16, klvy22, bkl+22, gjmz23, mnz24}.

\subsection{Our formulation}

We formulate several variants of quantum IOPs. Each formulation is such that when all computation and communication are restricted to be classical, the model reduces to standard classical IOPs. More specifically, we define quantum IOP systems and strong quantum IOP systems as follows.

\begin{definition}[Quantum IOP systems; informal version of Definition~\ref{def:gqiop}]
A \emph{quantum interactive oracle proof} (quantum IOP) system, parameterized by functions $q(n)$ and $l(n): \mathbb{Z}_+ \to \mathbb{Z}_+$, consists of a quantum verifier $\texttt{V}$ satisfying the following conditions for any input of size $n$:

\begin{itemize}
    \item The verifier interacts with the prover for at most $\mathrm{poly}(n)$ rounds.
    \item In each round, the prover sends a {quantum} message to the verifier. The message length may vary between rounds.
    \item In each round, the verifier:
    \begin{itemize}
        \item applies a quantum circuit to his private qubits and to portions of the prover’s message, including qubits received in the current round as well as those retained from previous rounds,
        \item stores part of the prover’s message registers for use in later rounds.
    \end{itemize}
    \item After each interaction, the verifier returns the remainder of the prover's {quantum} message (i.e., the part not stored) and may optionally send back part of their own private {quantum} register.
    \item At the end of the protocol, the verifier outputs a single decision bit.
\end{itemize}
The main complexity parameters are:
\begin{itemize}
    \item The \textbf{query complexity} $q(n)$: the total number of qubits from the prover's messages accessed by the verifier across all rounds.
    \item The \textbf{communication complexity} $l(n)$: the total number of qubits in the prover's messages across all rounds.
\end{itemize}
We refer to such a system as a \emph{quantum IOP} if $q(n)$ is bounded by a constant.\footnote{In the formal Definition~\ref{def:gqiop}, we refer to this as the \emph{general} quantum IOP to distinguish it from the restricted and strong variants.}

We define the following variants of a \emph{quantum IOP}:
\begin{itemize}
    \item A system is called a \emph{restricted quantum IOP} if the private register sent back by the verifier is required to be \emph{classical}.
    \item It is called a \emph{strong quantum IOP} if the verifier must return the \emph{entire} prover's message register in each round (i.e., does not store any part of the message).
\end{itemize}

\end{definition}

\begin{definition}[Quantum IOP systems for a promise problem, informal version of Definition~\ref{def:qiop_}]
A \emph{quantum interactive oracle proof} (quantum IOP) system for a language $\mathcal{L}$ with completeness $c: \mathbb{Z}_+ \rightarrow [0,1]$ and soundness $s: \mathbb{Z}_+ \rightarrow [0,1]$ satisfies the following:
\begin{itemize}
\item For any \emph{YES-instance} of size $n$, 
There exists a polynomial-time prover such that
the verifier accepts with probability at least $c(n)$.
\item For any \emph{NO-instance} of size $n$, for all provers, the verifier accepts with probability at most $s(n)$.
\end{itemize}
The \emph{completeness-soundness gap} of the system is defined as $c(n) - s(n)$.
\end{definition}
\begin{remark}
In the formal definition, we use the notion of query, as defined in Definition \ref{def:q}, to precisely quantify the query complexity and the proof length. This notion of query is motivated by the oracle model used in the literature on quantum query complexity.
\end{remark}

\begin{remark}
Classically, constant-round IOPs, where the verifier reads the prover's message at each round, are equivalent to those where the verifier reads the entire interaction only at the end. Moreover, adaptive and non-adaptive IOPs are equivalent. However, these equivalences do not necessarily hold in the quantum setting.
\end{remark}

\subsection{Main results}

The following are informal statements of our two main results. 

\begin{theorem}[Informal version of Theorem \ref{thm:qpcp_tlp}]\label{thm:inf1}
There exists a quantum IOP system for $\QMA$ with the following parameters: query complexity $q(n) = O(1)$, communication complexity $l(n) = \mathrm{poly}(n)$, and constant completeness-soundness gap. 
\end{theorem}

\begin{remark}
In our construction, the only quantum communication occurs in the first round, where the verifier sends the prover $\poly(n)$ EPR pairs. All subsequent messages exchanged between the prover and the verifier are classical.
\end{remark}

\begin{theorem}[Informal version of Theorem \ref{thm:sqiop}]\label{thm:inf2}
There exists a strong quantum IOP system for $\QMA$ with query complexity $q(n) = O(1)$, communication complexity $l(n) = \exp(\mathrm{poly}(n))$, and constant completeness-soundness gap.
\end{theorem}

For both results the number of rounds of interaction is a constant. We expand on the proof technique for each result in Section~\ref{sec:tech}. In few words, the first protocol is based on teleportation of the witness and adequate use of a PCP of proximity (PCPP) to enable the verifier to correct for the teleportation pad while only querying a constant number of bits from the prover. (However, the verifier does perform a polynomial-size quantum computation on their own private space, which contain the verifier's half-EPR pairs.) The second protocol is based on a multi-qubit test to delegate measurements to the prover and classical PCP techniques to enable the verifier to recover energy measurement outcomes. The multi-qubit test requires exponential quantum communication from the prover to the verifier, even though the verifier only operates on a constant number of qubits in total. Yet we believe the test should be of independent interest and may be a target for subsequent improvements.

\subsection{Open questions}

Classically, a public-coin interactive oracle proof with $O(\log n)$-bit verifier messages can be compiled into a PCP. 
It is also known that a public-coin interactive oracle proof with $\polylog(n)$ rounds, $O(1)$ query complexity, perfect completeness, and $O(1/\poly(n))$ soundness implies the existence of a polynomial-length PCP \cite{abcy22}.
The connection between quantum IOPs (qIOPs) and quantum PCPs (qPCPs) is not as well understood as in the classical setting. 
In fact, it is not even known whether a strong qIOP implies an exponential-size quantum PCP. 
As a result, our qIOP constructions do not yield new insights into the form of local Hamiltonians that could arise from quantum PCP reductions, nor into the reductions themselves. This raises the following question about the connection between the two frameworks:
\begin{question}
Is there a black-box transformation from a quantum IOP to a quantum PCP?
\end{question}

Since our current construction of strong quantum IOPs has exponential communication complexity, it is natural to ask:
\begin{question}
Are there strong quantum IOPs with polynomial communication complexity?
\end{question}

Furthermore, we are also curious about the cryptographic applications of the qIOP model. In \cite{gjmz23}, the authors construct a quantum-communication succinct argument for QMA under assumptions weaker than one-way functions (OWFs), assuming the existence of a quantum PCP. This leads to the following question:
\begin{question}
Are there protocols for delegating quantum computation based on assumptions weaker than OWFs, without relying on a positive resolution of the quantum PCP conjecture?
\end{question}

Another potential cryptographic application is the following:
\begin{question}
Can we apply a Fiat–Shamir transformation to the qIOP protocol to obtain a non-interactive succinct quantum argument for QMA, or a non-interactive quantum argument for QMA with constant query complexity?
\end{question}

In our construction of the strong quantum IOP, we observe that the verifier’s actions have a classical flavor. Specifically, in the protocol, the verifier either applies a Z gate or measures in the Z-basis. This raises a natural question about the extent to which the communication in such a proof system can be classical, while still allowing an efficient quantum prover to convince a quantum verifier of a QMA language. This motivates the following question:
\begin{question}
Does there exist a proof system for QMA languages in which the communication is entirely classical, yet both the prover and the verifier are quantum?
\end{question}

\subsection{Related works}

\paragraph{Interactive Oracle Proofs.}

Classically, IOP emerged after the proof of the PCP theorem \cite{bcs16, rrr16}. Interaction therein aids in reducing proof length, primarily through a tool known as interactive composition. However, our objective is not to simplify but to create something that does not yet exist. Consequently, experiences from classical IOP provide limited insights. Nevertheless, our main motivation is to sketch a possible path towards proving the quantum PCP conjecture, that as a first steps leverages insights from the study of classical IOPs as well as nonlocal games (see below) to, at least initially, avoid some of the main difficulties encountered in traditional approaches to qPCP (such as the no-cloning theorem).

\paragraph{Verification protocols for QMA with a limited quantum verifier.}

There exist a number of interactive protocols for information-theoretically secure verification of QMA, where the verifier has limited quantum capabilities. This includes~\cite{fitzsimons2017unconditionally,aharonov2017interactive,fitzsimons2018post}, see~\cite{gheorghiu2019verification} for a survey. While these protocols are based on a verifier whose quantum operations in each round are constant (or logarithmic) size, the protocols all have a polynomial number of rounds of interaction. Therefore, the total effort of the verifier, as well as total amount of quantum communication received and read from the prover, is always polynomial. Our focus is on a \emph{total} constant effort of the verifier; however, we do not limit the verifier's quantum memory. Therefore, the models are incomparable. We do not borrow techniques from this literature (other than the idea of reducing the QMA language to the local Hamiltonian problem, which is used in~\cite{fitzsimons2018post} and is by now very classical).

\paragraph{Exponential-length classical PCP for QMA.}

Although it is natural to derive qPCP by quantizing the classical PCP transformations, \cite{aav13} takes a step away from this approach, partly because \cite{bh13} suggests that ``the most `natural' transformations may fail in the quantum setting'' \cite{aav13}. More specifically, they demonstrate the existence of exponential-length classical PCP for QMA through the following: $ \mathsf{QMA}\subseteq\mathsf{PSPACE}=\mathsf{IP}$, and the natural transformation from $\mathsf{IP}$ to exponential-length classical PCP. Yet, the resulting proof system works directly with the local Hamiltonian rather than explicitly encoding the quantum witness. As a result, even if the prover possesses the witness, they cannot use it to improve their efficiency. In contrast, for both our constructions the honest prover strategy is based on a natural aglorithm. For the case of Theorem~\ref{thm:inf1} it is efficient (given a supply of copies of the quantum witness). For the case of Theorem~\ref{thm:inf2} it is efficient but for the coherent evaluation of exponentially many Pauli observables, in a very systematic manner that could easily be parallelized.

\paragraph{Game version of quantum PCPs.}

\cite{nv16} proposed a ``games version'' of qPCP for the XZ-Hamiltonian problem, in which multiple provers share entanglement and communicate classically with a verifier in a single-round protocol with constant-length prover messages. It is natural to use this and similar approaches as a starting point, but connecting the multi-prover setting to the single-prover one presents challenges. 

The only prior work we are aware of to establish such a connection is~\cite{wehner06}, which demonstrates the inclusion $\oplus \mathsf{MIP}^*[2] \subseteq \mathsf{QIP}(2)$, thus unconditionally ``compiling'' two quantum provers into a single one. However, their approach seems restricted to testing XOR-type relationships and  protocols involving the exchange of exactly two messages. We did not see how to extend it to more complex two-prover protocols. 

More recent methods~\cite{nz23, mnz24} compress provers using cryptographic techniques, but these rely on computational assumptions and impose restrictions on the prover's computational power which are not inherent to a proof system. More critically, \cite{nn24} notes that \cite{nv16} applies only to the XZ-Hamiltonian problem, while we aim to address QMA. To date, there is no known construction of a game version qPCP for QMA \cite{nn24}. 

Nevertheless, our construction of the strong quantum IOP is inspired by the techniques developed in~\cite{nv16}.
Rather than constructing a game version of qPCP for QMA and subsequently compiling it to qPCP, we adapt the qubit-testing technique used in \cite{nv16} for verifying proofs of the XZ-Hamiltonian and extend it to a single-prover protocol, applying it to QMA languages.

Our construction of a quantum IOP with shared EPR pairs is based on teleportation. 
We note that the protocol in~\cite{gri19} is also based on teleportation. 
In fact, our quantum IOP protocol with shared EPR pairs can be viewed as a single-prover analogue of the multi-prover protocol in~\cite{gri19}.
However, in their setting both provers are untrusted, and the verifier cannot control the operations performed by the prover, so additional self-testing is required. In our setting, the verifier effectively plays the role of one of the provers in \cite{gri19}, but as an honest quantum verifier, there is no need to prove the qubits. Moreover, because the verifier is honest and also acts as the prover for measuring the qubits, we do not need to hide which Hamiltonian term is being measured.

\paragraph{Commuting tests.}
While testing anti-commutation relations in the two-prover setting requires sophisticated constructions, testing commutation is a more elementary and well-established technique \cite{cgjv19}. Ultimately, our protocol for single-qubit tests is close to the protocol in \cite{crsv17} for testing state-dependent commutation. However, in their setting, each individual qubit is already known to anti-commute, and the primary question is whether different qubits commute with each other.

\section{Technical Overview}
\label{sec:tech}

In this paper, unless otherwise specified, we use the term \emph{observable} to refer to an operator that is both Hermitian and squares to the identity. Occasionally, we also use \emph{observable} to refer to an operator of the form $O = \sum_{i} a_i P_i$, where $a_i \in \mathbb{R}$, each $P_i$ is an orthogonal projection, and $\sum_i P_i = I$. The notation $O(a)$ for $a \in \mathbb{F}_2^n$ denotes a tensor product of Pauli $O\in\{X,Y,Z\}$ operators applied to selected qubits according to the bit string $a$:
\[
O(a) = O^{a_1} \otimes O^{a_2} \otimes \cdots \otimes O^{a_n},
\]
where $O^0 = I$ is the identity operator and $O^1 = O$ is the Pauli-$O$ matrix.

\subsection{The Clifford-Hamiltonian problem}

The Clifford-Hamiltonian problem  was introduced in \cite{bjsw16}. This is a special case of the local Hamiltonian problem where each local term has the form  $C^\dagger (\ketbra{0}{0})^{\otimes 5} C$, with $C$ being an element of the Clifford group on $5$ qubits.
The Clifford-Hamiltonian problem is the problem of determining the lowest energy (i.e., the smallest eigenvalue) of a Clifford-Hamiltonian given as input, up to an additive error that is inverse polynomial in the number of qubits.  This problem was shown QMA-complete in \cite{bjsw16}, by adapting the standard Kitaev reduction. Since all our qIOP protocols aim to obtain a constant gap between completeness and soundness, we first show how to  amplify the Clifford-Hamiltonian problem, trading locality for an improved gap.

Let  $H = \mathbb{E}[H_i]$ be a Clifford Hamiltonian. Here each $H_i$ is a $5$-local Clifford projection. After amplification, the Hamiltonian becomes
\begin{align*}
(1 - 2H)^{\otimes N} 
= \big(\mathbb{E}_i[1 - 2H_i]\big)^{\otimes N} = \mathbb{E}_{s_1,\ldots,s_N} \Big[\bigotimes_{j \in [N]} (1 - 2H_{s_j})\Big]\;.
\end{align*}
This amplified Hamiltonian consists of exponentially many terms; however, we can efficiently sample a term. Note that each sampled term is not necessarily local and, in general, does not correspond to a Clifford projection.

In many settings, verifying the measurement result of a Pauli observable is much easier  than verifying the measurement result of a Clifford projection. Therefore, we aim to reduce the measurement of Clifford projections to measurements of Pauli observables.
Note that the meaning of a state lying in a Clifford projection of the form  $C^\dagger (\ketbra{0}{0})^{\otimes 5} C$ is that after applying $C$, all the 5 qubits are in the state $\ket{0}$. Instead of measuring all of them at once, we can measure them one by one. Performing the original projection is then equivalent to projecting the state for each $j \in [5]$ using:
$$
C^\dagger  \Big(I^{\otimes j-1} \otimes \frac{I+Z}{2} \otimes I^{\otimes (5-j)}\Big) C \;.
$$
Because Clifford gates preserve the Pauli group, each of these terms is a projection onto the $+1$ eigenspace of a Pauli observable. By measuring the original state with these $5$ projections and taking the logical AND of the measurement results, we obtain a method to estimate the energy of the amplified Clifford-Hamiltonian using only Pauli observables.

A Hamiltonian is said to be an \emph{XZ-Hamiltonian} if it is a weighted sum of XZ terms, where each term is of the form $X(a)Z(b)$ with $a \wedge b = 0$, meaning that $a, b \in \mathbb{F}_2^n$ have disjoint support.
It is not known whether the XZ-Hamiltonian problem with a constant gap is QMA complete. \cite{nn24} presents several reasons why amplifying the gap while keeping the problem within QMA may be challenging. 
For this reason we mostly work with the the Clifford-Hamiltonian problem. However, it will sometimes be convenient to reason about XZ-Hamiltonian instances; in particular our method for measuring the energy of a Clifford-Hamiltonian essentially uses measurements of XZ-Hamiltonians as a subroutine.

\subsection{General quantum IOPs}
Our first construction is based on teleportation. 
Note that the constant-gap Hamiltonian problem we obtain, namely, the amplified Clifford Hamiltonian, is not local.
Therefore, the prover cannot simply send the witness to the verifier and have the verifier sample and measure a random term, since the verifier only has constant query complexity.
Recall that in the standard definition of qIOP, we only count the number of qubits that the verifier reads from the prover's message. If the verifier and prover share EPR pairs 
\footnote{The shared EPR pairs can also be established by having the verifier send the first message of half EPR pairs. However, in our query-complexity-based definition of qIOP, the verifier is not allowed to send long messages to the prover. The verifier only receives messages from the prover, queries them, and sends them back. If the verifier wishes to transmit information to the prover, he must do so by using the prover's previous message, swapping it with his own register. Therefore, we define the protocol to begin with shared EPR pairs.}
and the verifier measures his half, this measurement is not counted towards the query complexity. Building on this observation, we design the following protocol: the prover and verifier initially share EPR pairs. The prover then performs teleportation and sends the verifier an encoding of the classical one-time pad generated during teleportation.
Since each term in the Hamiltonian is non-local, the verifier still needs to read the one-time pad non-locally to correctly interpret the measurement result. Fortunately, the one-time pad is classical, and we can find a classical solution that enables the prover to provide exactly the information needed by the verifier, in a sound and succinct manner.

Note that the only feature of general quantum IOPs we use is that verifier's messages derived from the verifier’s private qubits are not required to be classical; in particular, the verifier may send arbitrary quantum states. 
In contrast, restricted or strong quantum IOPs allow only classical verifier's messages, except for returning the prover’s message registers.
In particular, in the general quantum IOPs setting, the verifier can send polynomially many half-EPR pairs to the prover in the first round. 

We will show how to estimate the energy of a non-local XZ-Hamiltonian; the approach for estimating the energy of a non-local Clifford-Hamiltonian is similar. 
For an XZ-Hamiltonian, to estimate the energy of the Hamiltonian on a given state, we can randomly sample an XZ term and accept or reject according to the measurement result. The acceptance probability is proportional to the energy. Since the teleported state is masked by a one-time pad, directly measuring it does not yield the correct result on the original state.
However, it is possible to infer the direct measurement outcome of the XZ term from the masked state, given access to the OTP: if measuring $X(a_i)Z(b_i)$ on the masked state yields $r$, and the mask is $X(p)Z(q)$, we know that the measurement on the original state would give $r + p \cdot b_i + q \cdot a_i$. (Here, $a_i,b_i,p,q$ are all $n$-bit strings and operations are mod $2$.)

Note that the observable $X(a_i)Z(b_i)$ that we need to measure is non-local.
We still need to figure out how to encode the one-time pad $(p, q)$ so that the verifier can compute the value $p \cdot b_i + q \cdot a_i$. While the Hadamard code can achieve this, it results in an exponentially long encoding. In the interactive setting, the verifier can request the prover's help, but cannot directly send the chosen XZ-term $(a_i, b_i)$ to the prover, as this would allow the prover to cheat.

The idea is for the prover to first submit a ``commitment" to the one-time pad $(p, q)$, and only then the verifier sends $(a_i, b_i)$. The prover must then prove that the committed values yield the correct result $p \cdot b_i + q \cdot a_i$. To achieve this, we encode the commitment using any error-correcting code, and the proof of correctness is provided through a PCP of proximity. According to the soundness of PCPs of proximity, if the commitment is far from any codeword consistent with a valid $(p, q)$ producing the claimed value, the verifier will reject with high probability.

We refer to this as a commitment in the Interactive Oracle Proof model. Notably, similar techniques have been used for other purposes in the cryptographic literature \cite{lms22, mnz24}.

 One advantage of this protocol is that the communication is entirely classical. However, due to the a priori use of EPR pairs, it remains unclear how to convert it into a succinct argument for QMA.

\subsection{Strong quantum IOPs}
\label{sec:intro-strong}

\paragraph{High-level approach.}

Many prior protocols for verification of QMA originate in multi-prover qubit tests, and we follow the same approach. Early works such as~\cite{nv16} generalize single-qubit anti-commuting games such as CHSH and the Magic Square game to many-qubit tests, and apply these results to game-based versions of qPCPs for the XZ-Hamiltonian problem \cite{nn24}. These game versions involve multiple entangled provers interacting with a classical verifier in a single-round protocol with constant-length messages from the provers.

One way to understand the construction of qIOP is to view it as a compilation of multi-prover, many-qubit tests into a single-prover, many-qubit test protocol, which can then be used to estimate the energy of a Hamiltonian corresponding to some QMA language—mimicking the techniques of \cite{nv16, gri19}. While this has been achieved under computational assumptions \cite{klvy22, mnz24}, it has never before been done in an information-theoretic way.

Converting the messages intended for all provers in a non-local game directly to a single message to the single prover in the qIOP setting would potentially disclose queries intended for other provers to the single qIOP prover, providing it with explicit questions and the capability to respond accurately. 
Classically, encoding the answering strategy of each prover for every possible question compiles a MIP, but with a blow-up in message size that is exponential in the size of the question \cite{bfl90}.
However, for quantum provers sharing entanglement, the prover’s strategies in response to different verifier questions may not commute. As a result, it is not straightforward to list all possible measurement outcomes, since the order in which measurements are performed can affect the result.

If we momentarily set aside the compilation approach, the natural approach to design a qIOP for QMA is as follows. First, the prover sends the verifier a ``commitment" to the quantum witness, which will prevent the prover from altering it midway. The verifier then performs measurements on the committed state, assisted by the prover.
An initial goal in this framework is to design such a commitment scheme that allows the verifier to measure simple observables, such as tensor products of Pauli $X(a), a \in \F_2^n$ and $Z(b), b \in \F_2^n$, with only constant query complexity. These operators commonly appear in many-qubit tests.

Rather than direct compilation, we take inspiration from nonlocal games to design the commitment procedure. 
In our protocol, in the first round, the prover sends the verifier the witness, measured in a purified way in the $X$ basis and encoded using an error-correcting code (specifically, the Hadamard code). In a subsequent round, the witness is measured in the $Z$ basis, also in a purified manner and similarly encoded.
This ``commitment" can be interpreted as the witness being destroyed in the $X$-basis and in the $Z$-basis (corresponding to two types of questions), with the measurement results encoded using an error-correcting code. Because the measurement is purified, the prover has the ability (given all qubits returned to him) to switch from one basis to the other. 

Interestingly, in our protocol, each round features a fixed question. 
Nevertheless, by the end of the protocol, we have asked the prover to provide answers according to X and according to Z, therefore, all possible questions that could have been asked to the provers are always covered.
Moreover, in the final round the verifier does need to sample questions from a distribution, as different questions are required to define distinct observables that characterize the prover's behavior. 
The different observables provide the objects for studying the anticommutation relationship that is key to characterizing the prover's measurements.

\paragraph{Single-prover single-qubit test.}
Following the formulation in \cite{vid20}, a qubit can be understood as a system that can be observed in two mutually incompatible ways—more precisely, it is a system equipped with a pair of anticommuting observables, essentially corresponding to the existence of $X$ and $Z$ operators up to isomorphism. 

We design a simple protocol in which the prover and verifier engage in several rounds of interaction. At the end, the verifier sends the prover a classical bit, and the prover responds with a classical bit. Without loss of generality, in the final round, the prover applies one of two types of observables, depending on the verifier's question bit, on some quantum state. We show that these two observables anti-commute on that state. While it is straightforward to force the prover to measure anti-commuting observables on, say, half of an EPR pair the other half of which is held by the verifier; here the challenge is to evaluate anti-commutation on a state of the prover's choice (i.e., one that is not constrainted by the verifier or the protocol). 

It is interesting to compare this statement with existing research in the multi-prover setting.
In \cite{nv16}, the authors use a standard non-local game where the robustness of the multi-prover anti-commuting game involves both the robustness of the EPR pair and the robustness of the anti-commuting operator acting on the EPR pair. It is well known that these robustness results are sufficient to construct the game version of the PCP for the XZ-Hamiltonian.
What is the proper condition that should hold for a meaningful anti-commuting relationship in order to achieve the desired QMA result for a single prover game?
Should we test whether the single prover holds EPR pairs? Since the original provers act on different Hilbert spaces, should we expect that, after compression, operators originally from different provers still commute with each other? At first glance, both of these questions might seem to have affirmative answers. However, neither of these appears as a subroutine in our final protocol. Intuitively, this makes sense because, in non-local games with EPR pairs, an operator acting on the first prover is equivalent to its transpose acting on the second prover. EPR pairs serve as a bridge in this context. Since there is only one prover, there is no need to bridge between them; both provers can act on the same qubits, and the operators no longer commute.

\paragraph{The construction.}

Our protocol consists of three rounds. 
\begin{itemize}
\item In the first round the honest prover sends the purified measurement of their state in the $Z$ basis. The purified measurement outcome is stored in a register $\mathsf{R}$.
\item The verifier can choose to either apply a phase flip gate ($Z$ gate) on register $\mathsf{R}$, measure it in the computational basis ($Z$ basis), or do nothing (apply the identity). Then, the verifier sends register $\mathsf{R}$ back to the prover. 
\item In the second round, the honest prover undoes the purified measurement in the $Z$ basis and sends the purified measurement of their state in the $X$ basis (in the same register $\mathsf{R}$). 
\item Again, the verifier may choose to apply a $Z$ gate, measure in the $Z$ basis, or do nothing.
\item In the third round, the verifier sends $\mathsf{R}$ back to the prover, together with a one-bit question, either ``X" or ``Z". 
\item The prover undoes the purified measurement in the $X$ basis and returns the measurement result corresponding to the verifier's question. 
\end{itemize}
Note that applying a phase flip gate to the purified measurement result of an observable is equivalent to applying the observable itself directly as a gate. 
Measuring the purified measurement result of an observable in the computational basis is equivalent to directly measuring the observable on the state.

In the protocol, the verifier uniformly samples one of the following four tests uniformly at random and performs it with the prover:
\begin{enumerate}
    \item \textbf{Z consistency test:} The verifier measures in the first round, does nothing in the second round, sends ``Z'' to the prover in the third round, and accepts if and only if the measurement result from the first round matches the prover's answer bit to the ``Z'' query in the third round.
    
    \item \textbf{X consistency test (standard):} The verifier does nothing in the first round, measures in the second round, sends ``X'' to the prover in the third round, and accepts if and only if the measurement result from the second round matches the prover's answer bit to the ``X'' query in the third round.
    
    \item \textbf{X consistency test (flipped):} The verifier applies a phase flip in the first round, measures in the second round, sends ``X'' to the prover in the third round, and accepts if and only if the measurement result from the second round matches the prover's answer bit to the ``X'' query in the third round.
    
    \item \textbf{Anti-commuting test:} The verifier measures in the first round, applies a phase flip in the second round, sends ``Z'' to the prover in the third round, and accepts if and only if the measurement result from the first round is \emph{opposite} 
    to the prover's answer bit to the ``Z'' query in the third round.
\end{enumerate}
If we interpret the observable applied by the prover in the final round as a unitary operator, then both $Z$ and $X$ can be applied together in the last round, and in different orders. The anti-commutation relation $XZ=-ZX$ between those observables is precisely the statement we aim to prove.
The anti-commuting test already provides an initial equation in the form of an anti-commutation relation, while the consistency tests serve to shift the operations in this equation to the observable used in the final round. 
Indeed can view $ZX$ and $XZ$ as a sequence of operations, $Z$ followed by $X$ in the case of $XZ$. 
These operations can be distributed across different stages of the interaction. 
For instance, we can apply $Z$ in both the first and last rounds. 
Another example is that applying a phase flip in the first round and measuring in the second corresponds to applying $Z$ followed by $X$. 
The relationship between applying $Z$ and $X$ in different rounds is enforced by the consistency tests.

The preceding operation justifies the flipped version of the consistency test as a means to induce different sequences of observables. From the perspective of non-local games, there is another explanation why this test is important. Intuitively, without flipping, the prover is only constrained to play the game with a single initial state of their choosing. By introducing flipping (which the verifier only performs before all measurements), the prover is forced to play the game across multiple initial states.
By comparing the protocol with non-local games, we can derive several interesting observations.

\paragraph{Comparison with non-local games.}

To illustrate the underlying intuition, we present a construction based on a single-prover variant of a non-local game. 
Since we did not identify any particular advantage gained by using the non-local games framework, our main results on quantum IOPs and strong quantum IOPs are not based on the non-local game-based tests described here.
Therefore, we omit the formal statement.

We use the Magic Square game as an illustrative example, but the construction extends naturally to other non-local games. 

\begin{center}
\begin{tabular}{|c|c|c|c|}
\hline
      & c1 & c2 & c3 \\\hline
r1 & (1) $I\otimes Z$ & (2) $Z\otimes I$ & (3) $Z\otimes Z$ \\\hline
r2 & (4) $X\otimes I$ & (5) $I\otimes X$ & (6) $X\otimes X$ \\\hline
r3 & (7) $X\otimes Z$ & (8) $Z\otimes X$ & (9) $Y\otimes Y$ \\\hline
\end{tabular}
\end{center}

The protocol consists of 6 rounds (12 messages), each corresponding to either a row (indexed r1, r2, r3 from top to bottom) or a column (indexed c1, c2, c3 from left to right) in the Magic Square game. The rounds follow a fixed order: r2, c3, r1, c2, r3, c1.\footnote{While a specific order is not strictly required, not every ordering will work. The order emerges from deduction. In particular, we need to preserve the following sequence: the index of the row containing position 4 must appear first, followed by the index of the column containing position 9 (the third element of the third row), and then the index of the row containing position 2. Likewise, the index of the column containing position 2 should appear first, followed by the index of the row containing position 9, and then the index of the column containing position 4. This ordering enables us to transform the pair (2, 4) into an expression involving 9 in the middle, and eventually into (4, 2).
} In each round, the honest prover sends a purified measurement of the relevant row or column (2 qubits, the third can be inferred). The verifier either measures in the $Z$ basis or applies a $Z$ gate to some of the qubits or does nothing, then sends the qubits back. The honest prover then undoes the purified measurement.

At the end of the interaction, the verifier sends an index in $\{1,\ldots,9\}$ (corresponding to the 9 positions in the Magic Square), and the prover returns a single bit. The only tests applied are consistency checks: ensuring consistency at the intersection of rows and columns, and consistency between the first 6 rounds and the final round. 
These tests also include various forms of flipping prior to the two measurements for consistency checking. 
Since there are only a constant number of flipping choices, we can perform each possible tests with constant probability.
The anti-commutation relationship manifests itself in the final round, specifically between position 2 (the second element of the first row) and position 4 (the first element of the second row).

Intuitively, requiring the prover to sequentially measure all rows or columns forces them to adopt a quantum strategy. In effect, the prover is committing to answers for all possible questions. Note that, since the game has no classical solution, it's impossible for the prover to respond correctly to every question in a single round.

Surprisingly, without all the flipping tests, although soundness is preserved, the protocol is not robust in the sense that it fails to enforce anti-commuting relationships. We were able to construct counterexamples that pass the single-prover version of the Magic Square tests using only commuting operators. However, in all of these counterexamples, the initial state is uniquely determined by the operators used.
Another perspective for understanding the flipping test is by comparing it to the role of consistency relations between the two provers. In the multi-prover setting, consistency relations can be used to switch which player an operation is applied to. This plays a crucial role in proving the robustness of the Magic Square game.
Such ``player switching" also plays an important role after establishing the anti-commutation relation between X and Z. This switching helps derive relations between elements of the Weyl–Heisenberg group:
The goal is to show that $X(a)Z(b)X(c)Z(d) = (-1)^{b c}X(a+c)Z(b+d)$, with a key subgoal being the ability to ``flip" the middle operators $Z(b)X(c)$. Consistency checks are used to enable ``player-switching", which in turn allows us to derive that
$X(a)Z(b)X(c)Z(d) - X(a)Z(b)X(c) \otimes Z(d) = 0$, thereby making it possible to flip $Z(b)X(c)$ within $X(a)Z(b)X(c)$.

\paragraph{Other difficulties in analyzing the robustness.}

Note that in the middle of the protocol, the honest prover does not necessarily apply a purified observable, but may instead apply an arbitrary unitary. This creates a difficulty for us, as we would prefer to formulate the prover's action in the middle round as an observable, i.e., an operator that is both Hermitian and squares to the identity, so that it can be directly compared with the observable in the final round. Moreover, the prover might apply an additional unitary to their state before the next round of interaction, which may allow them to cheat. However, this exhausts the flexibility that the prover may have. 

To address these issues, we can reinterpret the verifier's honest measurement or phase-flip gate as part of the prover's operation, absorbing them into the prover's action and treating them as the prover's observables. While the prover may still apply additional unitaries between rounds, this does not significantly impact our analysis. 
 
\paragraph{Many-qubits tests.}

We follow the standard approach from \cite{nv16} to extend the single-qubit test to a many-qubit test. Specifically, in each round, the prover sends the verifier the measurement of their witness in the basis corresponding to that round: Z in the first round, X in the second, and in the third round, acts according to the verifier’s question. The prover always responds with the Hadamard encoding of the purified measurement result. Although the prover runs in exponential time in terms of time complexity, they can still be implemented by a logarithmic-depth circuit using CNOT gates with arbitrary fan-out.

\paragraph{Testing the energy of a Clifford-Hamiltonian.}

To estimate the observable $X(a)Z(b)Y(c)$, we can leverage the first and second rounds of the qubits test. Measuring an observable is equivalent to applying the controlled version of that observable as a unitary, with a control qubit initialized in the $\ket{+}$ state, followed by a measurement of the control qubit in the $X$ basis.
In our setting, the verifier performs a CNOT gate on positions $b + c$ in the first round, and another CNOT on positions $a + c$ in the second round, along with an additional application of the phase gate $c$ times on the control qubit.
This procedure allows us to test the energy of a Clifford-Hamiltonian.

\paragraph{The challenge of making the communication polynomial.}

Note that the construction we described has exponential communication complexity. In each round, whether measuring observables of the form $X(a)$ or $Z(b)$, the prover must simultaneously send the measurement results for all $a, b \in \mathbb{F}_2^n$, so they cannot anticipate which specific position the verifier will read in that round.

We attempted to follow the approach of \cite{dls22} to reduce the protocol length to polynomial. Instead of using a $\lambda$-biased set as in \cite{dls22}, we chose to use a classical locally testable code. This choice stems from the fact that, although possible, it is complicated to test whether the prover is indeed sending a superposition of inner products between a $\lambda$-biased set and some string, i.e., the expected purified measurement result.

Moreover, in the setting of \cite{dls22}, using any error-correcting code yields similar results. However, a crucial property required in their analysis is that the prover's state must be totally mixed. This assumption is essential for ensuring that the representation $\pi(a)X = U(a)XU(a)^\dagger$ forms a unitary representation, which is necessary for applying the Poincar\'e inequality. 
In the multi-prover setting, this property holds without loss of generality, as the provers can be assumed to start with EPR pairs \footnote{This is because the game is a synchronous game.}. However, in our setting, the single prover is allowed to choose an arbitrary initial state, so this assumption no longer holds.

\section{Organization}

In Section~\ref{sec:notation} we introduce the necessary notation and  conventions used throughout the paper. We also define the Clifford-Hamiltonian problem and show an amplification of the gap.
In Section~\ref{sec:formulation}, we formally define quantum interactive oracle proofs (qIOPs) and their variants.
Section~\ref{sec:mqt} extends the discussion to testing many qubits, and based on this, Section~\ref{sec:sqiop} presents our construction of strong quantum IOPs.
Appendix~\ref{sec:sqt} presents our approach to testing single qubits. 
While not directly connected to the main results, this section elaborates the ideas behind testing qubits in the single-prover setting that was already given earlier.

\paragraph{Acknowledgments.} T.V. is supported by
AFOSR Grant No. FA9550-22-1-0391 and the Swiss State Secretariat for Education, Research and Innovation (SERI). Parts of this work was completed while T.V. was at the Weizmann Institute in Israel, where he was supported by a research grant from the Center for New Scientists.
B.S. is supported by a scholarship from the Weizmann School of Science. The authors would like to thank Prof. Zvika Brakerski for helpful discussions.

%% file: notation.tex
\section{Preliminaries}
\label{sec:notation}

In this section, we present the notation and conventions used throughout the paper to ensure clarity and consistency.

\paragraph{Register and space notations.}
\begin{itemize}
    \item \textbf{Registers:} Capital letters in sans serif font, such as $\mathsf{X, Y, Z}$, are used to denote quantum registers.
    \item \textbf{Hilbert spaces:} Capital calligraphic letters, such as $\mathcal{X, Y, Z}$, represent the complex Hilbert spaces associated with the respective registers.
    \item \textbf{Quantum channels:} Blackboard bold letters, such as $\mathbb{P, Q}$, or capital Greek letters, such as $\Phi, \Psi$, are used for quantum channels.
    \item \textbf{Algorithms:} Monospace capital letters, such as $\texttt{A, B, C}$, denote algorithms.
\end{itemize}

\paragraph{Controlled unitaries.}
To simplify our presentation, we introduce notation for  controlled unitary operations.
\begin{definition}[Controlled unitary]
Let $\mathsf{P}$ be a register and $U_\mathsf{P}$ a unitary acting on it. Let $\mathsf{R}$ be a one-qubit register.
The controlled unitary $C(U_\mathsf{P}, \mathsf{R})$ applies $U_\mathsf{P}$ if the qubit in $\mathsf{R}$ is $|1\rangle$:
\begin{align*}
    C(U_\mathsf{P}, \mathsf{R}) = \begin{pmatrix} I_\mathsf{P} & 0 \\ 0 & U_\mathsf{P} \end{pmatrix}
\end{align*}
More generally, for a register $\mathsf{R}'$ of dimension $d$ and $i \in [d]$, the controlled unitary $C(U_\mathsf{P}, \mathsf{R}, i)$ applies $U_\mathsf{P}$ if the qubit in register $\mathsf{R}$ is $\ket i$:
\begin{align*}
    C(U_\mathsf{P}, \mathsf{R}, i) = \mathrm{diag}(\underbrace{I_\mathsf{P}, \ldots, I_\mathsf{P}}_{i-1}, U_\mathsf{P}, \underbrace{I_\mathsf{P}, \ldots, I_\mathsf{P}}_{d-i})
\end{align*}
\end{definition}

\paragraph{Purified measurements.}
We define purified measurements, which simulate measuring an observable and storing the result in a register.
\begin{definition}[Purified measurement]
Let $\mathsf{P}, \mathsf{R}$ be registers and $O_\mathsf{P}$ an observable decomposed as $O = \sum_{i=1}^r a_i P_i$, where each $P_i$ is an orthogonal projection and $I = \sum_i P_i$. The unitary that performs the purified measurement and stores the outcome in $\mathsf{R}$ is:
\begin{align*}
    P(O_\mathsf{P}, \mathsf{R}) = \sum_{i=1}^r P_{i, \mathsf{P}} \cdot (X(a_i))_\mathsf{R}
\end{align*}
\end{definition}
Note that applying $P(O_\mathsf{P}, \mathsf{R})$ and then measuring $\mathsf{R}$ is equivalent to directly measuring $O_\mathsf{P}$ and storing the result in $\mathsf{R}$.

\paragraph{EPR basis.}
We define the EPR basis as follows:
\begin{definition}[EPR basis]
\label{def:epr_basis}
Let $|\phi^+\rangle = \frac{1}{\sqrt{2}}(|00\rangle + |11\rangle)$. The EPR basis states are:
\begin{align*}
    \ket{E_{00}} &= \ketbra{\phi^+}{\phi^+}, \\
    \ket{E_{10}} &= (Z \otimes I) \ketbra{\phi^+}{\phi^+} (Z \otimes I), \\
    \ket{E_{01}} &= (X \otimes I) \ketbra{\phi^+}{\phi^+}( X \otimes I), \\
    \ket{E_{11}} &= (ZX \otimes I) \ketbra{\phi^+}{\phi^+} (XZ \otimes I).
\end{align*}
\end{definition}

\paragraph{Subregisters and substrings.}
We introduce notation for accessing subregisters and substrings:
\begin{definition}
Let $\mathsf{A} = (\mathsf{A}_1, \ldots, \mathsf{A}_n)$ be a tuple of registers, and $S = \{s_1, \ldots, s_k\} \subset [n]$ with $s_1 < \cdots < s_k$. Define $\mathsf{A}_S = (\mathsf{A}_{s_1}, \ldots, \mathsf{A}_{s_k})$.
\end{definition}

\begin{definition}
Let $n, m \in \mathbb{Z}_+$, $w \in \mathbb{F}_2^n$, $\mathsf{T} = (\mathsf{T}_i)_{i \in [n]}$, and $s = (s_1, \ldots, s_m) \in [n]^m$. Define:
\begin{align*}
    w(s) = (w_{s_1}, \ldots, w_{s_m}), \quad \mathsf{T}(s) = (\mathsf{T}_{s_j})_{j \in [m]}\;.
\end{align*}
\end{definition}

\paragraph{Inner product with respect to a state.}
\begin{definition}
For a positive semidefinite operator $\sigma$, define:
\begin{align*}
    \langle A, B \rangle_\sigma = \mathrm{Tr}(A^\dagger B \sigma), \quad \|A\|_\sigma^2 = \mathrm{Tr}(A^\dagger A \sigma)\;.
\end{align*}
\end{definition}

%% file: clifford.tex
\subsection{Gap amplification of the Clifford-Hamiltonian problem}

The Clifford-Hamiltonian problem was introduced in \cite{bjsw16}. A Clifford-Hamiltonian is a local Hamiltonian such that each local term has the form  $C^\dagger (\ketbra{0}{0})^{\otimes 5} C$, with $C$ being an element of the $5$-fold Clifford group.
When choosing a Hamiltonian problem for amplification, we care about two key parameters: the energy gap and the $\ell_1$-norm of the Hamiltonian. The $\ell_1$-norm refers to the sum of the absolute values of the coefficients when decomposing the Hamiltonian into a linear combination of observables.
In Kitaev's clock construction, the minimal energies for both YES and NO instances are very close to $0$, making it a poor starting point for gap amplification. The standard approach for amplification begins by shifting the Hamiltonian $H$, for example, replacing it with $I - H$. However, this shifting causes the $\ell_1$-norm to grow beyond $1$, and when repeated, it can blow up exponentially.

In practice, since we ultimately measure observables, all we need is that after shifting, the Hamiltonian becomes a sum of observables. Then, the $\ell_1$-norm is naturally $1$ before amplification (i.e., repetition).
A key advantage of using Clifford-Hamiltonian terms is that they are projections $P$, which allows us to shift them to $1 - 2P$, yielding observables. Then, we can amplify the gap through repetition.

Note that $$
C^\dagger (\ketbra{0}{0})^{\otimes 5} C=\prod_{j\in[5]} C^\dagger \cdot \Big(I^{\otimes j-1} \otimes \frac{I+Z}{2} \otimes I^{\otimes (5-j)}\Big) \cdot C ,
$$
which may indicate that we can also transform a Clifford projection directly into a sum of XZ-Hamiltonians. This may suggest that we can now measure XZ-Hamiltonians. However, the transformation does not preserve the $\ell_1$ norm. 
After amplification, the situation becomes even worse: a Clifford Hamiltonian with $\ell_1$ norm $1$ may be transformed into a sum of XZ-Hamiltonians with exponentially larger $\ell_1$ norm. As a result, although amplification increases the spectral gap, the exponential blow-up in the norm causes the effective gap to shrink exponentially, rendering amplification ineffective.

The following definition describes a specific type of local Hamiltonian term used in QMA-complete constructions. Each term projects onto a Clifford-transformed all-zero state on a subset of qubits, while acting as the identity on the remaining qubits.
\begin{definition}[Clifford-Hamiltonian projection]
A Hamiltonian (hermitian matrix) $H$ is a  Clifford-Hamiltonian projection iff $H$ acts non-trivially on a set $S$ of qubits as $C^\dagger \ketbra{0^{|S|}}{0^{|S|}} C$, where $C$ is an element of the $|S|$-fold Clifford group. More specifically,
\begin{align*}
H = (C^\dagger  |0^{|S|}\rangle\langle0^{|S|}| C)_{S} \otimes I_{[n]-S}.
\end{align*}
\end{definition}

This definition generalizes the projection above to a full Hamiltonian composed of multiple such projections. Averaging them ensures the spectrum of $H$ lies between $0$ and $1$.
\begin{definition}[Clifford-Hamiltonian]
A Hamiltonian $H \in \C^{2^n} \otimes \C^{2^n}$ is a Clifford-Hamiltonian if it is the uniform average of a polynomial number $m=\poly(n)$ of Clifford-Hamiltonian projection terms $H_i$. That is,
\begin{align*}
H = \frac{1}{m} \sum_{i \in [m]} H_i\;.
\end{align*}

Moreover, $H$ is $k$-local if each $H_i$ acts non-trivially on at most $k$ qubits.
\end{definition}

This problem asks whether the ground-state energy of a given local Clifford-Hamiltonian lies
below $a(n)$ or above $b(n)$, with a promise that one of the two holds. 

\begin{definition}[Local Clifford-Hamiltonian problem]
Let $k \in \Z^+$ and let $a,b:\Z^+ \rightarrow \R^+$ be non-decreasing functions such that $\forall n \in \Z^+,~ a(n) < b(n)$. The local Clifford-Hamiltonian problem $\mathcal{LCH}(k,a,b)$ is the promise problem $(S_{yes}, S_{no})$ defined by:
\begin{itemize}
    \item $S_{yes}$: the set of $k$-local Clifford-Hamiltonians $H\in \C^{2^n} \otimes \C^{2^n}$ such that $\lambda_{\min}(H) \leq a(n)$,
    \item $S_{no}$: the set of $k$-local Clifford-Hamiltonians $H\in \C^{2^n} \otimes \C^{2^n}$ such that $\lambda_{\min}(H) \geq b(n)$.
\end{itemize}
\end{definition}

This theorem from \cite{bjsw16} shows that the local Clifford-Hamiltonian problem is QMA-complete for 5-local terms and polynomial small completeness-soundness gaps. It provides a concrete and structured Hamiltonian family for studying QMA.

\begin{theorem}[Theorem 2 in \cite{bjsw16}]
\label{thm:clifford_ham}
There exist polynomials $p$ and $q$ such that the problem $$\mathcal{LCH}(5, 2^{-p(n)}, 1/q(n))$$ is QMA-complete.
\end{theorem}

This lemma presents a gap amplification technique that enhances the spectral gap between YES and NO instances of the Clifford-Hamiltonian problem. Specifically, the transformation $H' = (1 - 2H)^{\otimes N}$ increases the gap exponentially. However, this comes at the cost of breaking the locality structure of the Hamiltonian. The proof leverages the spectral properties of Hermitian operators and the multiplicative behavior of their tensor powers.
\begin{lemma}[Gap amplification of the Clifford-Hamiltonian problem]\label{lem:c_ham_amp}
Let $(S_{yes}, S_{no}) = \mathcal{LCH}(k, a, b)$ with $a, b \in (0, 0.1)$. Let $H \in S_{yes} \cup S_{no}$, and define $H' = (1 - 2H)^{\otimes N}$ for some odd $N \in \Z^+$. Then:
\begin{enumerate}
    \item If $H \in S_{yes}$, then $\lambda_{\max}(H') \geq (1 - 2a)^N$.
    \item If $H \in S_{no}$, then $\lambda_{\max}(H') \leq (1 - 2b)^N$.
\end{enumerate}

Furthermore, let $p, q$ be as in Theorem~\ref{thm:clifford_ham}. For $k=5$, $a=2^{-p(n)}$, $b=1/q(n)$, and $N=q^2(n)$, there exists a polynomial $s$ such that:
\begin{enumerate}
    \item If $H \in S_{yes}$, then $\lambda_{\max}(H') \geq 1 - 2^{-s(n)}$.
    \item If $H \in S_{no}$, then $\lambda_{\max}(H') \leq 2^{-s(n)}$.
\end{enumerate}
\end{lemma}

\begin{proof}  
If $H\in S_{yes}$, then
\begin{align*}
\lambda_{\max}(H')\geq (1-2a)^N \geq 1-2aN \geq 1-2^{-s(n)}\;.
\end{align*}
If $H\in S_{no}$, then
\begin{align*}
\lambda_{\max}(H')\leq (1-2b)^N = (1-2b)^{1/(2b)\cdot 2bN}  \geq 1-2^{-s(n)}\;.
\end{align*}
\end{proof}

\subsection{Probabilistically checkable proofs of proximity}

In this section, we recall the notion of probabilistically checkable proofs of proximity (PCPPs), introduced in~\cite{bgh+04}, and their application to the canonical $\cal P$-complete problem \textsc{Circuit Value}. These notions play a central role in our deductions.

\begin{definition}[\cite{bgh+04}]
We define the $\mathcal{P}$-complete language \textsc{Circuit Value} as $\textsc{CktVal} = \{(C, w) : C(w) = 1\}$, where $C$ is a Boolean circuit and $w$ is an input string.  
For a pair language $L \subseteq \{0,1\}^* \times \{0,1\}^*$, we define $L(x) = \{y : (x, y) \in L\}$.
\end{definition}

The following definition formalizes a proof system in which the verifier can efficiently check whether a given input is close to being valid, without reading the entire input. 

\begin{definition}[PCPs of Proximity (PCPPs), \cite{bgh+04}]
\label{def:pcpp}
Let $s, \delta: \mathbb{Z}^+ \rightarrow [0, 1]$ be functions. A verifier $V$ is said to be a probabilistically checkable proof of proximity (PCPP) system for a pair language $L$ with proximity parameter $\delta$ and soundness error $s$ if, for every pair of strings $(x, y)$, the following conditions hold:
\begin{itemize}
    \item \textnormal{Completeness:} If $(x, y) \in L$, then there exists a proof string $\pi$ such that $V(x)$ accepts the oracle $y \circ \pi$ with probability $1$.
    \item \textnormal{Soundness:} If $y$ is $\delta(|x|)$-far from $L(x)$, then for every $\pi$, the verifier $V(x)$ accepts the oracle $y \circ \pi$ with probability strictly less than $s(|x|)$.
\end{itemize}
If $s$ and $\delta$ are not explicitly specified, they are assumed to be constants in the interval $(0, 1)$.
\end{definition}

The following theorem, due to~\cite{bgh+04}, establishes that the \textsc{Circuit Value} problem admits a PCPP. 

\begin{theorem}[\cite{bgh+04}]
\label{thm:pcpp}
There exists a constant $\delta_0 \in (0, 1)$ such that for all $n \in \mathbb{Z}^+$ and all $\delta < \delta_0$, the language \textsc{Circuit Value} has a PCP of proximity (for circuits of size $n$) with the following parameters:
\begin{itemize}
    \item randomness complexity $O(\log n)$,
    \item query complexity $O(1)$,
    \item perfect completeness, and
    \item soundness error $1 - \Omega(\delta)$ for proximity parameter $\delta$.
\end{itemize}
\end{theorem}

\subsection{Group stability}
\label{sec:siop:pre}

In this section, we review several mathematical tools and foundational results that will be used throughout our analysis. 
We do not directly copy from existing literature, either because the results are not explicitly stated or because they do not fully align with our requirements.
We begin by introducing the $\sigma$-weighted Hilbert–Schmidt positive semi-definite Hermitian form and seminorm, which allow us to reason about approximate representations and operator comparisons in the presence of a quantum state. We then recall a fundamental stability result for approximate representations of the Pauli group due to Gowers and Hatami, which will serve as the core analytic tool in our many-qubits test.

We define a positive semi-definite Hermitian form and seminorm with respect to a positive semidefinite operator $\sigma$ as follows. Similar definitions already appear in~\cite{vidick2011complexity} and in much of the literature on nonlocal games. 

 \begin{definition}
For a positive semidefinite operator $\sigma$, define:
\begin{align*}
    \langle A, B \rangle_\sigma = \mathrm{Tr}(A^\dagger B \sigma), \quad \|A\|_\sigma^2 = \mathrm{Tr}(A^\dagger A \sigma).
\end{align*}
\end{definition}

The following useful properties are well-known. 

\begin{lemma}
Let $\sigma$ be a positive semidefinite operator. Then for any operators $A, B$:
\begin{align*}
\|A\|_\sigma^2 &\geq 0, \\
\|A + B\|_\sigma^2 &\leq 2\|A\|_\sigma^2 + 2\|B\|_\sigma^2, \\
\|AB\|_\sigma^2 &\leq \|A\|^2 \cdot \|B\|_\sigma^2,
\end{align*}
and if $C^\dagger C = I$, then
\[
\langle CA, CB \rangle_\sigma = \langle A, B \rangle_\sigma.
\]
\end{lemma}

\begin{proof}
The first and last items follow by definition. For the second, 
\begin{align*}
\|A+B\|^2_\sigma = &~ \tr((A+B)^\dagger(A+B) \sigma ) \\
= &~ 2\tr(A^\dagger A \sigma )+2\tr(B^\dagger B \sigma )-\tr((A-B)^\dagger(A-B) \sigma ) \\
\leq &~ 2 \|A\|^2_\sigma + 2 \|B\|^2_\sigma\;.
\end{align*}
For the third, 
\begin{align*}
\|A B \|^2_\sigma 
= &~ \tr( A^\dagger A B \sigma B^\dagger) \\
\leq &~ \|A^\dagger A\| \cdot \tr(  B \sigma B^\dagger) \\
=&~ \|A\|^2 \cdot \|B \|^2_\sigma .
\end{align*}
\end{proof}

We next recall the triangle inequality for the $\sigma$-seminorm, see e.g.~\cite[Fact 4.28]{natarajan2019neexp}.

\begin{lemma}[Triangle inequality for $\sigma$-seminorm]
\label{lem:tri_sigma_norm}
Let $\sigma$ be a positive semidefinite operator. Then for any operators $A$ and $B$,
\[
\|A + B\|_\sigma \leq \|A\|_\sigma + \|B\|_\sigma.
\]
\end{lemma}

We now introduce several definitions that will be useful in the context of  group representations and approximate representations of group.

\begin{definition}[Unitary group]
Define $U(\mathbb{C}^d)$ as the group of $d \times d$ unitary matrices.
\end{definition}

\begin{definition}[Approximate representation]
Given a finite group $G$, a dimension $d \geq 1$, $\varepsilon \geq 0$, and a $d$-dimensional positive semidefinite matrix $\sigma$ with trace $1$, an $(\varepsilon, \sigma, D_1,D_2)$-representation of $G$ is a function $f: G \to U(\mathbb{C}^d)$ satisfying:
\[
\mathbb{E}_{x \sim D_1, y \sim D_2} \Re \left( \langle f(xy), f(x)f(y) \rangle_\sigma \right) \geq 1 - \varepsilon.
\]
\end{definition}

We now state a stability theorem due to Gowers and Hatami, which allows us to conclude that approximate representations of the Pauli group are close to the Pauli group itself, up to isometry.
We state the theorem specifically for the Weyl-Heisenberg group and for approximate representations satisfying \( f(-x) = -f(x) \), \( f(1) = I \). This is because we not only want the approximate representation to be close to some representation of the Weyl-Heisenberg group, but we also want it to correspond to the canonical representation \( \pm X(a)Z(b) \). The conditions \( f(-x) = -f(x) \) and \( f(1) = I \) serve to uniquely identify this representation. The work of~\cite{mnz24} approaches this issue differently: they first obtain a general representation of the Weyl-Heisenberg group and then show how to isolate the canonical representation \( \pm X(a)Z(b) \). They also did not state the result explicitly as a theorem. For these reasons, even though the proof is well-known, we include our own version, that emphasizes the perspective of representation theory.

\begin{theorem}[Gowers-Hatami stability theorem, \cite{gh17}]
Let $\mathcal{P}_n$ denote the $n$-qubit Weyl–Heisenberg group, and let $f: \mathcal{P}_n \to U(\mathbb{C}^d)$ be an $(\varepsilon, \sigma, \text{Uniform}(\mathcal{P}_n), \mu)$-representation of 
$\mathcal{P}_n$ satisfying:
\[
f(-x) = -f(x), \quad f(1) = I.
\]
Then there exists $d' \geq d$, an isometry $V: \mathbb{C}^d \to \mathbb{C}^{d'}$, and a  representation $g: \mathcal{P}_n \to U(\mathbb{C}^{d'})$ such that $g(x) = \rho(x) \otimes I$, and
\[
\mathbb{E}_{x \sim \mu} \|f(x) - V^\dagger g(x) V\|_\sigma^2 \leq \Theta(\varepsilon),
\]
where $\rho(\pm X(a)Z(b)) = \pm X(a)Z(b)$. 
\end{theorem}

\begin{proof}
Note that $N=|{\cal P}_n|= 2\cdot 4^n $, the only irreducible representation that is not one-dimensional $\rho(\pm X(a)Z(b)) = \pm X(a)Z(b) $ has dimension $d_0=2^n$. 

Define $V$ as follows: for $u\in \C^d$,
\begin{align*}
V(u) = \sqrt{\frac{N}{2d_0}} \cdot \E_{x\in {\cal P}_n} f(x) u \otimes \rho(x) .
\end{align*}

$V$ is an isometry: for any $u, v\in \C^d $,
\begin{align*}
u^\dagger  V^\dagger  V  v
=&~\tr( (Vu)^\dagger V v)  \\
=&~ \frac{N}{2d_0} \cdot \tr(\E_{x,y\in {\cal P}_n} u^\dagger f(x)^\dagger f(y) v \otimes \rho(x)^\dagger  \rho(y) )\\
=&~ \frac{N}{2d_0} \cdot (\E_{x,y\in {\cal P}_n} u^\dagger f(x)^\dagger f(y) v \cdot \tr(\rho(x^{-1} y) ) )\\
=&~ \frac{1}{2d_0} (\E_{x\in {\cal P}_n} u^\dagger f(x)^\dagger f(x) v \cdot \tr(\rho(I) ) + \E_{x\in {\cal P}_n} u^\dagger f(x)^\dagger f(-x) v \cdot \tr(\rho(-I) )) \\
=&~ u^\dagger  v\;. 
\end{align*}

Let, for $v\in \C^{d}$, $O\in \C^{d_0\times d_0}$,
\begin{align*}
(g(x))(v\otimes O) = v \otimes O \rho(x^{-1})\;. 
\end{align*}

Note that $g(-I) = -I$. Among all the irreducible representations of $\mathcal{P}_n$, the only one that maps $-I$ to $-I$ is $\rho$.
Then, for any $u, v\in \C^d $,
\begin{align*}
u^*V^*g(x)V v 
=&~\tr((Vu)^*g(x)Vu) \\
=&~ \frac{N}{2d_0} \cdot \tr(\E_{z,y\in {\cal P}_n} u^\dagger f(z)^\dagger f(y) v \otimes \rho(z)^\dagger  \rho(y) \rho(x^{-1}))\\
=&~ \frac{N}{2d_0} \cdot \E_{z,y\in {\cal P}_n} u^\dagger f(z)^\dagger f(y) v \cdot \tr( \rho(z)^\dagger  \rho(y) \rho(x^{-1}))\\
=&~ \frac{1}{2d_0} \cdot (\E_{z\in {\cal P}_n} u^\dagger f(z)^\dagger f(zx) v \cdot \tr(I)+\E_{z\in {\cal P}_n} u^\dagger f(z)^\dagger f(-zx) v \cdot \tr( -I))\\
=&~ \E_{z\in {\cal P}_n} u^\dagger f(z)^\dagger f(zx) v\; .
\end{align*}

Therefore, 
\begin{align*}
\E_{x\sim \mu} \Re( \langle f(x), V^* g(x) V\rangle_\sigma ) =&~ \E_{x\sim \mu,z\in_R {\cal P}_n} \Re(\tr(f(x)^\dagger f(z)^\dagger f(zx) \sigma))\\
=&~ \E_{x\sim \mu,z\in_R {\cal P}_n} \Re(\tr(f(zx)^\dagger f(z) f(x) \sigma))\\
 =&~ \E_{x\sim \mu,z\in_R {\cal P}_n} \Re(\langle f(z)^\dagger f(zx), f(x)\rangle_\sigma)\\
=&~ \E_{x\sim \mu,z\in_R {\cal P}_n} \Re(\langle  f(zx), f(z)f(x)\rangle_\sigma)\\
\geq&~ 1- \eps \;.
\end{align*}
As a result,
\begin{align*}
&~\E_{x\sim \mu} \|f(x)-V^\dagger g(x) V\|_\sigma^2 \\
= &~\E_{x\sim \mu} \tr((f(x)^\dagger-V^\dagger g(x)^\dagger V)(f(x)-V^\dagger g(x) V)\sigma) \\
= &~2 - 2 \Re(\E_{x\sim \mu} \tr(f(x)^\dagger\cdot V^\dagger g(x) V)\sigma) \\
\leq &~ \Theta (\eps)\;.
\end{align*}
\end{proof}
\begin{remark}
The original version of this theorem generalizes the Fourier transform to matrix-valued functions and non-Abelian groups. Notably, we proved that for any approximate representation of the Weyl-Heisenberg group, the only non-one-dimensional irreducible representation always contributes a significant component. 
\end{remark}
\begin{remark}
We follow the approach from \cite{vid20}, and the idea of switching distributions was proposed in \cite{mnz24}. However, it is possible to modify the energy consistency test to directly verify the correctness of a single observable such as $X(a)$, $Y(b)$, or $Z(c)$. By using self-correction, we can enforce the correctness for each individual $a$, $b$, and $c$, thereby avoiding the need to switch the distribution.
\end{remark}

The following theorem builds on the above and provides a central technical tool in our many-qubits tests.

\begin{theorem}[Robustness of the many qubits tests]
\label{thm:ourqubits}
Let $n, d \in \mathbb{Z}_+$, $\delta \in (0, 1)$, and $\sigma \in \mathbb{C}^d \otimes \mathbb{C}^d$. 
Let $\mu_1$ be a distribution over ${0,1}^n$. Sample $(b, a) \sim \mu_1$, then sample $r$ uniformly at random from ${0,1}^n$. Define $\mu_2$ as the distribution of $(b + r, a)$.
Let $\rho(\pm X(a)Z(b)) = \pm X(a)Z(b)$. 
Let $f_X, f_Z: \{0,1\}^n \to U(\mathbb{C}^d)$ such that:
\begin{itemize}
\item $f_X(a)^2 = f_Z(b)^2 = I$, $f_X(a)^\dagger = f_X(a)$, $f_Z(b)^\dagger = f_Z(b)$ for all $a,b\in \{0,1\}^n$,
\item $f_X(a)f_X(b) = f_X(a + b)$, $f_Z(a)f_Z(b) = f_Z(a + b)$, for all $a,b\in \{0,1\}^n$,
\item $\mathbb{E}_{(b,a) \sim \mu_1} \|f_Z(b)f_X(a) - (-1)^{a \cdot b} f_X(a)f_Z(b)\|_\sigma^2 \leq \delta$,
\item $\mathbb{E}_{(b,a) \sim \mu_2} \|f_Z(b)f_X(a) - (-1)^{a \cdot b} f_X(a)f_Z(b)\|_\sigma^2 \leq \delta$,
\end{itemize}
Then there exists $d' \geq d$, an isometry $V: \mathbb{C}^d \to \mathbb{C}^{d'}$, and a representation $g: \mathcal{P}_n \to U(\mathbb{C}^{d'})$ such that $g(x) = \rho(x) \otimes I$ and
\[
\mathbb{E}_{(b,a) \sim \mu_1} \|f_X(a)f_Z(b) - V^\dagger g(X(a)Z(b)) V\|_\sigma^2 \leq \Theta(\delta)\;.
\]
\end{theorem}

\begin{proof}
Using that any element of ${\cal P}_n$ has a unique representative of the form 
$\pm X(a)Z(b)$ for $a,b\in \{0,1\}^n$,
we define $f(\pm X(a)Z(b)) = \pm f_X(a)f_Z(b) $. Next we need to verify that $f$ is an approximate representation.
Note that
\begin{align*}
&~f_X(a_0)f_Z(b_0)f_X(a_1)f_Z(b_1) - (-1)^{b_0\cdot a_1}f_X(a_0+a_1)f_Z(b_0+b_1) \\
= &~ f_X(a_0)f_Z(b_0)f_X(a_1)f_Z(b_1) - (-1)^{b_1\cdot a_1} f_X(a_0)f_Z(b_0)f_Z(b_1)f_X(a_1) \\
&~\quad+ (-1)^{b_1\cdot a_1}f_X(a_0)f_Z(b_0)f_Z(b_1)f_X(a_1) - (-1)^{b_1\cdot a_1}f_X(a_0)f_Z(b_0+b_1)f_X(a_1) \\
&~\quad+ (-1)^{b_1\cdot a_1}f_X(a_0)f_Z(b_0+b_1)f_X(a_1) - (-1)^{b_0\cdot a_1}f_X(a_0)f_X(a_1)f_Z(b_0+b_1) \\
&~\quad+ (-1)^{b_0\cdot a_1}f_X(a_0)f_X(a_1)f_Z(b_0+b_1) - (-1)^{b_0\cdot a_1}f_X(a_0+a_1)f_Z(b_0+b_1) \;.
\end{align*}
Therefore 
\begin{align*}
\E_{(b_0,a_0)\sim Uniform({\cal P}_n),(b_1,a_1)\sim \mu_1} \|f_X(a_0)f_Z(b_0)f_X(a_1)f_Z(b_1) - (-1)^{b_0\cdot a_1}f_X(a_0+a_1)f_Z(b_0+b_1)\|_\rho^2 \leq \Theta(\delta)\;.
\end{align*}

\end{proof}

%% file: formulation2.tex
\section{Formulation of Quantum Interactive Oracle Proofs}
\label{sec:formulation}

This section formalizes the model of Quantum Interactive Oracle Proofs (QIOPs), a quantum analogue of classical IOPs with limited quantum access by the verifier. We first specify basic notation. We then introduce fundamental notation and primitives such as the notion of a query and of a quantum channel with queries. We then define what constitutes a qIOP system, along with variations that capture different resource constraints.

\subsection{Notation}

We use sans serif font \(\mathsf{A}\), \(\mathsf{X}\), etc.\ to denote quantum registers and calligraphic font \(\mathcal{A}\), \(\mathcal{X}\), etc., to denote the associated complex Hilbert space.  We let $L(\mathcal{H})$ denote the linear operators on $\mathcal{H}$, and $C(\mathcal{H})$ the set of quantum channels on $\mathcal{H}$, i.e.\ completely positive and trace-preserving linear maps on $L(\mathcal{H})$.

We use \(d(\mathsf{A})\) to denote the dimension of a register \(\mathsf{A}\) and call a register \emph{trivial} if \(d(\mathsf{A})=1\). If \(\mathsf{A}\) is a register of dimension \(d_A\), we define a computational basis measurement channel \(\mathbb{M}_{\mathsf{A}}\in C(\mathcal{A})\) on it by setting:
\[
\mathbb{M}_{\mathsf{A}} = \sum_{i=1}^{d_A} \mathbb{M}_{\mathsf{A}}^i, \quad \mathbb{M}_{\mathsf{A}}^i(\rho_{\mathsf{A}}) = \ketbra{i}{i} \rho_{\mathsf{A}} \ketbra{i}{i}.
\]
This channel models standard basis measurements for quantum registers.

\subsection{Quantum channels with queries}

The notion of query channel is the core primitive which we use for modeling quantum access in qIOPs. 

\begin{definition}[Query]
\label{def:q}
Let $\mathsf{W} = (\mathsf{W}_1,\ldots,\mathsf{W}_n)$,  $\mathsf{R} = (\mathsf{R}_1,\ldots,\mathsf{R}_m)$, and $\mathsf{R}'$ be qubit registers.  
Let $k \in [n]$ and $\{s_1,\ldots, s_k\} \subseteq [n]$.  
Let $\mathsf{S} = (\mathsf{W}_{s_1},\ldots,\mathsf{W}_{s_k})$ be a subregister of $\mathsf{W}$.  
Then, a \emph{query channel on $\mathsf{S}$ with input register $\mathsf{R}$ and output register $\mathsf{R}'$} is a quantum channel that acts only on the registers $\mathsf{S}$ and $\mathsf{R}$, and maps $\mathsf{S}$ to $\mathsf{S}$ and $\mathsf{R}$ to $\mathsf{R}'$.
Moreover, the \emph{query set} of $\Phi$ is defined as $\{s_1,\ldots, s_k\}$. 
\end{definition}

A quantum channel with query consists of a single invocation of the query operation.  
Specifically, the channel first selects (through measurement) a query set $\{s_1,\ldots, s_k\}$, and then applies a query channel of registers $\mathsf{S}$ and $\mathsf{R}$.  
This structure models computation with oracle-style access, where the query simulates a single round of interaction with an oracle.

\begin{definition}[Quantum channel with queries]
\label{def:c_q}
Let $k$, $\mathsf{W}$, $\mathsf{R}$, and $\mathsf{R}'$ be as in Definition~\ref{def:q}.  
Let $\mathsf{I} = (\mathsf{R}_{1}, \ldots, \mathsf{R}_{k \cdot \lceil \log(n) \rceil})$ be a subregister of $\mathsf{R}$.  
We refer to $\mathsf{W}$ as the \emph{witness register}, $\mathsf{R}$ as the \emph{reference register}, and $\mathsf{R}'$ as the \emph{post-processed reference register}.
Let $\mathsf{S}_s$ denote the subregister $(\mathsf{W}_{s_1}, \ldots, \mathsf{W}_{s_k})$ indexed by $s = (s_1, \ldots, s_k) \in [n]^k$.  
For each $s$, let $\Phi_{\mathsf{R} \mathsf{S}_s, s}$ be a query channel of registers $\mathsf{S}_s$ with input register $\mathsf{R}$ and output register $\mathsf{R}'$ as in Definition~\ref{def:q}.

A quantum channel $\Lambda$ is called a \emph{quantum channel with queries} if it decomposes as:
\[
\Lambda = 
( \sum_s \Phi_{\mathsf{R} \mathsf{S}_s, s} \otimes \mathbb{M}_{\mathsf{I}}^s ) \circ \Psi,
\]
where $\Psi \in C(\mathcal{R})$ is a quantum channel. 
\end{definition}

\begin{definition}
    For a channel $\Psi$, we denote by $t(\Psi) $ its time complexity. This is the smallest number of elementary gates used to implement $\Psi$.\footnote{For this, we assume that some finite universal gate set, made of single and two-qubit gates, has been fixed.} We often abuse notation and use $t(\Psi) $ to denote an upper bound on the time complexity of channel $\Psi$.
\end{definition}

\paragraph{Complexity parameters.} Given a quantum channel with queries, we associate with it the following parameters:
\begin{itemize}
\item The \textbf{proof length} is $l(\Lambda) = n$.
\item The \textbf{query complexity} is $q(\Lambda) = k$.
\item The \textbf{time complexity} is $t(\Lambda)$.
\end{itemize}

The following definition formalizes the notion of a quantum algorithm with queries. Specifically, it models such an algorithm as being specified by a classical polynomial-time algorithm that outputs a quantum channel with queries. 
\begin{definition}[Quantum algorithm with queries]\label{def:qwq}
An algorithm $\texttt{A}$ is said to be a quantum algorithm with queries if $\texttt{A}$ is a classical polynomial-time algorithm that, given an input $x \in \{0,1\}^*$, outputs a description of a quantum channel with queries.
Let $\mathbb{A}_{x}$ denote the quantum channel produced by $\texttt{A}(x)$. 

The complexity of $\texttt{A}$ is characterized by three parameters, proof length $l(n)$, query complexity $q(n)$, and time complexity $t(n)$, each defined as an upper bound on the corresponding resource used by the channel $\mathbb{A}_x$ for any input $x$ of length $n$.

Moreover, 
we use $\texttt A(x,\rho_{\mathsf{W}})$ to denote $ \mathbb A_x(\rho_{\mathsf{W}}\otimes {\ketbra{0}{0}_{\mathsf{R}}})$.
We use $\texttt A(x)$ to denote 
$ \mathbb A_x({\ketbra{0}{0}_{\mathsf{W}}}\otimes {\ketbra{0}{0}_{\mathsf{R}}})$. 
Finally, we use the witness register of $ \mathbb A_x$ to denote the witness register of $\texttt{A}(x) $. 
\end{definition}

\subsection{Quantum interactive oracle proof systems}

Before introducing quantum interactive oracle proof systems in full generality, we first use the notation introduced in the previous subsection to formulate the well-known quantum PCP conjecture, which will serve as a guiding example to help clarify the definitions.

\begin{definition}[Quantum probabilistically checkable proof systems]\label{def:qpcp_sys}
A quantum probabilistically checkable proof system for a promise problem $(S_{yes}, S_{no})$ is
a quantum algorithm with queries $\texttt V$ such that:
\begin{enumerate}
\item \textnormal{Efficient:} $\texttt V$ has polynomial time complexity.  
\item \textnormal{Completeness:} For any $x\in S_{yes}$, there exists a quantum state $\rho$ such that $\Pr[\texttt  V(x)=1] \geq \frac{2}{3}$.   
\item \textnormal{Soundness:} For any $x\in S_{no} $ and any quantum state $\rho$, $\Pr[\texttt  V(x)=1] \leq \frac{1}{3}$.   
\end{enumerate}

\end{definition}

\begin{definition}[Quantum probabilistically checkable proofs]
Let $l,q:\Z_+ \rightarrow \Z_+$. Define $\mathcal{QPCP}(l, q)$ as the set of promise problem $(S_{yes}, S_{no})$ that have a quantum probabilistically checkable proof system with proof length $l$ and query complexity $q$.
\end{definition}

\begin{remark}
Our definition matches the standard (non-adaptive definition) of quantum PCPs~\cite{aav13,gjmz23,buhrman2024quantum}.
\end{remark}

\begin{conjecture}[qPCP conjecture]
${\QMA}\subseteq \mathcal{QPCP}(\poly, O(1))$
\end{conjecture}

We now introduce the central definition of a quantum interactive oracle proof (qIOP) system, which serves as a quantum analogue of classical interactive oracle proof systems.

A qIOP system consists of a sequence of algorithms $\texttt{V}_1, \ldots, \texttt{V}_n$, where each $\texttt{V}_i$ is a quantum algorithm with queries. 
The verifier and prover begin with initial state $\ket{0}$. 
During the interaction, the verifier is allowed to read from and modify the prover's message registers as a way of responding to the prover's messages.
The verifier is also allowed to retain part of the prover's message register and return it in later rounds of the interaction.
While the verifier operates under several complexity constraints, the malicious prover is assumed to be computationally unbounded. Soundness and completeness are defined analogously to classical interactive proof systems.
Moreover, the verifier is always allowed to send classical messages to the prover.

We first review the definition of a quantum interactive proof, and then proceed to our definition of a quantum interactive oracle proof. See Figure~\ref{fig:qip} for an illustration of the $3$-message case. 

\begin{figure}
\begin{tikzpicture}[
  font=\small,
  lane/.style={very thick},
  op/.style={
    draw, rounded corners=2pt, thick, fill=white,
    minimum width=14mm, minimum height=7mm,
    align=center
  },
  msg/.style={-Latex, thick},
  lab/.style={midway, fill=white, inner sep=1pt},
  note/.style={font=\footnotesize, align=left}
]

\def\xL{0}
\def\xR{13}
\def\yP{1.6}
\def\yV{0.0}

\draw[lane] (\xL,\yP) -- (\xR,\yP);
\draw[lane] (\xL,\yV) -- (\xR,\yV);

\node[anchor=east] at (\xL-0.4,\yP) {Prover};
\node[anchor=east] at (\xL-0.4,\yV) {Verifier};

\node[note, anchor=west] at (\xL+4.3,\yP+0.25) {private: $\mathsf{P}$};
\node[note, anchor=west] at (\xL+0.1,\yV-0.35) {private: $\mathsf{R}$};

\def\tA{2.2}
\def\tB{5.3}
\def\tC{8.4}
\def\tD{11.5}

\node[op] (P0) at (\tA,\yP) {$\mathbb{P}_0$};
\node[op] (V1) at (\tB,\yV) {$\mathbb{V}_{x,1}$};
\node[op] (P1) at (\tC,\yP) {$\mathbb{P}_1$};
\node[op] (V2) at (\tD,\yV) {$\mathbb{V}_{x,2}$};

\draw[msg] (P0.east) -- node[lab] {$\mathsf{W}_0$} (V1.west);
\draw[msg] (V1.east) -- node[lab] {$(\mathsf{M}_1,\mathsf{J}_1)$} (P1.west);
\draw[msg] (P1.east) -- node[lab] {$\mathsf{W}_1$} (V2.west);

\node[note, anchor=west] at (\tD+0.9,\yV-0.35) {decision in $\mathsf{J}_2$};

\end{tikzpicture}
\caption{An example of a $m=3$-message quantum interactive proof.}\label{fig:qip}
\end{figure}
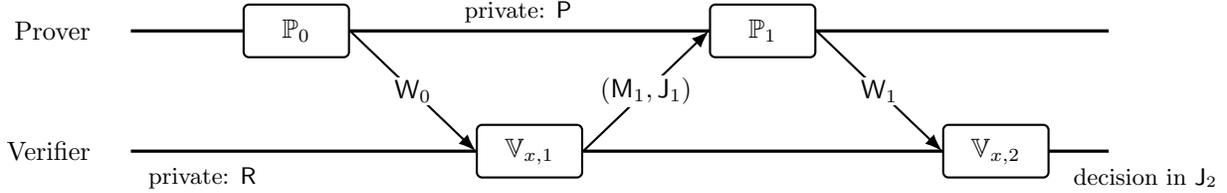

\begin{definition}[Quantum interactive proof systems]\label{def:qip}
Let $m \in \mathbb{Z}_+$ be the total number of messages, and let $h = \lfloor m/2 + 1 \rfloor$. An $m$-message quantum interactive oracle proof system is defined by a tuple of quantum algorithms with queries $\texttt{V} = (\texttt{V}_1, \texttt{V}_2, \ldots, \texttt{V}_h)$.
The tuple $\mathbb{V}_x = (\mathbb{V}_{x,1}, \ldots, \mathbb{V}_{x,h})$ represents the quantum channels with queries produced by the algorithm $\texttt{V}_i$ on input $x$. Each quantum channel $\mathbb{V}_{x,i}$, for $i\in[h]$, acts on a (quantum) witness register $(\mathsf{W}_{i-1}, \mathsf{W}'_{i-1})$ and (quantum) reference register $\mathsf{R}$, and it returns a message register  $(\mathsf{J}_i,\mathsf{M}_i)$ where  $\mathsf{J}_i$ is classical and  $\mathsf{M}_i$ quantum, together with private registers $( \mathsf{W}'_{i}, \mathsf{R})$.\footnote{For the present definition the registers $\mathsf{W}'_{i}$ are not needed and can be assumed to be trivial without loss of generality. We include them for future consistency.} 
Summarizing, the verifier channel $\mathbb{V}_{x,i}$ acts on the registers as follows:
\begin{align*}
\big((\mathsf{W}_{i-1}, \mathsf{W}'_{i-1}), \mathsf{R} \big)
\longrightarrow 
\big((\mathsf{M}_{i}, \mathsf{J}_{    i}), (\mathsf{W}'_{i},\mathsf{R}) \big)\;.
\end{align*}
If $m$ is odd, define the index set $I = \{0, \ldots, h-1\}$; if $m$ is even, let $I = \{1, \ldots, h-1\}$.  
A prover strategy is specified by a collection of quantum channels $\{\mathbb{P}_i\}_{i \in I}$ and a quantum register $\mathsf{P}$, where each $\mathbb{P}_i$ acts on the registers as follows:
\[
((\mathsf{M}_i, \mathsf{J}_i), \mathsf{P}) \longrightarrow (\mathsf{W}_i, \mathsf{P}).
\]
Let the first (classical) bit of the register $\mathsf{J}_{    h}$ in the output of $\mathbb{V}_{x,h}$ represent the verifier's final decision to accept or reject.
Let $\omega(\mathbb{V}_x)$ denote the value of the interactive game defined by $\mathbb{V}_x$. That is, $\omega(\mathbb{V}_x)$ is the maximum acceptance probability of $\mathbb{V}_x$ over all possible prover strategies $\{\mathbb{P}_i\}_{i \in I}$.

The following properties must be satisfied:
\begin{itemize}
\item \textnormal{Consistent:} 
Unless otherwise specified, we initialize $\mathsf{W}_i$, $\mathsf{W}'_i$, $ \mathsf{R}$, $ \mathsf{M}_i$ to $\ket{0}$, whereas $\mathsf{P}$ may be initialized in an arbitrary quantum state.
The register $\mathsf{J}_i$ is initialized to the all-zero string.
\item \textnormal{Efficient:} Each verifier's algorithm $\texttt{V}_i$ must be implementable in polynomial time and have polynomial time complexity as defined in Definition \ref{def:qwq}. 
\end{itemize}
\end{definition}

In the next definition we present a general model of quantum interactive oracle proofs, where the verifier is allowed to generate and send an arbitrary quantum state to the prover, rather than being restricted to sending only registers from the prover's message. 
In this definition, we denote by $\mathsf{M}'_i$ the message register generated by the verifier in round $i$. The full message register of the verifier in round $i$, denoted by $\mathsf{M}_i$, includes both $\mathsf{M}'_i$ and some registers received from the prover in the current and previous rounds.

In a later section, we provide a relatively ``simple'' construction of a general quantum interactive oracle proof, illustrating that this definition may be too permissive to yield meaningful or interesting results. For this reason, we subsequently introduce a more restricted definition where the $\mathsf{M}'_i$ registers are required to be trivial.

\begin{definition}[General quantum interactive oracle proof systems]\label{def:gqiop}
Let $m \in \mathbb{Z}_+$. 
An $m$-message general quantum interactive oracle proof system consists of an $m$-message quantum interactive proof system $\texttt{V}$ satisfying the same structure as in Definition~\ref{def:qip}, but with the following modifications to the register structure.

As before, in the $i$th round, for $i\in[h]$ where  $h = \lfloor m/2 + 1 \rfloor$, the verifier receives the quantum message register $\mathsf{W}_{i-1}$ from the prover. 
In addition, the verifier has in their possession  a register $\mathsf{W}'_{i-1}$ that stores (parts of) previously received messages from the prover, and a register $\mathsf{R}$ that is their private (reference) register. 
The register  $\mathsf{W}'_{i-1}$ decomposes as 
\[\mathsf{W}'_{i-1}=(\mathsf{W}'_{0,i-1},\mathsf{W}'_{1,i-1}, \ldots, \mathsf{W}'_{i-2,i-1})\;,\] where $\mathsf{W}'_{j,i-1}$ denotes the part of the prover's $j$th message (i.e., $\mathsf{W}_{j}$) that has so far not been sent back to the prover. 

For each $\mathsf{W}'_{j,i-1}$, the verifier sends a subregister $\mathsf{N}_{j,i}$ to the prover in the $i$th round, and keeps the remaining part as $\mathsf{W}'_{j,i}$. 
In addition, the verifier produces a quantum message register $\mathsf{M}'_i$ and a classical message register $\mathsf{J}_i$ . The verifier sends back the total message register $(\mathsf{M}_i, \mathsf{J}_i)$ to the prover, where
\[\mathsf{M}_i =
(\mathsf{M}'_i,(\mathsf{N}_{0,i},\mathsf{N}_{1,i}, \ldots, \mathsf{N}_{i-1,i})).\]

Finally, we define the following resource bounds of the protocol:
\begin{itemize}
    \item \textbf{Proof length:} The verifier $V$ has proof length $l : \mathbb{Z}_+ \rightarrow \mathbb{Z}_+$ if
    \[
    l(n) \geq \sum_{i \in [n]} l_i(n),
    \]
    where $l_i(n)$ is the proof length of $\texttt{V}_i$.

    \item \textbf{Query complexity:} The verifier $V$ has query complexity $q : \mathbb{Z}_+ \rightarrow \mathbb{Z}_+$ if
    \[
    q(n) \geq \sum_{i \in [n]} q_i(n),
    \]
    where $q_i(n)$ is the query complexity of $\texttt{V}_i$.
\end{itemize}

\end{definition}
\begin{remark}
Because $\mathbb{V}_{x,i}$ has witness register $ \mathsf{W}_{i-1},\mathsf{W}'_{i-1}$, any quantum register touched by gates in  $\mathbb{V}_{x,i}$ is going to be counted in the query complexity.
\end{remark}

\begin{definition}[Restricted quantum interactive oracle proof systems]\label{def:qiop}
An $m$-message \emph{restricted quantum interactive oracle proof system} consists of an $m$-message general quantum interactive oracle proof system such that for each $i$,  $\mathsf{M}'_i$ is a trivial quantum register. 
\end{definition}

\begin{definition}[Strong quantum interactive oracle proof systems]\label{def:sqiop}
An $m$-message \emph{strong quantum interactive oracle proof system} consists of an $m$-message quantum interactive oracle proof system satisfying that (using the notation in Definition \ref{def:qip} and Definition \ref{def:qiop}) 
for each $i \in [h]$, $\mathsf{M}'_i$ is trivial (no new quantum message sent by the verifier) and $\mathsf{N}_{i-1,i}=\mathsf{W}_{i-1}$ (all prover quantum messages are sent back to him). As a consequence, $\mathsf{M}_i=\mathsf{W}_{i-1}$, $\mathsf{W}'_{i-1}$ is always trivial, and $\mathbb{V}_{x,i}$ has witness register $ \mathsf{W}_{i-1}$ only. 
\end{definition}

\begin{definition}[Quantum interactive oracle proof systems for a promise problem]\label{def:qiop_}
Let $m \in \mathbb{Z}_+$ be the total number of messages. 
An $m$-message restricted quantum interactive oracle proof system (resp. general quantum interactive oracle proof system, strong  quantum interactive oracle proof system) for a promise problem $(S_{\text{yes}}, S_{\text{no}})$  satisfies the following, where $\omega(\mathbb{V}_x)$ is defined as in Definition~\ref{def:qiop} (resp. Definition~\ref{def:gqiop}, Definition~\ref{def:sqiop}). 
\begin{itemize}
\item \textnormal{Completeness:} There exists a constant $c \in (0,1]$ such that for all $x \in S_{\text{yes}}$, the maximum acceptance probability $\omega(\mathbb{V}_x) \geq c$. That is, there exists an honest prover strategy under which the verifier accepts every yes-instance with probability at least $c$.

\item \textnormal{Soundness:} There exists a constant $s \in [0, c)$ such that for all $x \in S_{\text{no}}$, the maximum acceptance probability $\omega(\mathbb{V}_x) \leq s$. That is, no prover strategy can cause the verifier to accept a no-instance with probability greater than $s$.
\end{itemize}
\end{definition}

We define three classes of quantum interactive oracle proof systems.
The class $\mathcal{QIOP}(r, l, q)$ consists of all promise problems that admit an $r$-message restricted quantum interactive oracle proof (qIOP) system with proof length bounded by $l(n)$ and query complexity bounded by $q(n)$. Here, the query complexity captures the total number of qubits accessed in the prover's message by the verifier across all rounds.

\begin{definition}[Quantum interactive oracle proofs]
Let $r,l,q:\Z_+ \rightarrow \Z_+$. Let $\mathcal{QIOP}(r, l, q)$ (resp. $\mathcal{GQIOP}(r, l, q)$, $\mathcal{SQIOP}(r, l, q)$) be the set of promise problem $(S_{yes}, S_{no})$ that has an $r$-message restricted (resp. general, strong) quantum interactive oracle proof system with proof length $l$ and query complexity $q$.
\end{definition}

Each of these models assumes a polynomial-time verifier and allows for a quantum prover that is computationally unbounded. The number of interaction rounds is bounded by $r(n)$, which is always taken to be constant. These complexity classes help characterize how much quantum information the verifier may access or process.

%% file: teleportation_approach.tex
\section{Quantum IOPs with Shared EPR Pairs and Classical Communication}
\label{sec:qiopepr}

In our first construction,  the verifier prepares EPR pairs and sends one half of each to the prover, thereby establishing shared entanglement.

In this section, we proceed to prove the following theorem.

\begin{theorem}\label{thm:qpcp_tlp}
$\mathsf{QMA}\subseteq {\cal GQIOP}(O(1), \poly, O(1))$.     
\end{theorem}

\subsection{Commitment in the Interactive Oracle Proof Model}

In the interactive oracle proof setting, the verifier leverages the prover's computational power to evaluate certain functions of the input.
In our specific setting, in will be the case that the prover holds a string that the verifier cannot directly access, while the verifier holds a circuit unknown to the prover. The combination of this string and the circuit determines whether the input belongs to the language.
The trivial solution is problematic: if the verifier simply sends the circuit to the prover, the prover can cheat by evaluating it on a different input rather than the one initially given.
An interactive oracle commitment provides a solution by allowing the prover to commit to their input at the outset. Later, the verifier sends a circuit, and the commitment ensures that the prover evaluates it honestly on the committed input. 

In the following definition, the first requirement enforces an interactive oracle-style interaction, while the second requires the prover to be capable of committing to any string and performing any computation.
(Note that the following definition refers to the classical setting.)

\begin{definition}[Interactive Oracle Commitment]\label{def:ioc}
An interactive oracle commitment scheme is specified by a classical probabilistic polynomial-time interactive machine $\texttt{V}$ and a function $D:\F_2^*\rightarrow \F_2^*$ satisfying the following:
\begin{enumerate}
\item The verifier's query complexity is a constant. 
\item There exists an honest polynomial time prover $\texttt{P}$, such that for any $n$, any $s\in \F_2^n$, for any circuit $C$ of size  $\poly(n)$ with a single-bit output, any $\eps\in (0, 1)$, at the end of the interaction between $ \texttt{P}(s)$ and $\texttt{V}({C}, \eps)$, the verifier outputs ${C}({s})$ with probability $1$. Moreover, the communication complexity between $\texttt{P}$ and $\texttt{V}$ is $\poly(n)$.
\item \textnormal{Soundness:} For any polynomial size circuit ${C}$ with a single bit output, any $\eps\in (0, 1)$, and any prover $\texttt{P}$, let $m$ be the prover's first message, at the end of the interaction, 
\begin{align*}
\Pr[\text{$\texttt{V}({C}, \eps)$ outputs ${C}({D(m)})$ or rejects}] \geq 1-\eps\;.
\end{align*} 
\end{enumerate}
\end{definition}
\begin{remark}
For the honest prover, the first message $m$ is the commitment and satisfies $D(m)=s$. Moreover, $Im(D)=\F_2^*$. 
\end{remark}

We will construct an interactive oracle commitment (IOC) protocol using a PCP of Proximity (PCPP) for the \textsc{Circuit Value} problem and a family of error-correcting codes with polynomial rate, constant distance, and efficient encoding and decoding. 

First, we recall a fundamental result about the existence of error correction codes. These codes form the basis of the prover's commitment in our protocol.

\begin{definition}[Codeword tester for error-correcting code]
\label{def:ecc_tester}
Let $C \subseteq \mathbb{F}_2^n$ be a binary error-correcting code. A deterministic \emph{tester} for $C$ is a classical algorithm 
$\texttt{T} : \mathbb{F}_2^n \to \{0, 1\}$
such that
\[
\texttt{T}(x) = 
\begin{cases}
1 & \text{if } x \in C, \\
0 & \text{otherwise}.
\end{cases}
\]
That is, $\texttt{T}$ accepts a string if and only if it is a valid codeword in $C$.
\end{definition}

\begin{theorem}[Error correction codes, \cite{556667}]
\label{thm:ltc}
There exist families of linear error-correcting codes with polynomial rate, constant relative distance, and deterministic polynomial-time encoding and decoding algorithms, along with a deterministic polynomial-time tester for checking codeword membership.
\end{theorem}

We now describe the structure of the prover in the interactive oracle commitment protocol. The prover uses a codeword encoding of its message to commit to a value and later proves consistency of this value through interaction with the verifier.
\begin{definition}[Construction of the Prover in the Interactive Oracle Commitment Protocol]
\label{def:ioc_hp}
\hfill
\begin{itemize}
\item \textnormal{The prover's first step:} The prover sends the verifier a string $E(s)$, where $E$ is the encoding function of the code from Theorem~\ref{thm:ltc}.
\item \textnormal{The verifier's first step:} The verifier sends the prover a description of a Boolean circuit $C$.
\item \textnormal{The prover's second step:} The prover responds with a string $C(s) \circ w$, where $w$ is a witness such that the verifier in Theorem~\ref{thm:pcpp} accepts the oracle $E(s) \circ w$ when invoked on input $x = C(D(\cdot)) + C(s) + 1$.
\item \textnormal{The verifier's second step:} The verifier decides whether to accept, and if so, produces an output.
\end{itemize}
\end{definition}

Next, we describe the behavior of the verifier in the interactive oracle commitment protocol. The verifier checks the proximity of the prover's message to a codeword and verifies the correctness of the claimed output using a PCPP verifier.
\begin{definition}[Construction of the Verifier in the Interactive Oracle Commitment Protocol]
\label{def:ioc_verifier}
\hfill
\begin{itemize}
\item \textnormal{The prover's first step:} The prover sends the verifier a string $m$.
\item \textnormal{The verifier's first step:} The verifier sends the prover a description of a Boolean circuit $C$.
\item \textnormal{The prover's second step:} The prover sends the verifier a string $r \circ w$.
\item \textnormal{The verifier's second step:} The verifier accepts and outputs $r$ if and only if the verifier from Theorem~\ref{thm:pcpp}, when run a constant number of times, accepts each time on input
\[
x = \big(C(D(\cdot)) + r + 1\big) \wedge T(\cdot),
\]
with implicit input $y = m$ and proof $\pi = w$, where $D$ is the decoding circuit and $T$ is the  code tester  from Theorem~\ref{thm:ltc}
\end{itemize}
\end{definition}

The following lemma states that the protocol defined above satisfies the requirements of an interactive oracle commitment. In particular, the verifier satisfies completeness and soundness properties.
\begin{lemma}
The verifier described in Definition~\ref{def:ioc_verifier} satisfies the requirements of an interactive oracle commitment.
\end{lemma}

\begin{proof}
By Definition~\ref{def:ioc_verifier}, Theorem~\ref{thm:ltc}, and Theorem~\ref{thm:pcpp}, the verifier's query complexity is constant.

\textbf{Completeness:}  
Using Theorem~\ref{thm:ltc} and Theorem~\ref{thm:pcpp}, if the prover is honest and provides a valid encoding and witness, the verifier will accept with probability $1$. Moreover, the communication complexity between the prover and verifier is polynomial in the input size.

\textbf{Soundness:}  
Let $S = \{t ~|~ C(D(t)) = r \wedge T(t) = 1\}$.  
If $\text{dist}(m, S) > \delta$, then $\Pr[ \texttt{V}(C, \varepsilon) \text{ accepts}] \leq \varepsilon$,
and therefore,
\[
\Pr[ \texttt{V}(C, \varepsilon) \text{ either outputs } C(D(m)) \text{ or rejects}] \geq 1 - \varepsilon.
\]
If $\text{dist}(m, S) \leq \delta$,
then there exists $s$ such that $C(s) = r$ and $\text{dist}(m, E(s)) \leq \delta$. Hence $D(m) = s$ and $C(D(m)) = r$, and the verifier will either accept with correct output or reject. Thus,
\[
\Pr[\text{the verifier } \texttt{V}(C, \varepsilon) \text{ either outputs } C(D(m)) \text{ or rejects}] = 1.
\]
\end{proof}

\subsection{Quantum IOPs with shared EPR pairs}

\begin{theorem}[Restatement of Theorem~\ref{thm:qpcp_tlp}]
$\QMA\subseteq {\cal GQIOP}(3, \poly, O(1))$.
\end{theorem}

We begin by defining the verifier and prover used in the proof of Theorem~\ref{thm:qpcp_tlp}. Following this, we proceed to prove the theorem.

A naive approach would be for the prover to teleport $N$ copies of the witness to the verifier, sending along the corresponding one-time pads required for the teleportation. The verifier would then measure their half of the EPR pairs, apply the appropriate corrections using the received one-time pads, and decide whether to accept based on the result. Note that the verifier's half of the EPR pairs is stored in private registers, and thus does not contribute to the query complexity.

In the construction that follows, the prover and verifier follow a protocol inspired by this naive approach. However, instead of directly transmitting the one-time pads, the prover commits to them using a commitment scheme. The verifier then uses their ability to extract functions of the committed value to recover the necessary information for correcting the measurement outcomes.

At the start of the protocol, the verifier holds a register $\mathsf{R}'$ and the prover holds a register $\mathsf{R}$, such that the joint quantum state on $\mathsf{RR}'$ consists of shared EPR pairs.

\begin{definition}[Honest prover strategy for Theorem~\ref{thm:qpcp_tlp}]
\label{def:qpcp_tlp_hp}
Let $p, q$ be as in Theorem~\ref{thm:clifford_ham}, and let $s, N$ be as in Lemma~\ref{lem:c_ham_amp}.  
Let $H = \frac{1}{m} \sum_{i \in [m]} H_i \in \mathbb{C}^{2^n} \otimes \mathbb{C}^{2^n}$ be a $5$-local Clifford Hamiltonian with $m = \mathrm{poly}(n)$, such that either $\lambda_{\min}(H) \leq 2^{-p(n)}$ or $\lambda_{\min}(H) \geq 1/q(n)$.  

Let $\mathsf{R}\mathsf{R}'$ be quantum registers initialized to the maximally entangled state:
\[
\frac{1}{\sqrt{2^{nN}}}\sum_{i=0}^{2^{nN}-1} \ket{i}_{\mathsf{R}} \ket{i}_{\mathsf{R}'}\;.
\]

\begin{itemize}
\item The prover holds registers $\mathsf{W}_i = (\mathsf{W}_{i,1}, \mathsf{W}_{i,2}, \ldots, \mathsf{W}_{i,n})$ for $i \in [N]$. Let $\ket{\phi}$ be the eigenstate corresponding to the maximum eigenvalue of $1 - 2H$. The prover prepares $N$ copies of $\ket{\phi}$ and stores them in the registers $\mathsf{W}_i$.
  
\item For each $i \in [N]$ and $j \in [n]$, the prover measures the joint system $\mathsf{W}_{i,j} \mathsf{R}_{(i-1)*n + j}$ in the EPR basis as defined in Definition~\ref{def:epr_basis}, and stores the measurement result in $(s_{0, i, j}, s_{1, i, j})$.

\item Let 
\[
s = \left(s_{k, i, j}\right)_{(k, i, j) \in \{0, 1\} \times [N] \times [n]}.
\]
The prover then engages in the interactive oracle commitment protocol defined in Definition~\ref{def:ioc}, using $s$ as input.
\end{itemize}
\end{definition}

\begin{definition}[Verifier strategy for Theorem~\ref{thm:qpcp_tlp}]
\label{def:qpcp_tlp_v}
Under the same setup assumptions as in Definition~\ref{def:qpcp_tlp_hp}, the verifier performs the following actions:
\begin{itemize}
\item For each $i \in [N]$, the verifier samples $l_i \in [m]$ independently and uniformly at random.

\item For each $i \in [m]$, let $H_i = \left(C_i^\dagger (\ketbra{0}{0})^{\otimes 5} C_i\right) \otimes I$. For $j \in [5]$, define
\[
O_{i,j} = \left(C_i^\dagger \cdot (I^{\otimes j-1} \otimes Z \otimes I^{\otimes (5-j)}) \cdot C_i \right) \otimes I,
\]
and let
\[
O_{i,j} = (-1)^{d_{i,j}} X(a_{i,j}) Z(b_{i,j}) Y(c_{i,j})
\]
be its Pauli representation as in Lemma~\ref{lem:clifford_pauli}.

\item For each $i \in [N]$, the verifier measures the register
\[
\mathsf{R'}_i = (\mathsf{R'}_{(i-1)*n + 1}, \ldots, \mathsf{R'}_{(i-1)*n + n})
\]
using the observables $O_{l_i, j}$ for $j = 1$ to $5$, and stores the measurement outcomes in $r_{i, j}$.

\item The verifier constructs the Boolean circuit:
\[
C(x) = \bigoplus_{i \in [N]} \Big(1 \oplus \bigvee_{j \in [5]} \big(r_{l_i,j} \oplus a_{l_i,j} x_{0,i} \oplus b_{l_i,j} x_{1,i} \oplus c_{l_i,j} x_{1,i} \oplus c_{l_i,j} x_{0,i} \big) \Big),
\]
where $x = \left(x_{k, i, j}\right)_{k \in \{0, 1\}, i \in [N], j \in [n]}$.

\item The verifier runs the interactive oracle commitment protocol in Definition~\ref{def:ioc} with the circuit $C $ and $\varepsilon = 0.001$.

\item The verifier accepts if and only if the interactive oracle commitment accepts and its output is $0$.
\end{itemize}
\end{definition}

\begin{remark}
Note that in our construction, after the verifier sends the half EPR pairs in the first round, all remaining communication is classical. If we instead assume that the verifier and the prover share EPR pairs at the beginning, then the protocol becomes fully classical.
\end{remark}

\begin{proof}[Proof of Theorem \ref{thm:qpcp_tlp}]
Let $ \texttt V$ be the verifier in Definition \ref{def:qpcp_tlp_v}. We will show that $\texttt V $
is $3$-message non-adaptive quantum classical interactive oracle proof system with $Nn$-shared EPR pairs for $\mathcal{LCH}(5,2^{-p(n)},1/q(n))$ and thus $$ \mathcal{LCH}(5,2^{-p(n)},1/q(n))\in \QIOPEPR(3, \poly, O(1), \poly).$$ 
By Theorem \ref{thm:clifford_ham}, $\mathcal{LCH}(5,2^{-p(n)},1/q(n))$ is $ \QMA$-complete. 
Therefore, $$\QMA\subseteq \QIOPEPR(3, \poly, O(1), \poly)\;.$$

By Definition \ref{def:qpcp_tlp_hp}, Definition \ref{def:qpcp_tlp_v}, and Definition \ref{def:qiop}, $\texttt V $ is consistent, has polynomial time complexity, is efficient, has proof length $ \poly(n) $, and has query complexity $ O(1) $. Moreover, the completeness and soundness properties are as follows.

\medskip

{\bf Completeness:} 
Let $r'_i$ be the random variable that has the distribution of measuring $\ket{\phi}$ under $1-2H_{l_i}$. By quantum teleportation, the state of the verifier after receiving $s_0, s_1$ is $X(s_1)Z(s_0) \ket{\phi}^{\otimes N}$. 
By Lemma \ref{lem:same_dist:tel_qiop}, Definition \ref{def:qpcp_tlp_hp}, and Definition \ref{def:qpcp_tlp_v}, given $l_1,l_2,\cdots, l_N, s_0, s_1$, 
\begin{align}
\Pr\big[C((s_0,s_1)) =0 ~|~ l_1,l_2,\cdots, l_N, s_0, s_1\big]
 = \Pr\Big[\bigoplus_{i\in [N]} r'_i =0\Big]\;.
\end{align}
Thus, by Lemma \ref{lem:xor_r:tel_qiop}, we can conclude that
\begin{align*}
\omega(V) 
= &~ \Pr[\text{Verifier accepts}]\\
 = &~ \sum_{l_1,l_2,\cdots, l_N\in [m]} \Pr[l_1,l_2,\cdots, l_N]\cdot  \Pr[\text{Verifier accepts}~|~ l_1,l_2,\cdots, l_N]\\
 = &~ \sum_{l_1,l_2,\cdots, l_N\in [m], s_0, s_1 \in \F_2^{Nn}} 
 \frac{1}{m^N} 
\Pr[s_0,s_1] \cdot \Pr[\text{Verifier accepts}~|~ l_1,l_2,\cdots, l_N, s_0, s_1]\\
 = &~ \sum_{l_1,l_2,\cdots, l_N\in [m], s_0, s_1 \in \F_2^{Nn}} 
 \frac{1}{m^N} 
\cdot (\frac{1}{2^{Nn}} )^2 \cdot \Pr[C((s_0,s_1)) =0 ~|~ l_1,l_2,\cdots, l_N, s_0, s_1]\\
 = &~ \sum_{l_1,l_2,\cdots, l_N\in [m]} 
 \frac{1}{m^N} \cdot \Pr[\bigoplus_{i\in [N]} r'_i =0]\\ 
 = &~ \sum_{l_1,l_2,\cdots, l_N\in [m]} 
 \frac{1}{m^N} \cdot \Pr[\text{The outcome of measuring $ \ket{\phi}^{\otimes N}$ under $\bigotimes_{i\in [N]} (1-2 H_{l_i})$  is $0$}]\\ 
  = &~ \sum_{l_1,l_2,\cdots, l_N\in [m]} 
 \frac{1}{m^N} \cdot 
 \bra{\phi}^{\otimes N} \frac{I+\bigotimes_{i\in [N]} (1-2 H_{l_i})}{2}  \ket{\phi}^{\otimes N} 
 \\ 
 = &~  
 \bra{\phi}^{\otimes N} \frac{I+(1-2H)^{\otimes N}}{2}  \ket{\phi}^{\otimes N} 
 \\ 
 = &~  \frac{1}{2}+\frac{(\bra{\phi}(1-2H)\ket{\phi})^N}{2}
 \\ 
 \geq &~ 1- \Theta(2^{-s(n)}) \;,
\end{align*}
where the fourth step follows from the completeness of the interactive oracle commitment. 

\medskip

{\bf Soundness:} 
Let the first message sent by the prover in the interactive oracle commitment be $M$, and let $(s_0,s_1)$ be such that $D(M)=(s_0,s_1)$. Without loss of generality, the whole quantum state after the prover sends the message is 
$$ \frac{1}{\sqrt{2^{nN}}}\sum_{i=0 }^{2^{nN}-1} (U P_{M}\ket{0}_{\mathsf{W}}\ket{i}_{\mathsf{R}})\ket{i}_{\mathsf{R}'},$$ where $\mathsf{W}$ is arbitrary register, $U$ is an arbitrary unitary acting on $ \mathsf{W}\mathsf{R} $, and $\{P_{M}\}_{M\in \F_2^{h}}$ is and arbitrary orthogonal projective measurement acting on $ \mathsf{W}\mathsf{R} $, where $h=|E(1^{2Nn})|$ is the output length of the encoding circuit of the  codes in Theorem \ref{thm:ltc}.

For $i\in [N]$, define an  orthogonal projection $$ Q_i^{(1)} = C_{l_i}^\dagger \bigotimes_{j\in [5]} \frac{I+(-1)^{a_{l_i,j}s_{0, i}+ b_{l_i,j}s_{1, i}+ c_{l_i,j} (s_{0, i} + s_{1, i})} Z}{2}  C_{l_i}\;,\qquad  Q_i^{(0)}= I - Q_i^{(1)}\;.$$ 
Note that 
\begin{align*}
    X(s_1)Z(s_0) Q_i^{(j)} = \Big(\frac{1}{2}+(-1)^{j}\Big(\frac{1}{2}-H_{l_i}\Big)\Big) X(s_1)Z(s_0)\;.
\end{align*}
The probability of the verifier accepting can be bounded as follows:
\begin{align*}
\omega(V) 
= &~ \Pr[\text{Verifier accepts}]\\
 = &~ \sum_{l_1,l_2,\cdots, l_N\in [m]} \Pr[l_1,l_2,\cdots, l_N]\cdot  \Pr[\text{Verifier accepts}~|~ l_1,l_2,\cdots, l_N]\\
 = &~ \sum_{l_1,l_2,\cdots, l_N\in [m], M \in \F_2^{h}} 
 \frac{1}{m^N} 
\Pr[M] \cdot \Pr[\text{Verifier accepts}~|~ l_1,l_2,\cdots, l_N, M] \;.
\end{align*}
By the soundness of the interactive oracle commitment, 
\begin{align*}
\Pr[\text{the verifier in interactive oracle commitment either outputs ${C}({s_0,s_1})$ or reject}] \geq 1-\eps.
\end{align*}
Thus,
\begin{align*}
 \Pr[&\text{Verifier accepts}~|~ l_1,l_2,\cdots, l_N, M] \\
=&~ \Pr[\text{Verifier in interactive oracle commitment accepts and outputs $0$}~|~ l_1,l_2,\cdots, l_N, M] \\
\leq &~ \max\big\{\Pr[{C}({s_0,s_1})=0~|~ l_1,l_2,\cdots, l_N, M]\ ,\; \eps\big\} .
\end{align*}
Let $$
r_i=1 \oplus \bigvee_{j\in [5]} (r_{l_i,j} \oplus a_{l_i,j} s_{0, i} \oplus b_{l_i,j} s_{1,i} \oplus c_{l_i,j} s_{1,i} \oplus  c_{l_i,j} s_{0, i} )\;.
$$
Then,
\begin{align*}
\Pr[{C}({s_0,s_1})=0~|~ l_1,l_2,\cdots, l_N, M]
= &~ \sum_{\oplus_{i\in[N]} r_i = 0}  
\Pr[r_1,r_2,\cdots, r_N~|~ l_1,l_2,\cdots, l_N, M] \;.
\end{align*}
Moreover, we have that
\begin{align*}
&~ \Pr[r_1,r_2,\cdots, r_N~|~ l_1,l_2,\cdots, l_N, M] \cdot \Pr[M]\\
 = &~  \frac{1}{2^{nN}} \Big\|\sum_{i=0 }^{2^{nN}-1} (U P_{M}\ket{0}_{\mathsf{W}}\ket{i}_{\mathsf{R}}) \prod_{j=1}^N Q_j^{(r_j)} \ket{i}_{\mathsf{R}'}\Big\|^2 \\ 
= &~ \frac{1}{2^{nN}} \Big\|\sum_{i=0 }^{2^{nN}-1} (U P_{M}\ket{0}_{\mathsf{W}}\ket{i}_{\mathsf{R}}) (X(s_1)Z(s_0))_{\mathsf{R}'} \prod_{j=1}^N Q_j^{(r_j)} \ket{i}_{\mathsf{R}'}\Big\|^2 \\
= &~ \frac{1}{2^{nN}}\Big\|\sum_{i=0 }^{2^{nN}-1} (U P_{M}\ket{0}_{\mathsf{W}}\ket{i}_{\mathsf{R}}) \prod_{j=1}^N (\frac{1}{2}+(-1)^{r_j}(\frac{1}{2}-H_{l_j})) X(s_1)Z(s_0)\ket{i}_{\mathsf{R}'}\Big\|^2.
\end{align*}

We use $\rho$ to denote $$ \frac{1}{2^{nN}} \sum_{M \in \F_2^{h}} \sum_{i=0 }^{2^{nN}-1} (U P_{M}\ket{0}_{\mathsf{W}}\ket{i}_{\mathsf{R}}) X(s_1)Z(s_0)\ket{i}_{\mathsf{R}'} \sum_{l=0 }^{2^{nN}-1} (\bra{0}_{\mathsf{W}}\bra{l}_{\mathsf{R}}P_{M}U^\dagger) \bra{l}_{\mathsf{R}'}X(s_1)Z(s_0).$$ 

Note that 
\begin{align*}
\tr(\rho)=&~ \frac{1}{2^{nN}}\sum_{M} \tr\Big(\sum_{i=0 }^{2^{nN}-1}(\bra{0}_{\mathsf{W}}\bra{i}_{\mathsf{R}}P_{M} \ket{0}_{\mathsf{W}}\ket{i}_{\mathsf{R}}) \Big)\\
=&~ \frac{1}{2^{nN}}\tr\Big(\sum_{i=0 }^{2^{nN}-1}(\bra{0}_{\mathsf{W}}\bra{i}_{\mathsf{R}} I\ket{0}_{\mathsf{W}}\ket{i}_{\mathsf{R}})\Big)\\
=&~1.
\end{align*}

Since $\rho$ is a quantum state, 
\begin{align*}
\omega(V)\leq &~ \eps+\sum_{l_1,l_2,\cdots, l_N\in [m], M \in \F_2^{h}} \frac{1}{m^N} \Pr[M] \cdot \sum_{\oplus_{i\in[N]} r_i = 0} \Pr[r_1,r_2,\cdots, r_N~|~ l_1,l_2,\cdots, l_N, M] \\
=&~ \eps+\sum_{l_1,l_2,\cdots, l_N\in [m], \oplus_{i\in[N]} r_i = 0} \frac{1}{m^N } \tr[ \prod_{j=1}^N (\frac{1}{2}+(-1)^{r_j}(\frac{1}{2}-H_{l_j})) \rho] \\
=&~ \eps+\sum_{l_1,l_2,\cdots, l_N\in [m]}  \frac{1}{m^N } \tr[ \frac{I+\otimes_{i\in [N]} (1-2H_{l_i})}{2} \rho] \\
=&~ \eps+\tr[ \frac{I+ (1-2H)^{\otimes N} }{2} \rho] \\
\leq &~ \frac{1}{2} + \eps+\frac{1}{2} \cdot  2^{-s(n)}.
\end{align*}

\end{proof}

\subsection{Technical claims}

\begin{lemma}
\label{lem:clifford_pauli}
For any $k\in\Z_+$ and any element of the $ k $-fold Clifford group $C$ there exists $ a,b,c\in \F_2^k, d\in \F_2$, such that 
\begin{align}
\label{eq:lem:clifford_pauli}
    C^\dagger 
(Z\otimes \underbrace{I \otimes I\otimes\cdots \otimes I }_{k-1}) C = (-1)^d X(a)Z(b)Y(c)
\end{align} and  $a \land b =b \land c = c \land a=0 $. We call $(-1)^d X(a)Z(b)Y(c)$ the \emph{Pauli representation} of $C^\dagger 
(Z\otimes \underbrace{I \otimes I\otimes\cdots \otimes I }_{k-1}) C$.
\end{lemma}
\begin{proof}
Because $Z\otimes \underbrace{I \otimes I\otimes\cdots \otimes I }_{k-1}$ is an element of the $k$-fold Pauli group and $ C$ is an element of the $ k $-fold Clifford group, 
$C^\dagger (Z\otimes \underbrace{I \otimes I\otimes\cdots \otimes I }_{k-1}) C $ is an element of the $k$-fold Pauli group. Suppose $$C^\dagger (Z\otimes \underbrace{I \otimes I\otimes\cdots \otimes I }_{k-1}) C = (-1)^d 
\i^e X(a)Z(b)Y(c)$$ where $ a,b,c\in \F_2^k, d, e\in \F_2$ and $ a \land b =b \land c = c \land a=0$. 

Because 
\begin{align*}
(-1)^d 
\i^e X(a)Z(b)Y(c)=&~C^\dagger (Z\otimes \underbrace{I \otimes I\otimes\cdots \otimes I }_{k-1}) C \\
=&~ (C^\dagger (Z\otimes \underbrace{I \otimes I\otimes\cdots \otimes I }_{k-1}) C)^\dagger \\
=&~ ((-1)^d 
\i^e X(a)Z(b)Y(c))^\dagger,    
\end{align*}
thus $e=0$. 
\end{proof}

\begin{lemma}
\label{lem:switch_XZ}
Let $A,B$ be matrices.
For $x\in \F_2$, $ \frac{I+(-1)^x A}{2} B = B\frac{I+(-1)^{x+1} A}{2}$ if $ \{A,B\}=0 $, $ \frac{I+(-1)^x A}{2} B = B\frac{I+(-1)^x A}{2}$ if $  [A,B]=0$. 
\end{lemma}
\begin{proof}

{\bf Part 1:}
Because $ AB=-BA$,
$$\frac{I+(-1)^x A}{2} B = B \frac{I+(-1)^{x+1} A}{2}.$$

{\bf Part 2:}
Because $ AB=BA$,
$$\frac{I+(-1)^x A}{2} B = B \frac{I+(-1)^x A}{2}.$$

\end{proof}

\begin{lemma}\label{lem:otimes_obs_eq_sum_otimes_proj}
For any observable $O_i$ such that $ O_i^2 = I$, 
\begin{align*}
\frac{I+\bigotimes_{i\in [N]} O_i}{2} = \sum_{\bigoplus_{i\in [N]} r_i =0} \bigotimes_{i\in [N]} \frac{(I+(-1)^{r_i} O_i)}{2} 
\end{align*}
\end{lemma}
\begin{proof}
The lemma follows from 
\begin{align*}
\sum_{\bigoplus_{i\in [N]} r_i =0} \bigotimes_{i\in [N]} \frac{(I+(-1)^{r_i} O_i)}{2} + \sum_{\bigoplus_{i\in [N]} r_i =1} \bigotimes_{i\in [N]} \frac{(I+(-1)^{r_i} O_i)}{2} =I . 
\end{align*}
and 
\begin{align*}
\sum_{\bigoplus_{i\in [N]} r_i =0} \bigotimes_{i\in [N]} \frac{(I+(-1)^{r_i} O_i)}{2} - \sum_{\bigoplus_{i\in [N]} r_i =1} \bigotimes_{i\in [N]} \frac{(I+(-1)^{r_i} O_i)}{2} = \bigotimes_{i\in [N]} O_i . 
\end{align*}
\end{proof}

\begin{lemma}
\label{lem:same_dist:tel_qiop}
Let $ C$ be an element of the $5$-fold Clifford group. Let $\ket\phi$ be a quantum state. 
Let $r'$ be the random variable that has the distribution of measuring $\ket{\phi}$ under $1-2 C^\dagger 
(\ketbra{0}{0})^{\otimes 5}  C$. 
Let $$O_{j} = (C^\dagger \underbrace{I\otimes I \cdots \otimes I}_{j-1} \otimes Z \otimes \underbrace{I\otimes I \cdots \otimes I}_{5-j}  C ) \otimes I, $$
and let $$O_{j} =  (-1)^{d_{j}} X(a_{j}) Z(b_{j}) Y(c_{j})  $$ be the Pauli representation of $ O_{j}  $ as defined in Lemma \ref{lem:clifford_pauli}. 
Let $ s_0,s_1 \in_R \F_2^n$. 
Let $ (r_1,r_2,\cdots, r_5)$ be the random variable that has the distribution of the measurement outcomes of measuring $Z(s_0)X(s_1)\ket\phi$ under $O_{j}$ from $j = 1$ to $j = 5$  sequentially.
Let $r = 1 \oplus \bigvee_{j\in [5]} (r_{j} \oplus a_{j} s_{0} \oplus b_{j} s_{1} \oplus c_{j} s_{1} \oplus  c_{j} s_{0} ) $. 

Then, the random variables $ r$ and $r'$ have the same distribution.
\end{lemma}
\begin{proof}

The distribution of $r'$ can be described by $\Pr[r'=1]=\bra{\phi} H \ket{\phi}$.

Let $r''_{j} = a_{j} s_{0} \oplus b_{j} s_{1} \oplus c_{j} s_{1} \oplus  c_{j} s_{0}$. Because $ r = 1 \oplus \bigvee_{j\in [5]} (r_{j} \oplus  r''_{j})$, the distribution of $r$ can be described by 
$$ \Pr[r=1]=\|\prod_{j=1}^5  \frac{I+(-1)^{  r''_{j} } O_{{6-j}}}{2}\cdot  X(s_{1})Z(s_{0})\ket{\phi}\|^2. $$ 
Note that $ O_{j} = (-1)^{d_{j}} X(a_{j}) Z(b_{j}) Y(c_{j}) $, by Lemma \ref{lem:switch_XZ}, we have that for $j\in [m]$,
\begin{align*}
\frac{I+ (-1)^{ r''_{j}} O_{j}}{2}\cdot  X(s_{1})Z(s_{0}) =X(s_{1})Z(s_{0}) \cdot \frac{I+ O_{j}}{2}
\end{align*}
Thus, 
\begin{align*}
\Pr[r=1]=&~\|\prod_{j=1}^5  \frac{I+(-1)^{  r''_{j} } O_{{6-j}}}{2}\cdot  X(s_{1})Z(s_{0})\ket{\phi}\|^2\\
=&~\|X(s_{1})Z(s_{0}) \prod_{j=1}^5  \frac{I+O_{{6-j}}}{2}\cdot  \ket{\phi}\|^2 \\
=&~ \| \prod_{j=1}^5  \frac{I+O_{{6-j}}}{2}\cdot  \ket{\phi}\|^2\\
=&~ \bra{\phi}\prod_{j=1}^5  
C^\dagger  \underbrace{I\otimes I \cdots \otimes I}_{5-j} \otimes Z^+ \otimes \underbrace{I\otimes I \cdots \otimes I}_{j-1}  C \ket{\phi}\\
=&~ \bra{\phi}
H \ket{\phi}.
\end{align*}

\end{proof}

\begin{lemma}
\label{lem:xor_r:tel_qiop}
Let $O_i$ for $i\in [m]$ be observables. Let $ H = \frac{1}{m}\sum_{i\in [m]}O_i$. Let $ \ket{\phi}$ be a quantum state. Let $l_i\in [m]$ for $i\in [N]$. 
Let $r'_i$ be the random variable that has the distribution of measuring $\ket{\phi}$ under $O_{l_i}$. Then,
\begin{align}
\Pr[\bigoplus_{i\in [N]} r'_i =0]
 = \Pr[\text{The outcome of measuring $ \ket{\phi}^{\otimes N}$ under $\bigotimes_{i\in [N]} O_{l_i}$  is $0$}] \label{eq:trivial_qcpcp:c:1}
\end{align}
\end{lemma}
\begin{proof}

This is because 
\begin{align*}
\Pr[\bigoplus_{i\in [N]} r'_i =0] = \sum_{\bigoplus_{i\in [N]} r'_i =0} \prod_{i\in [N]} \bra{\phi} \frac{(I+(-1)^{r'_i} O_{l_i})}{2} \ket{\phi},
\end{align*}
and 
\begin{align*}
\Pr[\text{The outcome of measuring $ \ket{\phi}^{\otimes N}$ under $\bigotimes_{i\in [N]} (1-2 H_{l_i})$  is $0$}] = \bra{\phi}^{\otimes N}  \frac{I+\bigotimes_{i\in [N]} O_{l_i}}{2} \ket{\phi}^{\otimes N}.
\end{align*}
The conclusion follows from Lemma \ref{lem:otimes_obs_eq_sum_otimes_proj}.
\end{proof}

%% file: test_n_qubits.tex
\section{Many Qubits Tests}
\label{sec:mqt}

In this section, we consider the extended Pauli matrices $X(g), Z(g)$. 
We design a protocol that forces the prover's strategies to align with these matrices. 
Specifically, an honest prover using these matrices can succeed in the protocol. Conversely, any prover that succeeds with sufficiently high probability must exhibit an approximate anti-commuting relationship when $|g|$ is odd and an approximate commuting relationship when $|g|$ is even, similar to the relationship between $Z(g)$ and $X(g)$. Since the test is a generalization of the single-qubit test sketched in Section~\ref{sec:intro-strong}, we recommend that the reader review that section first. 

\subsection{Error Correcting Code}

In this section we recall the properties of a specific family of error-correcting codes that play a central role in our construction. While we focus on the Hadamard code, most of our arguments apply to any code that satisfies a set of essential properties, such as local testability and self-correction. These properties are crucial for constructing interactive proof systems that are robust to errors and support verification with only limited access to the communicated data.

\begin{theorem}[Locally testable and self-correctable codes]
\label{thm:3c_code}
The family of Hadamard codes $\{C_k\}$ satisfies the following properties:
\begin{enumerate}
    \item The code has rate $r=1/\exp(k)$ and relative distance $d=\Theta(1)$.
    \item It is a linear code with generator matrix $G_k \in \mathbb{F}_2^{n(k) \times k}$, where $n(k)$ denotes the length of the encoded message as a function of the message length $k$.
    \item The code admits deterministic encoding and decoding algorithms, denoted by $\texttt{E}$ and $\texttt{D}$ respectively.
\end{enumerate}

Furthermore, there exists $\varepsilon \in (0, 0.001)$, $\kappa, \kappa' \in (0, 1)$, and $q_{\texttt{L}}, q_{\texttt{S}} \in \mathbb{Z}_+$ such that: 
\begin{enumerate}
\item \textnormal{Local Testing:}  
For all $n$, the code $C_n$ is $\kappa$-locally testable with query complexity $q_{\texttt{L}}$. That is, there exists a non-adaptive probabilistic polynomial-time oracle machine $\texttt{L}$ with query complexity $q_{\texttt{L}}$, such that:
\begin{enumerate}
\item For any $k \in \mathbb{Z}_+$ and $x \in \mathbb{F}_2^k$,
        \[
    \Pr[\texttt{L}^{\texttt{E}(x)}(1^k) = 1] = 1.
        \]
        \item For any $k \in \mathbb{Z}_+$ and $w \in \mathbb{F}_2^{n(k)}$,
        \[
        \Pr[\texttt{L}^w(1^k) = 1] \leq 1 - \kappa \cdot \mathrm{dist}(w, C_n).
        \]
    \end{enumerate}
    
\item \textnormal{Self-Correction:}  
For all $n$, the code $C_n$ is self-correctable with query complexity $q_{\texttt{S}}$. There exists a non-adaptive probabilistic polynomial-time oracle machine $\texttt{S}$ with query complexity $q_{\texttt{S}}$ such that for any $k \in \mathbb{Z}_+$ and any $w \in \mathbb{F}_2^{n(k)}$ satisfying
    \[
    \mathrm{dist}(w, C_n) < \kappa' \cdot d / 2,
    \]
    it holds for all $p \in [n(k)]$ that
    \[
    \Pr[\texttt{S}^w(1^k, p) = (\texttt{E} \circ \texttt{D}(w))_p] \geq 1 - \varepsilon.
    \]
\end{enumerate}
\end{theorem}

\begin{remark}
For $i \in \{0,1\}^k$, let $n(i) \in \mathbb{Z}$ denote the integer whose binary representation is $i$. Then, the $n(i)$-th row of $G_k$ equals $i$, i.e., $G_{k, n(i), \cdot} = i$.
\end{remark}

We now define the randomness-aware views of the testing and correction procedures.

\begin{definition}[Randomness-aware evaluation of local tester and self-corrector]
\label{def:3c}
Let $\{C_n\}$ be the code defined in Theorem~\ref{thm:3c_code}.
\begin{itemize}
    \item Define $\texttt{L}^w(1^k; r)$ as the output of $\texttt{L}^w(1^k)$ when using randomness $r \in \mathbb{F}_2^{l_{\texttt{L}}}$, where $l_{\texttt{L}}$ is the randomness complexity of $\texttt{L}$.
    \item Define $\texttt{S}^w(1^k, p; r)$ as the output of $\texttt{S}^w(1^k, p)$ when using randomness $r \in \mathbb{F}_2^{l_{\texttt{S}}}$, where $l_{\texttt{S}}$ is the randomness complexity of $\texttt{S}$.
\end{itemize}
\end{definition}

Next, we define unitary operations that simulate the decoding, local testing, and self-correction functionalities on quantum registers. These operations will be used in our quantum protocol constructions.

\begin{definition}[Unitary decoding operator]
\label{def:decode}
Let $\texttt{D}$ be the decoding algorithm from Theorem~\ref{thm:3c_code}. For $k \in \mathbb{Z}_+$, $w \in \mathbb{F}_2^{n(k)}$, and $a \in \mathbb{F}_2^k$, define the decoding unitary
\[
D_C(\ket{a} \ket{w}) = \ket{a \oplus \texttt{D}(w)} \ket{w}.
\]
\end{definition}

\begin{definition}[Unitary validity test operator]
\label{def:valid}
Let $\texttt{L}$ be as in Definition~\ref{def:3c}. For $k \in \mathbb{Z}_+$, $w \in \mathbb{F}_2^{n(k)}$, $r \in \mathbb{F}_2^{l_{\texttt{L}}(k)}$, and all $a \in \mathbb{F}_2$, define
\[
L(\ket{a} \ket{w} \ket{r}) = \ket{a \oplus \texttt{L}^w(1^k; r)} \ket{w} \ket{r}.
\]

\end{definition}

\begin{definition}[Unitary self-correction operator]
\label{def:self_corr}
Let $\texttt{S}$ be as in Definition~\ref{def:3c}. For $k \in \mathbb{Z}_+$, $w \in \mathbb{F}_2^{n(k)}$, $r \in \mathbb{F}_2^{l_{\texttt{S}}(k)}$, $a \in \mathbb{F}_2$, and $p \in [n(k)]$, define
\[
S_p(\ket{a} \ket{w} \ket{r}) = \ket{a \oplus \texttt{S}^w(1^k, p; r)} \ket{w} \ket{r}.
\]

\end{definition}

\subsection{Honest prover's strategy}

The honest prover interacts with the verifier over three rounds. We first describe the prover’s actions in the first two rounds, and treat the third round separately.

In the first round, the honest prover performs a purified measurement of the $k$-qubit witness in the $Z$ basis and encodes the $k$-bit outcome using the Hadamard code. He then sends the resulting state to the verifier. After receiving the same register back from the verifier, the prover applies the inverse of the same operation.

In the second round, the prover again performs a purified measurement of the witness, this time in the $X$ basis, and encodes the outcome in the Hadamard code. As in the first round, the prover sends the state to the verifier and later applies the inverse process after receiving the register back.

In the third round, the verifier sends a bit indicating a basis ($X$ or $Z$), and the honest prover performs a purified measurement of the witness in the corresponding basis and encodes the result using the Hadamard code.

We now define the prover’s actions in the first two rounds more precisely. These rounds follow a common template parameterized by observables $O$. The honest prover uses different values of $O$, specifically $O_a = Z(a),a\in\F_2^k$ in the first round and $O_b = X(b), b\in\F_2^k$ in the second round.
We now define the template for the honest prover’s actions in the first two rounds. The detailed implementation of this template will be described later.

\begin{definition}[Template of the honest prover's actions]
\label{def:mqt:honest_prover:block}
Let $k,n\in \Z_+$. 
Let $\mathsf{P}$ be a register of $k$ qubits, $\mathsf{R}=(\mathsf{R}_{j})_{j\in [n]}$ a register of $ n$ qubits. 
$\mathsf{P, R}$ are initially held by the prover. 
For $j\in[n]$, let $ O_{j}$ be observables on register $\mathsf{P}$. 
\begin{mdframed}
\begin{itemize}
\item \textnormal{The prover's first step:} 
The prover applies the following unitary 
\begin{align*}
\prod_{p\in[n]}P(O_{ p, \mathsf{P}}, \mathsf{R}_{p}). 
\end{align*}
The prover sends the verifier the register $\mathsf{R}$.
\item \textnormal{The verifier's first step:} 
The verifier receives the register $\mathsf{R}$ from the prover. The verifier sends the prover the register $\mathsf{R}$. 
\item \textnormal{The prover's second step:} 
The prover receives the register $\mathsf{R}$ from the verifier. The prover applies the following unitary 
\begin{align*}
(\prod_{p\in[n]} P(O_{p, \mathsf{P}}, \mathsf{R}_{p}))^{-1}.
\end{align*}
\end{itemize}
\end{mdframed}
\end{definition}
\begin{remark}
When we implement the template, 
$ O_{j}, j\in [n]$ always commute with each other. 
\end{remark}

Then, we describe the honest prover’s actions in the last round:
\begin{definition}[The honest prover's actions in the last round]
\label{def:mqt:honest_prover:l_consist}
Let $k,n\in \Z_+$. 
Let $\mathsf{P}$ be a register of $k$ qubits, 
$\mathsf{R}=(\mathsf{R}_{j})_{j\in [n]}$ a register of $ n$ qubits. 
$\mathsf{P, R}$ are initially held by the prover. 
For $q'\in[2], j\in[n]$, let $ O_{q', j}$ be observables on register $\mathsf{P}$. 
\begin{mdframed}
\begin{itemize}
\item The verifier sends the prover $q' \in [2]$. 
\item The honest prover acts as in Definition \ref{def:mqt:honest_prover:block} with $ O_{p}=O_{q', p}$ for $p\in [n]$. 
\end{itemize}
\end{mdframed}
\end{definition}
\begin{remark}
When we implement the template, for $q'\in [2], j_1,j_2\in [n]$, 
$ O_{q', j_1}$ always commutes with $ O_{q', j_2}$. 
\end{remark}
\begin{remark}
Note that the prover's second step in Definition \ref{def:mqt:honest_prover:block} is not necessary for Definition \ref{def:mqt:honest_prover:l_consist}. But for the ease of presentation, we keep it. In this section, the same holds for other definitions for the last rounds. Therefore, our protocol essentially contains $6$ messages, instead of $7$ messages.
\end{remark}

Now we summarize the honest prover's actions:
\begin{definition}[Many qubits tests, honest prover]\label{def:mqt:honest_prover}

Let $k,n\in \Z_+$. 
Let $ G_k=[g_1, g_2,\cdots, g_n]^T \in \F_2^{n} \otimes \F_2^k$ be as in Definition \ref{def:3c}.
Let $\mathsf{P}$ be a register of $k$ qubits, $\mathsf{R}$ a register of $n$ qubits. $\mathsf{P, R}$ are initially held by the prover and initialized to $\ket{0}$. 

The honest prover's actions are defined as follows: 
\begin{mdframed}
\begin{itemize}
\item The honest prover prepares the witness in register $\mathsf{P}$.
\item The honest prover acts as in Definition \ref{def:mqt:honest_prover:block} with $ O_{p}=Z(g_p)$ for $p\in [n]$.
\item The honest prover acts as in Definition \ref{def:mqt:honest_prover:block} with $ O_{p}=X(g_p)$ for $p\in [n]$. 
\item The honest prover acts as in Definition \ref{def:mqt:honest_prover:l_consist} with $ O_{1, p}=Z(g_p), O_{2, p}=X(g_p)$ for $p\in [n]$.
\end{itemize}
\end{mdframed}

\end{definition}

\subsection{Verifier's strategy}

Then we present the verifier's action. 
For ease of presentation, we first abstract the process of applying a Z gate on a self-corrected qubit.

\begin{definition}[The verifier's purified forms, Read out with self-correction]
\label{def:mqt:read_self_corr}
Let $S_{p} $ be as in Definition \ref{def:self_corr}. 
Let $\mathsf{R}=(\mathsf{R}_{j})_{j\in [n]}$ be a register of $n$ qubits. 
Let $\mathsf{S}$ be a register. 
Define 
\begin{align*}
{\tilde O}_{p, \mathsf{SRT_S}} = 
S_{p, \mathsf{SRT_S}}^\dagger \cdot Z_{\mathsf{S}} \cdot 
S_{p, \mathsf{SRT_S}}.
\end{align*}

\end{definition}
\begin{remark}
Here, $S_p$ performs self-correction on the codeword stored in register $\mathsf{R}$ at position $p$, and stores the self-corrected result in register $\mathsf{S}$. To perform the self-correction, the algorithm requires some classical randomness, which is stored in $\mathsf{T_S}$. Consequently, the quantum state stored in $\mathsf{T_S}$ is the all-$\ket{+}$ state. Note that although we describe the algorithm in this way, in the actual execution, the verifier samples classical randomness.
\end{remark}

In the first two rounds, the verifier can respond in various ways to the prover's messages. Later, we will show how the verifier chooses to respond in the whole protocol. 
In the following definition, $M, F, m, f$ can be viewed as the input to the template. 
$M$ indicates whether or not to measure. 
$m$ indicates which register to measure. 
$F$ indicates whether or not to flip. 
$f$ indicates which register to flip.

\begin{definition}[Template of the verifier's action for the consistency checks in the first two rounds]\label{def:mqt:verifier:block}
Let $\mathsf{R} =(\mathsf{R}_{j})_{j\in[n]}$ be a register. Let $ \mathsf{S}$ be a register held by the verifier and initialized with $ \ket{0}$. Let $ \mathsf{T_S} $ be a register held by the verifier and initialized with $H^{\otimes l_{\mathsf{S}}} \ket{0}$, where $l_{\mathsf{S}} $ is defined as in Definition \ref{def:3c}. 
Let $M,F\in \{0,1\}, m, f\in [n]$. 
\begin{mdframed}
\begin{itemize}
\item \textnormal{The prover's first step:} 
The prover sends the verifier the register $\mathsf{R}$.
\item \textnormal{The verifier's first step:} 
The verifier receives the register $\mathsf{R}$ from the prover. 
If $F=1$, the verifier measures $\mathsf{T_S}$, then the verifier applies ${\tilde O}_{ f, \mathsf{SRT_S}}$. If $M=1$, the verifier measures $\mathsf{T_S}$. Then, the verifier
measures ${\tilde O}_{m,\mathsf{SRT_S}}$ and outputs the measurement result $r$. 
The verifier sends the prover the register $\mathsf{R}$. 
\item \textnormal{The prover's second step:} 
The prover receives the register $\mathsf{R}$ from the verifier. 
\end{itemize}
\end{mdframed}
\end{definition}

Then, we describe the verifier's action in the last round:
\begin{definition}[The verifier's action for the consistency checks in the last round]
\label{def:mqt:verifier:last}
Let $q'\in[2], m\in [n]$. 
\begin{mdframed}
\begin{itemize}
\item  
The verifier sends $q'$ to the prover.
\item The verifier acts as in Definition \ref{def:mqt:verifier:block} with $M=1,F=0,m=m,f=0$. 
\end{itemize}
\end{mdframed}
\end{definition}

Now, we introduce the verifier's action in the many qubits tests. The verifier has two types of checks: consistency checks and anti-commuting checks. We first introduce the consistency check. In this check, the verifier checks the consistency between the first two rounds and the last round.
More specifically, the verifier checks that the answer in the final round matches the answer from the previous rounds.
\begin{definition}[Many qubits tests, consistency checks]\label{def:mqt:verifier:c} 
Let $ D$ be a distribution over $[n]$. 
Let $\mathsf{R} =(\mathsf{R}_{j})_{j\in[n]}$ be a register. Let $ \mathsf{S}$ be a register held by the verifier and initialized with $ \ket{0}$. Let $ \mathsf{T_{S, i}} $ for $i\in [3]$ be a register held by the verifier and initialized with $H^{\otimes l_{\mathsf{S}}} \ket{0}$. 
The verifier's actions are defined as follows: 
\begin{mdframed}
\begin{itemize}
\item The verifier samples $u\in_R [2]$, sets $M_u=1$, and sets $ M_i =0$ for $i \neq u, i\in [2]$. 
\item If $M_1=1$, the verifier sets $F_i = 0$ for $i\in [2]$. If $M_2=1$, the verifier sets $F_1 \in_R \F_2$, $F_2 = 0$.
\item The verifier samples $f_1,f_2\sim D$. 
\item The verifier repeats the verifier's actions in Definition \ref{def:mqt:verifier:block} twice. In the $i$-th time, the verifier acts as in Definition \ref{def:mqt:verifier:block} with $ M = M_i, F = F_i, m=f_i, f=f_i$; if $M_i=1$ he sets $r_i$ to be the output. 
\item In the end, the verifier acts as in Definition \ref{def:mqt:verifier:last} with $q'= u, m=f_u$. Let $r'$ be the output of the verifier.  
\item The verifier accepts iff $ r_u=r'$.
\end{itemize}
\end{mdframed}
\end{definition}

Then, we introduce the anti-commuting tests, where the verifier checks the consistency corresponding to the first $Z$ measurement and the $Z$ measurement in the last round with an $X$ applied in the middle.
\begin{definition}[Many qubits tests, anti-commuting tests]\label{def:mqt:verifier:a}
Let $ D$ be a distribution over $[n]$. 
Let $ G_k=[g_1, g_2,\cdots, g_n]^T \in \F_2^{n} \otimes \F_2^k$ be defined as in Definition \ref{def:3c}.
Let $\mathsf{R} =(\mathsf{R}_{j})_{j\in[n]}$ be a register. Let $ \mathsf{S}$ be a register held by the verifier and initialized with $ \ket{0}$. Let $ \mathsf{T_{S, i}} $ for $i\in [3]$ be a register held by the verifier and initialized with $H^{\otimes l_{\mathsf{S}}} \ket{0}$. 
The verifier's action are defined as follows: 
\begin{mdframed}
\begin{itemize}
\item The verifier sets $M_1=1, F_1=0, M_2=0, F_2=1$.  
\item The verifier samples $ (f_1,f_2)\sim D$. 
\item The verifier repeats the verifier's action in Definition \ref{def:mqt:verifier:block} for $2$ times. In the $i$-th time, the verifier acts as in Definition \ref{def:mqt:verifier:block} with $ M = M_i, F = F_i, m=f_i, f=f_i$, if $M_i=1$, sets $r_i$ be the output. 
\item In the end, the verifier acts as in Definition \ref{def:mqt:verifier:last} with $q'= 1, m=f_1$. Let $r'$ be the output of the verifier. 
\item The verifier accepts iff $ r_1=r' \oplus (g_{f_1}\cdot  g_{f_2})$.
\end{itemize}
\end{mdframed}
\end{definition}

In the end, we define the validity checks, where we check that the prover indeed sends the verifier an instance of the error correction code. 
\begin{definition}[Template of the verifier's validity checks]\label{def:mqt:verifier:b_code}
Let $\mathsf{R} =(\mathsf{R}_{j})_{j\in[n]}$ be a register. Let $ \mathsf{L}$ be a register held by the verifier and initialized with $ \ket{0}$. Let $ \mathsf{T_L} $ be a register held by the verifier and initialized with $H^{\otimes l_{\mathsf{L}}} \ket{0}$. 
Let $ \mathsf{T_S} $ be a register held by the verifier and initialized with $H^{\otimes l_{\mathsf{S}}} \ket{0}$. Let $t\in \{0,1\}, F\in\{0,1\},f\in[n]$. 
\begin{mdframed}
\begin{itemize}
\item \textnormal{The prover's first step:} 
The prover sends the verifier the register $\mathsf{R}$.
\item \textnormal{The verifier's first step:} 
The verifier receives the register $\mathsf{R}$ from the prover. 
If $F=1$, the verifier measures $\mathsf{T_S}$, then applies ${\tilde O}_{ f, \mathsf{SRT_S}}$.
\begin{itemize}
\item If $t=0$, the verifier sends the prover the register $\mathsf{R}$. 
\item If $t=1$, the verifier measures $\mathsf{T_L}$, then the verifier applies ${L}_{\mathsf{LRT_L}}$ and measures $Z_{\mathsf{L}}$ and outputs the measurement result $ r$. 
The verifier sends the prover the register $\mathsf{R}$. 
\end{itemize}
\item \textnormal{The prover's second step:} 
The prover receives the register $\mathsf{R}$ from the verifier. 
\end{itemize}
\end{mdframed}
\end{definition}

Then, we describe the verifier's action in the last round:

\begin{definition}[The verifier's validity checks in the last round]
\label{def:mqt:verifier:l_code}
Let $t\in \{0,1\}$. Let $q'\in [2]$. 
\begin{mdframed}
\begin{itemize}
\item  
The verifier sends $q'$ to the prover.
\item The verifier acts as in Definition \ref{def:mqt:verifier:b_code} with $t=t, F=0, f=0$. 
\end{itemize} 
\end{mdframed}
\end{definition}

We now specify the verifier’s parameters used during the validity checks.
\begin{definition}[Many qubits tests, validity checks]\label{def:mqt:verifier:v} 
Let $ D$ be a distribution over $[n]$. 
Let $\mathsf{R} =(\mathsf{R}_{j})_{j\in[n]}$ be a register. Let $ \mathsf{S}, \mathsf{L}$ be registers held by the verifier and initialized with $ \ket{0}$. Let $ \mathsf{T_{S, i}} $ for $i\in [3]$ be a register held by the verifier and initialized with $H^{\otimes l_{\mathsf{S}}} \ket{0}$.
Let $ \mathsf{T_{L, i}} $ for $i\in [3]$ be a register held by the verifier and initialized with $H^{\otimes l_{\mathsf{L}}} \ket{0}$.
The verifier's actions are defined as follows: 
\begin{mdframed}
\begin{itemize}
\item The verifier samples $u\in_R [3]$. The verifier sets $t_u=1$ and $t_i=0$ for $i\neq u, i\in [3]$.
\item The verifier samples $F_i\in_R \{0, 1\}$ for $i\in [u-1]$. The verifier sets $F_i=0$ for $i\geq u$. 
\item The verifier samples $f_1,f_2\sim D$. 
\item The verifier samples $q'\in_R [2]$. 
\item The verifier repeats the verifier's actions in Definition \ref{def:mqt:verifier:b_code} twice. In the $i$-th time, the verifier acts as in Definition \ref{def:mqt:verifier:b_code} with $t=t_i, F=F_i, f=f_i$ and if $t_i=1$ sets $r_i$ to be the output. 
\item In the end, the verifier acts as in Definition \ref{def:mqt:verifier:l_code} with $t=t_3,q'=q'$ and if $t_3=1$ sets $r_3$ to be the output.  
\item The verifier accepts iff $ r_u=1$.
\end{itemize}
\end{mdframed}
\end{definition}

Before describing the verifier’s actions in the many-qubits test, we first present the distribution used by the verifier during the test.

\begin{definition}[Verifier distribution for many qubits tests]\label{def:mqdef}
Let $p, q$ be as in Theorem~\ref{thm:clifford_ham}, and let $s, N$ be as in Lemma~\ref{lem:c_ham_amp}. 
Let $H = \frac{1}{m} \sum_{i \in [m]} H_i \in \mathbb{C}^{2^n} \otimes \mathbb{C}^{2^n}$ be a $5$-local Clifford Hamiltonian, where $m = \mathrm{poly}(n)$, and either $\lambda_{\min}(H) \leq 2^{-p(n)}$ or $\lambda_{\min}(H) \geq 1/q(n)$. 

For each $i \in [m]$, let
\begin{align*}
H_i =&~ \left(C_i^\dagger (\ketbra{0}{0})^{\otimes 5} C_i \right) \otimes I,\\
O_{i,j} =&~ \left(C_i^\dagger (I^{\otimes (j-1)} \otimes Z \otimes I^{\otimes (5-j)}) C_i\right) \otimes I.    \end{align*}
Each $O_{i,j}$ admits a Pauli representation
\[
O_{i,j} = (-1)^{d_{i,j}} X(a_{i,j}) Z(b_{i,j}) Y(c_{i,j}),
\]
as described in Lemma~\ref{lem:clifford_pauli}.
Let $t \in \mathbb{F}_2^{5N}$ and $l \in [m]^N$. Define
\[
O(t, l) = \bigotimes_{i \in [N]} \prod_{j \in [5]} O_{l_i, j}^{t_{i, j}},
\]
which has Pauli representation:
\[
O(t, l) = (-1)^{d(t, l)} X(a(t, l)) Z(b(t, l)) Y(c(t, l)).
\]
Define the measurement distributions used in the test as follows:
\begin{itemize}
\item Sample $ t \in_R \mathbb{F}_2^{5N}$, $r\in_R \mathbb{F}_2^{Nn}$, and $l_i \in_R [m]$ for each $i \in [N]$.
\item Define $\mu_1$ as the distribution over the pair $(b(t, l) + c(t, l), a(t, l) + c(t, l))$.
\item Define $\mu_2$ as the distribution over the pair $(b(t, l) + c(t, l) + r, a(t, l) + c(t, l))$.
\end{itemize}
\end{definition}

Now we put everything together and define the actions of the verifier in the many qubits tests.

\begin{definition}[Many qubits tests, verifier]\label{def:mqt:verifier} 
Let $\mathsf{R} $ be a register.  

The verifier's actions are defined as follows: 
\begin{mdframed}
\begin{enumerate}
\item With probability $1/6$, the verifier acts as in Definition \ref{def:mqt:verifier:c} with $D=\mu_1$,
\item With probability $1/6$, the verifier acts as in Definition \ref{def:mqt:verifier:a} with $D=\mu_1$,
\item With probability $1/6$, the verifier acts as in Definition \ref{def:mqt:verifier:v} with $D=\mu_1$,
\item With probability $1/6$, the verifier acts as in Definition \ref{def:mqt:verifier:c} with $D=\mu_2$,
\item With probability $1/6$, the verifier acts as in Definition \ref{def:mqt:verifier:a} with $D=\mu_2$,
\item With probability $1/6$, the verifier acts as in Definition \ref{def:mqt:verifier:v} with $D=\mu_2$,
\end{enumerate}
\end{mdframed}
\end{definition}

\subsection{General prover's actions}
Before we state our main theorem in this section, we first introduce a general prover's actions. 
\begin{definition}[The canonical prover's actions]
\label{def:mqt:canonical_prover}
Let $\mathsf{P}$ be a register, $\mathsf{R}$ be a register of $n$ qubits. 
$\mathsf{P, R}$ are initially hold by the prover and initialized by $\ket{0}$. 
For $q'\in [2]$, let $ U'_{q'}$ be a unitary. 
For $i\in [2]$, let $ U_i$ be a unitary. 
\begin{mdframed}
\hfill 

The canonical prover's actions in the first two rounds are defined as follows: At round $i$, 
\begin{itemize}
\item \textnormal{The prover's first step:} 
The prover applies a unitary $ {U}_{i,\mathsf{PR}}$. 
The prover sends the verifier the register $\mathsf{R}$.
\item \textnormal{The verifier's first step:} 
The verifier receives the register $\mathsf{R}$ from the prover. The verifier sends the prover the register $\mathsf{R}$. 
\item \textnormal{The prover's second step:} 
The prover receives the register $\mathsf{R}$ from the verifier. 
\end{itemize}

The canonical prover's actions in the last round  are defined as follows: 
\begin{itemize}
\item \textnormal{The verifier's first step:} 
The verifier sends $q'\in [2]$ to the prover.
\item \textnormal{The prover's first step:} 
The prover applies a unitary $ {U}'_{q',\mathsf{PR}}$. 
The prover sends the verifier the register $\mathsf{R}$.
\item \textnormal{The verifier's first step:} 
The verifier receives the register $\mathsf{R}$ from the prover. The verifier sends the prover the register $\mathsf{R}$. 
\item \textnormal{The prover's second step:} 
The prover receives the register $\mathsf{R}$ from the verifier. 
\end{itemize}
\end{mdframed}
\end{definition}
\begin{remark}
Note that any prover's actions can be purified into a canonical prover's actions described above. 
\end{remark}

For the ease of presentation, we introduce prover's and verifier's purified forms and clean forms.
\begin{definition}
[The prover's purified forms]
\label{def:pure_prover}
Let $\mathsf{P}, \mathsf{S}, \mathsf{T_S}$ be a register, $\mathsf{R}$ be a register of $n$ qubits. 
For $q'\in [2]$, let $ U'_{q'}$ be a unitary. 
For $a\in \F_2^k$, define the prover's hermitian forms in the last round as:
\begin{align*}
{\tilde O}'_{q',j,\mathsf{SPRT_S}} = ({U}'_{q', \mathsf{PR}})^{\dagger} \cdot {\tilde O}_{ j, \mathsf{SRT_S}}\cdot {U}'_{q', \mathsf{PR}}.
\end{align*}
\end{definition}
\begin{remark}
 The verifier's purified forms are defined in Definition \ref{def:mqt:read_self_corr}.
\end{remark}

\begin{definition}
[The verifier's clean forms]
\label{def:hermitian_verifier}
Let $\mathsf{P}, \mathsf{D}$ be a register, $\mathsf{R}$ be a register of $n$ qubits. 
Let $D_C$ be defined as in Definition \ref{def:decode}.
For $a\in \F_2^k$, define the verifier's hermitian forms in the first two rounds as:
\begin{align*}
{O}_\mathsf{DR}(a) = D_{C, \mathsf{DR}}^\dagger Z_{\mathsf{D}}(a) D_{C, \mathsf{DR}} .
\end{align*}
\end{definition}
\begin{remark}
The unitary $D_C$ decodes the quantum state stored in register $\mathsf{R}$ and stores it in register $\mathsf{D}$. The overall operator effectively applies $Z(a)$ directly on the quantum state protected by the classical error correction code.
\end{remark}

\begin{definition}
[The prover's clean forms]
\label{def:hermitian_prover}
Let $\mathsf{P}, \mathsf{D}$ be a register, $\mathsf{R}$ be a register of $n$ qubits. 
For $q'\in [2]$, let $ U'_{q'}$ be a unitary. 
For $a\in \F_2^k$, define the prover's hermitian forms in the last round as:
\begin{align*}
{O}'_{q',\mathsf{DPR}}(a) = ( U'_{q', \mathsf{PR}})^\dagger O_{\mathsf{DR}}(a)  U'_{q', \mathsf{PR}}.
\end{align*}
\end{definition}

For ease of presentation, we introduce transition actions.
\begin{definition}[Transition actions]\label{def:trans}
Let $ {U}_{i}$ for $i\in[2]$ be defined as in Definition \ref{def:mqt:canonical_prover}. 
For $i<j$, define the transition action from $i$ to $j$ as:
\begin{align*}
{U}_{i\rightarrow j} ={U}_{j} \cdots {U}_{i+2}{U}_{i+1}.
\end{align*}
\end{definition}

\begin{definition}[Transition actions with flips in the purified forms]\label{def:trans_f}
Let $ {U}_{i}$ for $i\in[2]$ be defined as in Definition \ref{def:mqt:canonical_prover}.
For $F\in \{0, 1\}^2, f\in [n]^2$, define the transition action with flips in the purified forms as:
\begin{align*}
{\tilde U}_{1, F, f, \mathsf{SPRT_S}} =&~ {\tilde O}_{f_1, \mathsf{SRT_S}}^{F_1} 
{U}_{1}, \\
{\tilde U}_{2, F, f, \mathsf{SPRT_S}}=&~ {\tilde O}_{f_2,\mathsf{SPRT_S}}^{F_2} {U}_{2}\cdot  {\tilde U}_{1, F, f, \mathsf{SPRT_S}} .
\end{align*}

\end{definition}

\begin{definition}[Transition actions with flips in the clean forms]\label{def:trans_f_herm}
Let $ G_k=[g_1, g_2,\cdots, g_n]^T \in \F_2^{n} \otimes \F_2^k$ be defined as in Definition \ref{def:3c}. 
Let $ {U}_{i}$ for $i\in[2]$ be defined as in Definition \ref{def:mqt:canonical_prover}.
For $F\in \{0, 1\}^2, f\in [n]^2$, define the transition action with flips in the clean forms as:
\begin{align*}
{U}_{1, F, f} =&~  O(g_{f_1})^{F_1} {U}_{1}, \\
{U}_{2, F, f} =&~  O(g_{f_2})^{F_2} {U}_{2} \cdot {U}_{1, F, f}.
\end{align*}

\end{definition}

\subsection{Main theorem}

Now, we state our main theorem in this section. 
\begin{theorem}[Many qubits tests]
\label{thm:mqt}
Let $ \texttt{V}$ be the verifier defined in Definition \ref{def:mqt:verifier}. Then, 

\begin{enumerate}
\item \textnormal{Completeness:} $ \texttt{V}$ accepts the prover in Definition \ref{def:mqt:honest_prover}  with probability $1$.
\item \textnormal{Soundness:} Let $ \texttt{P}$ be the prover defined in Definition \ref{def:mqt:canonical_prover}.  Let $ O'_{1}(\cdot),O'_{2}(\cdot)$ be defined as in Definition \ref{def:hermitian_prover}. 
If $ \texttt{V}$ accepts $ \texttt{P}$ with probability greater than $ 1- \delta$, then approximate anti-commuting / commuting relationships of $ O'_{1}(a),O'_{2}(b)$ hold under the state 
$$\rho ={U}_{0\rightarrow 2}\ket{000}_{\mathsf{PRD}}\bra{000}_{\mathsf{PRD}}{U}_{0\rightarrow 2}^\dagger. $$
More specifically,
$$\E_{(a,b)\sim \mu_1}\| O'_{1}(a)O'_{2}(b)-(-1)^{a\cdot b}O'_{2}(b)O'_{1}(a)\|^2_\rho =  \Theta(\delta),$$
and 
$$\E_{(a,b)\sim \mu_2}\| O'_{1}(a)O'_{2}(b)-(-1)^{a\cdot b}O'_{2}(b)O'_{1}(a)\|^2_\rho =  \Theta(\delta).$$
\end{enumerate}
\end{theorem}

Before proving the main theorem, we first 
establish several lemmas that will aid in the proof.
We begin by examining some basic properties of Definition \ref{def:hermitian_prover} and Definition \ref{def:hermitian_verifier}.
\begin{lemma}
[Properties of the prover's hermitian forms]
Let $O'_{q'}(a)$ for $q'\in [2], a\in \F_2^k$ be as in Definition \ref{def:hermitian_prover}. 

\begin{enumerate}
\item For $q'\in [2], a\in \F_2^k$,  $$O'_{q'}(a) = (O'_{q'}(a))^\dagger,$$
\item For $q'\in [2], a\in \F_2^k$,  $$(O'_{q'}(a))^2 = I,$$
\item For $q'\in [2], a,b\in \F_2^k$,   $$[O'_{q'}(a), O'_{q'}(b)]=0,$$
\item For $q'\in [2], a,b\in \F_2^k$,   $$O'_{q'}(a)\cdot O'_{q'}(b)=O'_{q'}(a\oplus b).$$
\end{enumerate}
\end{lemma}
\begin{proof}
This lemma follows from a straightforward calculation.
\end{proof}

\begin{lemma}
[Properties of the verifier's hermitian forms]
Let $O(a)$ for $a\in \F_2^k$ be defined as in Definition \ref{def:hermitian_verifier}. 

\begin{enumerate}
\item For $a\in \F_2^k$,  $$O(a) = (O(a))^\dagger,$$
\item For $a\in \F_2^k$,  $$(O(a))^2 = I,$$
\item For $a,b\in \F_2^k$,   $$[O(a), O(b)]=0,$$
\item For $a,b\in \F_2^k$,   $$O(a)\cdot O(b)=O(a\oplus b).$$
\end{enumerate}
\end{lemma}
\begin{proof}
This lemma follows from a straightforward calculation.
\end{proof}

Next, we examine the relationship between the overall success probability and the success probabilities in each configuration.
\begin{lemma}\label{lem:mqt:prob}
Let $\texttt{V},\texttt{P}$ be defined as in Theorem \ref{thm:mqt}. 
Let $D_c$ be the distribution of $ M,F$ in Definition \ref{def:mqt:verifier:c}.
Let $D_v$ be the distribution of $ u$ in Definition \ref{def:mqt:verifier:v}.

Then
$$ \Pr[\text{ \texttt{V} accepts \texttt{P} } ] \geq 1-\Theta(\delta),$$  
if and only if
\begin{itemize}
\item For any $u\in[3], F\in \{0,1\}^2,q'\in[2]$ in the support of $D_v$,
\begin{align*}
\Pr[\text{ validity check passes }~|~F,q',u, D~] \geq 1- \Theta(\delta).
\end{align*}
\item For any $M,F\in\{0, 1\}^2$ in the support of $D_c$,
\begin{align*}
\Pr[\text{ consistency check passes }~|~M,F, D~] \geq 1- \Theta(\delta).
\end{align*}
\item Moreover, 
\begin{align*}
\Pr[\text{ anti-commuting test passes }~|~ D] \geq 1- \Theta(\delta).
\end{align*}
\end{itemize}
\end{lemma}
\begin{proof}
This lemma follows from a straightforward application of the union bound.
\end{proof}

The following lemma describes the relationship between the success probability of the local test (see Definition~\ref{def:valid}) and the coefficients of the underlying quantum state. 

For the next lemmas, recall notation from Theorem~\ref{thm:3c_code}.

\begin{lemma}\label{lem:7.38}
Let
\begin{align*}
\ket\psi=\ket{\phi}_{\mathsf{PR}}\otimes \ket{0}_{\mathsf{L}} \otimes H^{\otimes l_{\mathsf{L}}}\ket{0}_{\mathsf{T_L}}\;,
\end{align*}
where
\begin{align*}
\ket{\phi}=\sum_{w\in\F_2^n} \alpha_{w} \ket{\phi_w}_\mathsf{P}\ket{w}_\mathsf{R}\;,
\end{align*}
for some arbitrary (normalized) coefficients $\alpha_w$. Then
\begin{align*}
\|(\ketbra{1}{1}_\mathsf{L}\otimes I_{\mathsf{PRT_L}})  \cdot {L}_{ \mathsf{LRT_L}}\ket\psi\|^2 = \sum_{w\in\F_2^n}  |\alpha_{w}|^2\cdot \Pr[\texttt{L}^w(1^k)=1],
\end{align*}
and
\begin{align*}
\|(\ketbra{1}{1}_\mathsf{L}\otimes I_{\mathsf{PRT_L}})  \cdot {L}_{ \mathsf{LRT_L}}\ket\psi\|^2 \leq&~1-\frac{\kappa \cdot \kappa'\cdot d }{2}\Big(1- \sum_{dist(w, C_n) < \kappa'\cdot d / 2}  |\alpha_{ w}|^2 \Big).
\end{align*}
\end{lemma}

\begin{proof}
Note that
\begin{align*} &~ L_{}(\ket{a}\ket{w}\ket{r}) = \ket{a\oplus \texttt{L}^w(1^k; r)}\ket{w}\ket{r}, \end{align*}
thus
\begin{align*}
\Pr[\texttt{L}^w(1^k)=1]=\|(\ketbra{1}{1}_\mathsf{L}\otimes I_{\mathsf{PRT_L}})  \cdot {L}_{ \mathsf{LRT_L}}(\ket{\phi_w}_{\mathsf{P}}\ket{w}_{\mathsf{R}} \ket{0}_{\mathsf{L}} H^{\otimes l_{\mathsf{L}}}\ket{0}_{\mathsf{T_L}})\|^2 ,
\end{align*}
and
\begin{align*}
\sum_{w\in\F_2^n}|\alpha_w|^2\cdot \Pr[\texttt{L}^w(1^k)=1]=\|(\ketbra{1}{1}_\mathsf{L}\otimes I_{\mathsf{PRT_L}})  \cdot {L}_{ \mathsf{LRT_L}}(\ket\psi)\|^2.
\end{align*}
Moreover, by the guarantees of Theorem~\ref{thm:3c_code},
\begin{align*} \Pr[\texttt{L}^{w}(1^k)=1] \leq 1- \kappa \cdot dist(w, C_n). \end{align*} 
Thus
\begin{align*}
&~ \sum_{w\in\F_2^n} |\alpha_{w}|^2\cdot \Pr[\texttt{L}^w(1^k)=1]\\
\leq &~ \sum_{w\in\F_2^n}  |\alpha_{w}|^2\cdot(1- \kappa \cdot dist(w, C_n)) \\
\leq &~ \sum_{dist(w, C_n) < \kappa'\cdot d / 2} |\alpha_{w}|^2 + \sum_{dist(w, C_n) \geq \kappa'\cdot d / 2}  |\alpha_{w}|^2\cdot(1- \kappa \cdot \kappa'\cdot d / 2) \\
= &~1-\frac{\kappa \cdot \kappa'\cdot d }{2}\Big(1- \sum_{dist(w, C_n) < \kappa'\cdot d / 2}  |\alpha_{w}|^2 \Big).
\end{align*}
\end{proof}
 
The following lemma characterizes the relationship between the difference of the purified and clean forms and the coefficients of the underlying quantum state.

\begin{lemma}\label{lem:7.39}
Let 
\begin{align*}
\ket\psi=\ket{\phi}_{\mathsf{PR}}\otimes \ket{0}_{\mathsf{SD}} \otimes H^{\otimes l_{\mathsf{S}}}\ket{0}_{\mathsf{T_S}},
\end{align*}
and
\begin{align*}
\ket{\phi}=\sum_{w\in\F_2^n} \alpha_{w} \ket{\phi_w}_\mathsf{P}\ket{w}_\mathsf{R}.
\end{align*}
Then
\begin{align*}
\|
{\tilde O}_{p,\mathsf{SRT_S}}\ket\psi -
O_\mathsf{DR}(g_p)\ket\psi \|^2
\leq &~4-(4-4\eps) \sum_{ dist(w, C_n) < \kappa'\cdot d / 2 }|\alpha_{w} |^2  .
\end{align*} 

\end{lemma}
\begin{proof}

Note that
\begin{align*}
{O}_\mathsf{DR}(a) = D_{ \mathsf{DR}}^\dagger Z_{\mathsf{D}}(a) D_{ \mathsf{DR}} ,
\end{align*}
thus
\begin{align*}
{O}_\mathsf{DR}(g_p)\ket{0}\ket{w} =  (-1)^{(\texttt{E}\circ\texttt{D}(w))_p}\ket{0}\ket{w}.
\end{align*}
Moreover, note that
\begin{align*} {\tilde O}_{p, \mathsf{SRT_S}} = S_{p}^\dagger \cdot Z \cdot S_{p},\end{align*} 
and
\begin{align*} &~ S_{p}(\ket{a}\ket{w}\ket{r}) = \ket{a\oplus \texttt{S}^w(1^k, p; r)}\ket{w}\ket{r}, \end{align*}
thus
\begin{align*}
{\tilde O}_{p}  \ket{0}\ket{w}\ket{r} = (-1)^{\texttt{S}^w(1^k, p;r)} \ket{0}\ket{w}\ket{r}. 
\end{align*}
Then, we have that
\begin{align}
&~\|
{\tilde O}_{p,\mathsf{SRT_S}}\ket\psi -
O_\mathsf{DR}(g_p)\ket\psi \|^2\notag\\
&~=2-2\Re[({\tilde O}_{p,\mathsf{SRT_S}}\ket\psi)^\dagger (O_\mathsf{DR}(g_p)\ket\psi)],\label{eq:o-1}
\end{align}
and
\begin{align}
&~({\tilde O}_{p,\mathsf{SRT_S}}\ket\psi)^\dagger (O_\mathsf{DR}(g_p)\ket\psi)\notag\\
=&~(\sum_{w}\E_r \alpha_{w} (-1)^{\texttt{S}^w(1^k, p; r)} \ket{0}\ket{\phi_w}\ket{w}\ket{r})^\dagger(\sum_{w}\E_r \alpha_{w} (-1)^{(\texttt{E}\circ\texttt{D}(w))_p}\ket{0}\ket{\phi_w}\ket{w}\ket{r} )\notag\\
=&~\sum_{w}\E_r |\alpha_{w}|^2\cdot (-1)^{\texttt{S}^w(1^k, p; r)+(\texttt{E}\circ\texttt{D}(w))_p} .\label{eq:o-2}
\end{align}
Note that 
\begin{align*}
dist(w, C_n) < \kappa'\cdot d / 2 \implies \forall p\in[n], \Pr[\texttt{S}^w(1^k, p)=(\texttt{E}\circ\texttt{D}(w))_p] \geq 1-\eps,
\end{align*}
thus
\begin{align*}
&~\sum_w \E_r|\alpha_{w} |^2 \cdot (-1)^{\texttt{S}^w(1^k, p;r)+(\texttt{E}\circ\texttt{D}(w))_p}  \\
\geq&~\sum_{ dist(w, C_n) < \kappa'\cdot d / 2 }|\alpha_{w} |^2 \E_r (-1)^{\texttt{S}^w(1^k, p; r)+(\texttt{E}\circ\texttt{D}(w))_p} - \sum_{ dist(w, C_n) \geq \kappa'\cdot d / 2 } |\alpha_{w} |^2 \\
\geq&~ (1-2\eps) \sum_{ dist(w, C_n) < \kappa'\cdot d / 2 }|\alpha_{w} |^2  - \sum_{ dist(w, C_n) \geq \kappa'\cdot d / 2 } |\alpha_{w} |^2 \\
=&~ (2-2\eps) \sum_{ dist(w, C_n) < \kappa'\cdot d / 2 }|\alpha_{w} |^2  - 1.
\end{align*}
Plugging back into~\eqref{eq:o-1} and~\eqref{eq:o-2},
\begin{align*}
\|
{\tilde O}_{p,\mathsf{SRT_S}}\ket\psi -
O_\mathsf{DR}(g_p)\ket\psi \|^2\leq&~ 2-2((2-2\eps) \sum_{ dist(w, C_n) < \kappa'\cdot d / 2 }|\alpha_{w} |^2  - 1)\\
=&~ 4-(4-4\eps) \sum_{ dist(w, C_n) < \kappa'\cdot d / 2 }|\alpha_{w} |^2  
.
\end{align*}
\end{proof}

\begin{lemma}[Implication of the local testing]\label{lem:7.40}
Let $\ket{\psi}$ be in Lemma~\ref{lem:7.38}. Then
\begin{align*}
\|
{\tilde O}_{p,\mathsf{SRT_S}}\ket\psi -
O_\mathsf{DR}(g_p)\ket\psi \|^2 \leq 4\eps+\frac{8-8\eps}{\kappa \cdot \kappa'\cdot d }\Big(1-\|(\ketbra{1}{1}_\mathsf{L}\otimes I_{\mathsf{PRT_L}})  \cdot {L}_{ \mathsf{LRT_L}}\ket\psi\|^2\Big). 
\end{align*}
Furthermore, for $\eps, \kappa, \kappa', d$ defined as in Theorem \ref{thm:3c_code}, 
\begin{align*}
\|
{\tilde O}_{p,\mathsf{SRT_S}}\ket\psi -
O_\mathsf{DR}(g_p)\ket\psi \|^2 \leq \Theta(\eps)+\Theta\Big(1-\|(\ketbra{1}{1}_\mathsf{L}\otimes I_{\mathsf{PRT_L}})  \cdot {L}_{ \mathsf{LRT_L}}\ket\psi\|^2\Big). 
\end{align*}
\end{lemma}

\begin{proof}
By Lemma~\ref{lem:7.38} we have that
\begin{align*}
1-\frac{2}{\kappa \cdot \kappa'\cdot d }\Big(1-\|(\ketbra{1}{1}_\mathsf{L}\otimes I_{\mathsf{PRT_L}})  \cdot {L}_{ \mathsf{LRT_L}}\ket\psi\|^2
\Big)
 \leq&~\sum_{dist(w, C_n) < \kappa'\cdot d / 2}  |\alpha_{ w}|^2.
\end{align*}
Thus,
\begin{align*}
&~ 4-(4-4\eps) \sum_{ dist(w, C_n) < \kappa'\cdot d / 2 }|\alpha_{w} |^2  \\
\leq&~  4\eps+\frac{8-8\eps}{\kappa \cdot \kappa'\cdot d }\Big(1-\|(\ketbra{1}{1}_\mathsf{L}\otimes I_{\mathsf{PRT_L}})  \cdot {L}_{ \mathsf{LRT_L}}\ket\psi\|^2\Big)
\end{align*}
Using Lemma~\ref{lem:7.39}, 
\begin{align*}
\|
{\tilde O}_{p,\mathsf{SRT_S}}\ket\psi -
O_\mathsf{DR}(g_p)\ket\psi \|^2 \leq 4\eps+\frac{8-8\eps}{\kappa \cdot \kappa'\cdot d }\Big(1-\|(\ketbra{1}{1}_\mathsf{L}\otimes I_{\mathsf{PRT_L}})  \cdot {L}_{ \mathsf{LRT_L}}\ket\psi\|^2\Big). 
\end{align*}

\end{proof}

Then, we start to consider the implication of the validity checks in the first two rounds.
\begin{lemma}
\label{lem:mqt:v1:0}
Let 
\begin{align*}
\ket\phi=\ket{0}_{\mathsf{PRDS}}\otimes H^{\otimes l_\mathsf{S}}\ket{0}_\mathsf{T_{S,1}}\otimes H^{\otimes l_\mathsf{S}}\ket{0}_\mathsf{T_{S,2}}\cdots \otimes H^{\otimes l_\mathsf{S}}\ket{0}_\mathsf{T_{S, u}},
\end{align*}
then for $u\in[2]$, 
\begin{align*}
&~ \E_{(f_1,f_2)\sim D}[\|{\tilde U}_{u, F, f}\ket\phi -  O(g_{f_u})^{F_u}{U}_{u}{\tilde U}_{u-1, F, f}\ket\phi\|^2]\\
\leq&~ \Theta(\eps)+ \Theta\Big(1- \Pr[\text{ validity check passes }~|~F, u, D~]  \Big).
\end{align*}
\end{lemma}
\begin{proof}

Let $\ket{\phi'} = {U}_{u}{\tilde U}_{u-1, F, f}\ket{\phi}$, $ \ket{\phi''}=\ket{\phi'}\otimes \ket{0}_\mathsf{L}\otimes H^{\otimes l_\mathsf{L}}\ket{0}_\mathsf{T_{L}}$. Then
\begin{align*}
&~\Pr[\text{ validity check passes }~|~F, u~] \\
=&~ \E_{(f_1,f_2)\sim D} \|(\ketbra{1}{1}_\mathsf{L}\otimes I)  \cdot {L}_{ \mathsf{LRT_L}}(\ket{\phi''})\|^2.
\end{align*}

Moreover,  
\begin{align*}
&~ \E_{(f_1,f_2)\sim D}[\|{\tilde U}_{u, F, f}\ket\phi -  O(g_{f_u})^{F_u}{U}_{u}{\tilde U}_{u-1, F, f}\ket\phi\|^2]\\
=&~ \E_{(f_1,f_2)\sim D}[\|
{\tilde O}_{ f_u}^{F_u}
\ket{\phi'} -  O(g_{f_u})^{F_u}\ket{\phi'}\|^2]\\
=&~ \E_{(f_1,f_2)\sim D}[\|
{\tilde O}_{ f_u}^{F_u}
\ket{\phi''} -  O(g_{f_u})^{F_u}\ket{\phi''}\|^2]\\
\leq &~ \E_{(f_1,f_2)\sim D}[\Theta(\eps)+\Theta\Big(1-\|(\ketbra{1}{1}_\mathsf{L}\otimes I)  \cdot {L}_{ \mathsf{LRT_L}}\ket{\phi''}\|^2\Big)]\\
= &~ \Theta(\eps)+\Theta\Big(1-\Pr[\text{ validity check passes }~|~F, u,D~] \Big)\;,
\end{align*}
where the inequality is by Lemma~\ref{lem:7.40}.
\end{proof}

\begin{lemma}[Implication of the validity checks I]\label{lem:mqt:v1}
Let 
\begin{align*}
\ket\phi=\ket{0}_{\mathsf{PRDS}}\otimes H^{\otimes l_\mathsf{S}}\ket{0}_\mathsf{T_{S,1}}\otimes H^{\otimes l_\mathsf{S}}\ket{0}_\mathsf{T_{S,2}}\cdots \otimes H^{\otimes l_\mathsf{S}}\ket{0}_\mathsf{T_{S, u}},
\end{align*}
then for $u\in[2]$, 
\begin{align*}
&~ \E_{(f_1,f_2)\sim D}[\|{\tilde U}_{u, F, f}\ket\phi -  {U}_{u, F, f}\ket\phi\|^2]\\
\leq&~ \Theta(\eps)+ \sum_{v\leq u}\Theta\Big(1- \Pr[\text{ validity check passes }~|~F, u=v,D~]  \Big)
\end{align*}
\end{lemma}
\begin{proof}

This lemma follows from the triangle inequality and Lemma \ref{lem:mqt:v1:0}. 
\end{proof}

Then, we start to consider the implication of the validity checks in the last round.

\begin{lemma}\label{lem:mqt:v2:0}
Let $ G_k=[g_1, g_2,\cdots, g_n]^T \in \F_2^{n} \otimes \F_2^k$ be defined as in Definition \ref{def:3c}. 
Let 
\begin{align*}
\ket\phi=\ket{0}_{\mathsf{PRDS}}\otimes H^{\otimes l_\mathsf{S}}\ket{0}_\mathsf{T_{S,1}}\otimes H^{\otimes l_\mathsf{S}}\ket{0}_\mathsf{T_{S,2}}\cdots \otimes H^{\otimes l_\mathsf{S}}\ket{0}_\mathsf{T_{S, u}}\;.
\end{align*}
Then for $q'\in[2]$, 
\begin{align*}
&~ \E_{(f_1, f_2)\sim D}[\|{\tilde O}'_{q', {f_{q'}}}{\tilde U}_{2, F, f}\ket\phi -  {O}'_{q' }(g_{f_{q'}}){\tilde U}_{2, F, f}\ket\phi\|^2]\\
\leq&~ \Theta(\eps)+ \Theta\Big(1- \Pr[\text{ validity check passes }~|~F, q',u=3, D~]  \Big)
\end{align*}
\end{lemma}

\begin{proof}
Let $\ket{\phi'} = {U}'_{q'}{\tilde U}_{2, F, f}\ket{\phi}$, $ \ket{\phi''}=\ket{\phi'}\otimes \ket{0}_\mathsf{L}\otimes H^{\otimes l_\mathsf{L}}\ket{0}_\mathsf{T_{L}}$. Then by definition,
\begin{align*}
&~\Pr[\text{ validity check passes }~|~F, q',u=3, D~] \\
=&~ \E_{(f_1, f_2)\sim D} \|(\ketbra{1}{1}_\mathsf{L}\otimes I)  \cdot {L}_{ \mathsf{LRT_L}}(\ket{\phi''})\|^2.
\end{align*}
Moreover,  
\begin{align*}
&~ \E_{(f_1, f_2)\sim D}[\|{\tilde O}'_{q', {f_{q'}}}{\tilde U}_{2, F, f}\ket\phi -  {O}'_{q' }(g_{f_{q'}}){\tilde U}_{2, F, f}\ket\phi\|^2]\\
=&~ \E_{(f_1, f_2)\sim D}[\|{\tilde O}_{{f_{q'}}}\ket{\phi'} -  O(g_{f_{q'}})\ket{\phi'}\|^2]\\
=&~ \E_{(f_1, f_2)\sim D}[\|{\tilde O}_{{f_{q'}}}\ket{\phi''} -  O(g_{f_{q'}})\ket{\phi''}\|^2]\\
\leq &~ 
\E_{(f_1, f_2)\sim D}[\Theta(\eps)+\Theta\Big(1-\|(\ketbra{1}{1}_\mathsf{L}\otimes I)  \cdot {L}_{ \mathsf{LRT_L}}(\ket{\phi''})\|^2\Big)]\\
= &~ \Theta(\eps)+\Theta\Big(1-\Pr[\text{ validity check passes }~|~F, q',u=3, D~] \Big)\;,
\end{align*}
where the inequality is by Lemma~\ref{lem:7.40}.
\end{proof}

\begin{lemma}[Implication of the validity checks II]\label{lem:mqt:v2}

Let 
\begin{align*}
\ket\phi=\ket{0}_{\mathsf{PRDS}}\otimes H^{\otimes l_\mathsf{S}}\ket{0}_\mathsf{T_{S,1}}\otimes H^{\otimes l_\mathsf{S}}\ket{0}_\mathsf{T_{S,2}}\cdots \otimes H^{\otimes l_\mathsf{S}}\ket{0}_\mathsf{T_{S, u}},
\end{align*}
then for $q'\in[2]$, 
\begin{align*}
&~ \E_{(f_1, f_2)\sim D}[\|{\tilde O}'_{q', {f_{q'}}}{\tilde U}_{2, F, f}\ket\phi -  {O}'_{q' }(g_{f_{q'}}){U}_{2, F, f}\ket\phi\|^2]\\
\leq&~ \Theta(\eps)+ \sum_{u\in[3]}\Theta\Big(1- \Pr[\text{ validity check passes }~|~F, q',u, D~]  \Big).
\end{align*}
\end{lemma}

\begin{proof}
This lemma follows from the triangle inequality, Lemma \ref{lem:mqt:v1}, and Lemma \ref{lem:mqt:v2:0}. 
\end{proof}

\begin{lemma}[Implication of the consistency checks]\label{lem:mqt:c}
Let $ {U}_{i}$ for $i\in[2]$ be defined as in Definition \ref{def:mqt:canonical_prover}.
Let $\mathsf{P},\mathsf{R}$ be defined as in Definition \ref{def:mqt:canonical_prover}. 
Let $D$ be the sampling distribution of $ M,F$ in Definition \ref{def:mqt:verifier:c}. For any $M,F\in\{0, 1\}^2$ in the support of $D$. Let $ M_u=1$. 

Then
\begin{align*}
&~\E_{(f_1, f_2)\sim D}\|(U_{u\rightarrow 2, \mathsf{PR}} {O}_{\mathsf{DR}}(g_{f_u}) {U}_{u, F, f, \mathsf{DPR}}-{O}'_{u, \mathsf{DPR}}(g_{f_u}) {U}_{2, F, f, \mathsf{DPR}})\ket{0}_\mathsf{DPR}\|^2\\
\leq &~ \Theta(1-\Pr[\text{ consistency check passes }~|~M,F, D])\\
&~\quad+\sum_{ F,q',u}\Theta(1- \Pr[\text{ validity check passes }~|~F, q',u, D~])+\Theta(\eps).
\end{align*}

\end{lemma}

\begin{proof}

By a straightforward calculation,
\begin{align*}
&~\Pr[\text{ consistency check passes }~|~M,F, D]\\
=&~\E_{(f_1, f_2)\sim D}\sum_{a\in \{\pm 1\}}
\|  \frac{I + a\cdot  {\tilde O}'_{u, f_u}}{2} \cdot { U}_{u\rightarrow 2} \cdot \frac{I + a\cdot  {\tilde O}_{f_u}}{2}  \cdot {\tilde U}_{u, F, f} \ket{0}  \|^2\\
=&~\E_{(f_1, f_2)\sim D}\bra{0} \frac{2I + {\tilde U}_{u, F,f}^\dagger {\tilde O}_{f_u} U_{u\rightarrow 2}^\dagger {\tilde O}'_{u, f_u} {\tilde U}_{2, F, f}+{\tilde U}_{2, F, f}^\dagger {\tilde O}'_{u, f_u} U_{u\rightarrow 2} {\tilde O}_{f_u} {\tilde U}_{u, F, f} }{4} \ket{0} \\
=&~1-\E_{(f_1, f_2)\sim D}\frac{\|U_{u\rightarrow 2} {\tilde O}_{f_u} {\tilde U}_{u, F, f}\ket{0}-{\tilde O}'_{u, f_u} {\tilde U}_{2, F, f}\ket{0}\|^2 }{4}  .
\end{align*} 

Moreover,
\begin{align*}
&~\|U_{u\rightarrow 2} {\tilde O}_{f_u} {\tilde U}_{u, F, f}\ket{0}-{\tilde O}'_{u, f_u} {\tilde U}_{2, F, f}\ket{0}\|^2\\
\geq&~\|U_{u\rightarrow 2} {O}(g_{f_u}) {U}_{u, F, f}\ket{0}-{O}'_{u}(g_{f_u}) {U}_{2, F, f}\ket{0}\|^2\\
&~\quad-\|{O}'_{u}(g_{f_u}) {U}_{2, F, f}\ket{0}-{\tilde O}'_{u, f_u} {\tilde U}_{2, F, f}\ket{0}\|^2\\
&~\quad-\|U_{u\rightarrow 2} {\tilde O}_{f_u} {\tilde U}_{u, F, f}\ket{0}-U_{u\rightarrow 2} {O}(g_{f_u}) {U}_{u, F, f}\ket{0}\|^2\\
\geq&~ \|U_{u\rightarrow 2} {O}(g_{f_u}) {U}_{u, F, f}\ket{0}-{O}'_{u}(g_{f_u}) {U}_{2, F, f}\ket{0}\|^2\\
&~\quad-\Theta(\eps)- \sum_{ F,q',u}\Theta\Big(1- \Pr[\text{ validity check passes }~|~F, q',u, D~]  \Big).
\end{align*}

Therefore,
\begin{align*}
&~\E_{(f_1, f_2)\sim D}\|U_{u\rightarrow 2} {O}(g_{f_u}) {U}_{u, F, f}\ket{0}-{O}'_{u}(g_{f_u}) {U}_{2, F, f}\ket{0}\|^2\\
\leq &~ 4(1-\Pr[\text{ consistency check passes }~|~M,F, D])\\
&~\quad+\Theta(\eps)+ \sum_{ F,q',u}\Theta\Big(1- \Pr[\text{ validity check passes }~|~F, q',u, D~]  \Big).
\end{align*}

\end{proof}

\begin{lemma}[Implication of the anti-commuting tests]\label{lem:mqt:a}
Let $ {U}_{i}$ for $i\in[2]$ be defined as in Definition \ref{def:mqt:canonical_prover}.
Let $\mathsf{P},\mathsf{R}$ be defined as in Definition \ref{def:mqt:canonical_prover}. 

Then,
\begin{align*}
&~\E_{(m_1,m_2)\sim D}\|{O}'_{1}(m_1) {O}(m_2) U_{0 \rightarrow 2}\ket{0}-(-1)^{g_{m_1}\cdot g_{m_2}}{ O}(m_2)  U_{2} { O}(m_1) U_{1}\ket{0}\|^2 \\
\leq &~ \Theta(1-\Pr[\text{ anti-commuting test passes }~|~D])\\
&~\quad+\sum_{ F,q',u}\Theta(1- \Pr[\text{ validity check passes }~|~F, q',u, D~])+\Theta(\eps).
\end{align*}
\end{lemma}
\begin{proof}
By a straightforward calculation,
\begin{align*}
&~\Pr[\text{ anti-commuting test passes }~|~D]\\
=&~\E_{(m_1,m_2)\sim D}\sum_{a\in \{\pm 1\}}
\|  \frac{I +(-1)^{g_{m_1}\cdot g_{m_2}} \cdot a\cdot  {\tilde O}'_{1, m_1}}{2} \cdot {\tilde O}_{m_2} \cdot U_{2} \cdot \frac{I + a\cdot  {\tilde O}_{m_1}}{2} \cdot U_{1} \ket{0}  \|^2\\
=&~\E_{(m_1,m_2)\sim D}\bra{0} \frac{2I +(-1)^{g_{m_1}\cdot g_{m_2}}\cdot\Big(U_{0\rightarrow 2}^\dagger {\tilde O}_{m_2} {\tilde O}'_{1, m_1}  {\tilde O}_{m_2}  U_{2} {\tilde O}_{m_1} U_{1}+U_{1}^\dagger {\tilde O}_{m_1} U_{2}^\dagger {\tilde O}_{m_2} {\tilde O}'_{1, m_1} {\tilde O}_{m_2} U_{0 \rightarrow 2}\Big) }{4} \ket{0} \\
=&~1-\E_{(m_1,m_2)\sim D} \frac{ \|{\tilde O}'_{1, m_1} {\tilde O}_{m_2} U_{0 \rightarrow 2}\ket{0}-(-1)^{g_{m_1}\cdot g_{m_2}}{\tilde O}_{m_2}  U_{2} {\tilde O}_{m_1} U_{1}\ket{0}\|^2  }{4}  . 
\end{align*}

Moreover,
\begin{align*}
&~\|{\tilde O}'_{1, m_1} {\tilde O}_{m_2} U_{0 \rightarrow 2}\ket{0}-(-1)^{g_{m_1}\cdot g_{m_2}}{\tilde O}_{m_2}  U_{2} {\tilde O}_{m_1} U_{1}\ket{0}\|^2 \\
\geq &~ \|{O}'_{1}(m_1) {O}(m_2) U_{0 \rightarrow 2}\ket{0}-(-1)^{g_{m_1}\cdot g_{m_2}}{ O}(m_2)  U_{2} { O}(m_1) U_{1}\ket{0}\|^2 \\
&~\quad-\|{\tilde O}'_{1, m_1} {\tilde O}_{m_2} U_{0 \rightarrow 2}\ket{0}-{O}'_{1}(m_1) {O}(m_2) U_{0 \rightarrow 2}\ket{0}\|^2 \\
&~\quad-\|{\tilde O}_{m_2}  U_{2} {\tilde O}_{m_1} U_{1}\ket{0}-{ O}(m_2)  U_{2} { O}(m_1) U_{1}\ket{0}\|^2 \\
\geq&~ \|{O}'_{1}(m_1) {O}(m_2) U_{0 \rightarrow 2}\ket{0}-(-1)^{g_{m_1}\cdot g_{m_2}}{ O}(m_2)  U_{2} { O}(m_1) U_{1}\ket{0}\|^2\\
&~\quad-\Theta(\eps)- \sum_{ F,q',u}\Theta\Big(1- \Pr[\text{ validity check passes }~|~F, q',u, D~]  \Big).
\end{align*}

Therefore,
\begin{align*}
&~\E_{(m_1,m_2)\sim D}\|{O}'_{1}(m_1) {O}(m_2) U_{0 \rightarrow 2}\ket{0}-(-1)^{g_{m_1}\cdot g_{m_2}}{ O}({m_2})  U_{2} { O}(m_1) U_{1}\ket{0}\|^2\\
\leq &~ 4(1-\Pr[\text{ anti-commuting test passes }~|~D])\\
&~\quad+\Theta(\eps)+ \sum_{ F,q',u}\Theta\Big(1- \Pr[\text{ validity check passes }~|~F, q',u, D~]  \Big).
\end{align*}

\end{proof}

The following lemmas will be instrumental in proving the soundness of Theorem \ref{thm:mqt}.
\begin{lemma}
\label{lem:mqt:anti}

Let \begin{align*}
    \rho = U_{0\rightarrow 2}\ket{00}_{\mathsf{PRD}}\bra{00}_{\mathsf{PRD}} U_{0\rightarrow 2}^\dagger.
\end{align*}

Then
\begin{align*}
&~ \E_{(a, b)\sim D}\| O'_{1}(a)O'_{2}(b)-(-1)^{a\cdot b}O'_{2}(b)O'_{1}(a)\|^2_\rho  \\
\leq &~  \sum_{ M,F} \Theta(1-\Pr[\text{ consistency check passes }~|~M,F, D])\\
&~\quad+\Theta(1-\Pr[\text{ anti-commuting test passes }~|~D])\\
&~\quad+\sum_{ F,q',u}\Theta(1- \Pr[\text{ validity check passes }~|~F, q',u,D~])+\Theta(\eps).
\end{align*}
\end{lemma}
\begin{proof}

We have that, 
\begin{align*}
&~ \E_{(a, b)\sim D}\| O'_{1}(a)O'_{2}(b)-(-1)^{a\cdot b}O'_{2}(b)O'_{1}(a)\|^2_\rho \\
=&~\E_{(a, b)\sim D}\bra{00}_{\mathsf{PRD}}U_{0\rightarrow 2}^\dagger(h.c.)\cdot (O'_{1}(a)O'_{2}(b)-(-1)^{a\cdot b}O'_{2}(b)O'_{1}(a))U_{0\rightarrow 2}\ket{00}_{\mathsf{PRD}} \\
=&~\E_{(a, b)\sim D}\|O'_{1}(a)O'_{2}(b)U_{0\rightarrow 2}\ket{00}_{\mathsf{PRD}}-(-1)^{a\cdot b} O'_{2}(b)O'_{1}(a)U_{0\rightarrow 2}\ket{00}_{\mathsf{PRD}}\|^2 .
\end{align*}

We use $ \ket{\alpha(a,b)} \approx_{\delta, D} \ket{ \beta(a,b)}$ to denote $\E_{(a,b)\sim D} \|\ket{\alpha(a,b)}- \ket{\beta(a,b)}\|^2 \leq \Theta (\delta)$.

By Lemma \ref{lem:mqt:prob}, \ref{lem:mqt:c}, and \ref{lem:mqt:a},
we establish the following relationship:
\begin{align*}
&~ O'_{1}(a)O'_{2}(b)U_{0\rightarrow 2}\ket{0}\\
\approx_{\delta, D} &~ O'_{1}(a)O(b) U_{2} U_1\ket{0}\\
\approx_{\delta, D} &~ (-1)^{a\cdot b} O(b) U_{2} O(a) U_1\ket{0}\\
\approx_{\delta, D} &~ (-1)^{a\cdot b} O'_2(b) U_{2} O(a) U_1\ket{0}\\
\approx_{\delta, D} &~ (-1)^{a\cdot b} O'_2(b) O'_1(a) U_{0\rightarrow 2}\ket{0}.
\end{align*}

Therefore,  
\begin{align*}
\E_{(a, b)\sim D}\|O'_{1}(a)O'_{2}(b)U_{0\rightarrow 2}\ket{00}_{\mathsf{PRD}}-(-1)^{a\cdot b} O'_{2}(b)O'_{1}(a)U_{0\rightarrow 2}\ket{00}_{\mathsf{PRD}}\|^2 \leq \Theta(\delta).
\end{align*}

\end{proof}

The following lemmas will be instrumental in proving the completeness of Theorem \ref{thm:mqt}.

\begin{lemma}[The honest prover in the canonical prover’s form]
\label{lem:mqt:hp_cp}
Let 
$$ R_1 = \prod_{p\in[n]}P((Z(g_p))_{\mathsf{P}}, \mathsf{R}_p), R_2 = \prod_{p\in[n]}P((X(g_p))_{\mathsf{P}}, \mathsf{R}_p). $$

The prover in Definition \ref{def:mqt:honest_prover} is the prover in Definition \ref{def:mqt:canonical_prover} with:
\begin{align*}
U_1 =&~ R_1,\\
U_2 =&~ R_2\cdot R_{1}^{-1},\\
U'_{1}=&~ R_1\cdot R_{2}^{-1},\\
U'_{2}=&~ I.
\end{align*}
\end{lemma}
\begin{proof}
This lemma follows from a straightforward calculation.
\end{proof}

\begin{lemma}[Properties of the honest prover]\label{lem:mqt:hp_properties}
For $i\in [2]$, let $ U_i$ be defined as in Lemma \ref{lem:mqt:hp_cp}. For $q'\in [2]$, let $ U'_{q'}$ be defined as in Lemma \ref{lem:mqt:hp_cp}. 
Then:

\begin{itemize}
\item For any $u\in[2]$,
\begin{align*}
{\tilde O}'_{u, j} U_{u\rightarrow 2} = U_{u\rightarrow 2} {\tilde O}_{j}.
\end{align*}
\item Moreover, 
\begin{align*}
{\tilde O}'_{1, j_1} {\tilde O}_{j_2} U_{2} = (-1)^{g_{j_1}\cdot g_{j_2}} {\tilde O}_{j_2} U_{2} {\tilde O}_{j_1}
\end{align*}
\end{itemize}
\end{lemma}

\begin{proof}
Let $R_i$ be defined as in Lemma \ref{lem:mqt:hp_cp}. 

{\bf Part I:}
It is sufficient to show that
\begin{align*}
R_2^{-1} {\tilde O}'_{u, j} R_2 = R_u^{-1} {\tilde O}_j R_u,
\end{align*}
which follows from 
\begin{align*}
R_2^{-1} {\tilde O}'_{1, j} R_2=&~ (Z(g_j))_{\mathsf{P}}, \\
R_2^{-1} {\tilde O}'_{2, j} R_2=&~ (X(g_j))_{\mathsf{P}}, 
\end{align*}
and
\begin{align*}
R_1^{-1} {\tilde O}_j R_1 =&~ (Z(g_j))_{\mathsf{P}}, \\
R_2^{-1} {\tilde O}_j R_2 =&~ (X(g_j))_{\mathsf{P}}.
\end{align*}

{\bf Part II:}
It is sufficient to show that
\begin{align*}
R_1^{-1} {\tilde O}_{j_1}  R_1 R_2^{-1} {\tilde O}_{j_2}  R_2 = (-1)^{g_{j_1}\cdot g_{j_2}}R_2^{-1} {\tilde O}_{j_2}  R_2 R_1^{-1} {\tilde O}_{j_1}  R_1,
\end{align*}
which follows from 
\begin{align*}
(Z(g_{j_1}))_{\mathsf{P}} (X(g_{j_2}))_{\mathsf{P}}  = (-1)^{g_{j_1}\cdot g_{j_2}}(X(g_{j_2}))_{\mathsf{P}} (Z(g_{j_1}))_{\mathsf{P}} .
\end{align*}
\end{proof}

\begin{proof}[Proof of Theorem \ref{thm:mqt}]

\hfill 

{\bf Completeness:}
For $i\in [2]$, let $ U_i$ be defined as in Lemma \ref{lem:mqt:hp_cp}. For $q'\in [2]$, let $ U'_{q'}$ be defined as in Lemma \ref{lem:mqt:hp_cp}. 
Let $D$ be the sampling distribution of $ M,F$ in Definition \ref{def:mqt:verifier:c}.

By Definition \ref{def:mqt:honest_prover} and Definition \ref{def:mqt:verifier}, the verifier’s validity checks pass with probability $1$. 

Following the proof of Lemma \ref{lem:mqt:c},
\begin{align*}
\Pr[\text{ consistency check passes }~|~M,F]=&~1-\E_{(f_1, f_2)\sim D}\frac{\|U_{u\rightarrow 2} {\tilde O}_{m_u} {\tilde U}_{u, F, f}\ket{0}-{\tilde O}'_{u, m_u} {\tilde U}_{2, F, f}\ket{0}\|^2 }{4}\;.
\end{align*}
Following the proof of Lemma \ref{lem:mqt:a},
\begin{align*}
\Pr[\text{ anti-commuting test passes }]=&~1-\E_{(m_1,m_2)\sim D} \frac{ \|{\tilde O}'_{1, m_1} {\tilde O}_{m_2} U_{0 \rightarrow 2}\ket{0}-(-1)^{g_{m_1}\cdot g_{m_2}}{\tilde O}_{m_2}  U_{2} {\tilde O}_{m_1} U_{1}\ket{0}\|^2  }{4} 
\end{align*}
Therefore, by Lemma \ref{lem:mqt:prob}, we only need to prove that for any $M,F\in\{0, 1\}^2$ in the support of $D$, letting $ M_u=1$,
\begin{align*}
{\tilde O}'_{u, m_u} U_{u\rightarrow 2} =U_{u\rightarrow 2} {\tilde O}_{m_u},
\end{align*}
and 
\begin{align*}
{\tilde O}'_{1,m_1} {\tilde O}_{m_2} U_{0 \rightarrow 2} = (-1)^{g_{m_1}\cdot g_{m_2}}{\tilde O}_{m_2} U_{2} {\tilde O}_{m_1} U_{1}.
\end{align*}
Both parts follow directly from Lemma \ref{lem:mqt:hp_properties}.

{\bf Soundness:}
The soundness follows from Lemma \ref{lem:mqt:anti}. 
\end{proof}

%% file: sqiop.tex
\section{Strong Quantum IOPs with Logarithmic-Depth Honest Prover}
\label{sec:sqiop}

In this section we present the complete construction of our strong quantum IOP. In addition to the many-qubits test introduced in Section~\ref{sec:mqt}, we also introduce the energy test and the energy consistency test. Furthermore, we separate the Pauli measurement component from the energy consistency test for clarity.

The high-level idea is as follows: the verifier randomly samples a term from the amplified Clifford Hamiltonian and instructs the prover to measure according to it. However, the prover may not follow the honest strategy, so we must correlate the prover’s measurement strategy with a reference that is known to be close to the honest strategy; this is achieved via the many-qubits test.

Note that a $5$-local Clifford projection is not itself a Pauli measurement. To address this, we decompose each projection into $5$-local Pauli measurements. Since the Hamiltonian is applied to $N$ copies of the original witness, the honest prover will obtain $5N$ measurement outcomes. As the verifier cannot read all $5N$ bits, we follow the approach in Section~\ref{sec:qiopepr}, using a PCPP proof to certify that the measurement result would cause the verifier to accept.

To test the correctness of the $5N$ outcomes, we draw inspiration from the equality test \cite{bk97}. Specifically, we compute the inner product of the $5N$-bit string with a random mask. 
For the honest prover, this inner product corresponds to a Pauli measurement obtained by multiplying the $5N$ local Pauli operators selected according to the random mask.
If any of the $5N$ individual outcomes deviate from honest measurement results, the resulting inner product is unlikely to match the actual measurement outcome of the composed Pauli observable. This composed measurement can be obtained from the many-qubits test.

To obtain Pauli measurements we do not use the observables derived from the third round of interaction in the many-qubits tests. Instead, we rely on the prover’s messages in the first two rounds: the purified measurement results in the $Z$ basis and the $X$ basis. Measuring an observable is equivalent to applying a controlled observable with a control qubit in the state $\ket{+}$ and measuring the control qubit in the Hadamard basis. To measure an observable of the form $X(a)Z(b)Y(c)$, we combine a controlled-$Z(b+c)$ from the first round, a controlled-$X(a+c)$ from the second round, and apply the phase factor $i^c$ directly via the verifier.
We note that our approach to measuring $Y$ observables is significantly simpler than the method used in prior work, such as~\cite{cgjv19}.

\begin{theorem}[Restatement of Theorem~\ref{thm:sqiop}]
\label{thm:restated_sqiop}
$\QMA \subseteq \SQIOP(O(1), O(\exp \circ \poly), O(1))$.
\end{theorem}

\begin{remark}
For the honest prover, the quantum circuit can be implemented with logarithmic depth. In the multi-qubit test, the prover performs a purified measurement of their witness in either the $Z$ or $X$ basis. Switching between these bases requires only a single layer of Hadamard gates. Using unbounded fan-out CNOT gates, the witness can be copied in the computational basis an exponential number of times. An inner product is then computed for each copy with the corresponding index, resulting in a Hadamard code encoding of the witness. The deepest part of this computation, the inner product, requires only logarithmic depth.

The strong quantum IOP protocol is based on this multi-qubit test. In the energy consistency test, the prover’s behavior for Pauli measurements matches that of the multi-qubit test. The energy test requires a measurement of the witness, which consists of $N$ copies of the original witness. Each copy is measured using a local Clifford projection, which is also constant-depth.

The other components of the energy consistency test involve providing a PCPP proof for the inner product with a random mask. Similarly, in the energy test, the prover supplies an error-correcting encoding of the measurement result along with a PCPP proof that convinces the verifier to accept. These steps are classical and therefore their computational complexity (beyond being polynomial time) is not considered.
\end{remark}

The section is organized as follows.  In Section~\ref{sec:siop:xz}, we present a construction that allows the verifier to perform Pauli measurements, where the prover behaves identically to the many-qubits test. This forms a component of the energy consistency test, which checks that the measurement used in the energy test matches that of the many-qubits test. Section~\ref{sec:siop:e} presents the full construction of the verifier and honest prover behavior in both the energy test and energy consistency test. In Section~\ref{sec:siop:formula}, we describe the general strategy of a malicious prover in a without-loss-of-generality form. Finally, in Section~\ref{sec:siop:main}, we prove the main theorem of this section.

\subsection{Measurement in the Pauli bases}
\label{sec:siop:xz}
We define the verifier’s strategy for measuring observables of the form $X(a)Z(b)Y(c)$. We begin by describing the verifier’s general actions in the first two rounds of the protocol, without specifying the parameters. Then, we detail how the verifier instantiates this template to implement measurements of $X(a)Z(b)Y(c)$.

\begin{definition}[Template of the verifier’s actions for measurements in the first two rounds]\label{def:mqt:verifier:m}
Let $\mathsf{R} =(\mathsf{R}_{j})_{j\in[n]}$ be a register. Let $ \mathsf{S}$ be a register held by the verifier and initialized with $ \ket{0}$. Let $ \mathsf{M}$ be a register held by the verifier. Let $ \mathsf{T_S} $ be a register held by the verifier and initialized with $H^{\otimes l_{\mathsf{S}}} \ket{0}$. Let $m\in [n]$. Let $\tilde{O}_{m,\mathsf{SRT_S}}$ be defined as in Definition~\ref{def:mqt:read_self_corr}.
\begin{mdframed}
\begin{itemize}
\item \textnormal{The prover's first step:} 
The prover sends the verifier the register $\mathsf{R}$.
\item \textnormal{The verifier's first step:} 
The verifier receives the register $\mathsf{R}$ from the prover. 
The verifier measures $\mathsf{T_S}$. Then, the verifier
applies $C({\tilde O}_{m,\mathsf{SRT_S}}, \mathsf{M})$. 
The verifier sends the prover the register $\mathsf{R}$. 
\item \textnormal{The prover's second step:} 
The prover receives the register $\mathsf{R}$ from the verifier. 
\end{itemize}
\end{mdframed}
\end{definition}

\begin{definition}[Measure $X(a)Z(b)Y(c)$]\label{def:mqt:verifier:measure} 
Let $ a,b,c\in \{0, 1\}^k$ such that $ a\cap b = b\cap c = c\cap a = \emptyset $. 
Let $\mathsf{R} =(\mathsf{R}_{j})_{j\in[n]}$ be a register. Let $ \mathsf{S}$ be a register held by the verifier and initialized with $ \ket{0}$. Let $ \mathsf{M}$ be a register held by the verifier and initialized with $\ket{+}$. 
Let $ \mathsf{T_{S, i}} $ for $i\in [3]$ be a register held by the verifier and initialized with $H^{\otimes l_{\mathsf{S}}} \ket{0}$. 
The verifier's actions are defined as follows: 
\begin{mdframed}
\begin{itemize}
\item The verifier sets $m_1 = b+c, m_2 = a+c$.
\item The verifier repeats the verifier's actions in Definition \ref{def:mqt:verifier:block} twice. In the $i$-th time, the verifier acts as in Definition \ref{def:mqt:verifier:block} with $m=m_i$.
\item In the end, the verifier applies the phase gate $|c|$ times as $S^{|c|}_{\mathsf{M}}$. Then, the verifier applies $ H_{\mathsf{M}}$, measures the register $\mathsf{M}$ under $Z_\mathsf{M}$, and outputs the measurement result. 
\end{itemize}

\end{mdframed}
\end{definition}

\subsection{Energy tests and energy consistency tests}
\label{sec:siop:e}

In this section, we describe the actions of the verifier and the prover in the energy test and the energy consistency test.

To perform the energy test, the verifier begins by sending the prover the indices $l_i \in [m]$, for $i \in [N]$, corresponding to the Hamiltonian terms the verifier wishes to test. The honest prover is expected to return the measurement outcomes of all observables $O_{l_i,j}$ for $i \in [N]$, $j \in [5]$, along with a PCPP proof certifying that these outcomes will lead the verifier to accept.

However, since the prover may not be honest, the energy consistency test is introduced to ensure integrity. In this test, the verifier randomly samples a binary mask, and the prover is expected to compute the inner product between the measurement outcomes and the mask. Classically, if even a single reported outcome is incorrect, the entire inner product will be incorrect with probability exactly $1/2$.

Importantly, the verifier can directly measure this inner product as a single Pauli observable, since it corresponds to the product of all $O_{l_i,j}$ with $l_i,j$ under the mask. This measurement can be performed in the first two rounds of the protocol, allowing the verifier to compare the result with the prover’s claimed value.

We begin this section by recalling the relevant PCPP construction, and then present the strategies of the honest prover and verifier in both tests.

\begin{definition}[Energy tests and energy consistency tests, the honest prover’s actions in the last round]
\label{def:mqt:hp_last_energy}
Let $\mathsf{P}$ be a register of $k$
qubits. 
Let $p,q$ be as in Theorem \ref{thm:clifford_ham}, $s,N$ be as in Lemma \ref{lem:c_ham_amp}. 
Use the notation for $H$, $O_{i,j}$,  and $a_{i,j}, b_{i,j}, c_{i,j}$ from Definition~\ref{def:mqdef}
Let $$ C(x) = \bigoplus_{i\in [N]} \big(1 \oplus \bigvee_{j\in [5]} x_{i,j} \big) ,$$ 
where $x=(x_{i,j})_{i \in [N], j\in[5]}$. Let $E$ be the encoding circuit of the codes in Theorem \ref{thm:ltc}, $D$ the decoding circuit of the codes in Theorem \ref{thm:ltc}, and $T$ the  code tester 
circuit of the codes in Theorem \ref{thm:ltc}.

\begin{mdframed}
\begin{itemize}
\item The verifier sends the prover $l_i\in[m]$ for $i\in[N]$.
\item For $i\in [N],j\in [5]$, the honest prover measures $\mathsf{P}$ under observable 
\begin{align*}
\underbrace{I\otimes I \otimes \cdots \otimes I}_{i-1} \otimes O_{l_i, j} \otimes \underbrace{I\otimes I \otimes \cdots \otimes I}_{N-i},
\end{align*}
and stores the  measurement result in $r_{i,j}$. Let \begin{align*}
    r=(r_{i,j})_{i\in [N], j\in [5]}.
\end{align*}
\item The honest prover sends the verifier $E(r)$, $C(r)$, and $w$, where $w$ is a witness such that the verifier in Theorem \ref{thm:pcpp} with explicit input the circuit $C(D(\cdot)) + C(r) + 1$ and implicit input the input of the circuit $E(r)$ accepts the PCPP witness $w$.
\item The verifier sends the prover $ t\in_R\F_2^{5N}$.
\item Let $ C'(r)=r\cdot t$. The prover sends the verifier $C'(r)$ and $w'$, where $w'$ is a witness such that the verifier in Theorem \ref{thm:pcpp} with explicit input $V(C'(D(\cdot )) + C'(r) + 1)$ and implicit input $E(r)$ accepts the witness $w'$.
\end{itemize}

\end{mdframed}

\end{definition}
\begin{remark}
If $C(D(e')) + r' + 1 = 1$, which means that $C(D(e')) = r'$, then the verifier accepts with probability $1$ using the explicit input, the circuit $C(D(\cdot)) + r' + 1$, and the implicit input, the input of the circuit $e'$, along with the PCPP witness $w$.
\end{remark}

\begin{definition}[Energy tests and energy consistency tests, honest prover]\label{def:mqt:hp_energy} 
The honest prover's actions are defined as follows: 
\begin{mdframed}
\begin{itemize}
\item The honest prover acts as in Definition \ref{def:mqt:honest_prover} for the first two rounds. 
\item The verifier sends the prover a symbol ``EnergyTests" or ``QubitsTests". 
\item If the symbol is ``QubitsTests", the honest prover continues as in Definition \ref{def:mqt:honest_prover}. If the symbol is ``EnergyTests", the honest prover acts as in Definition \ref{def:mqt:hp_last_energy}.
\end{itemize}

\end{mdframed}

\end{definition}

\begin{definition}[Energy tests and energy consistency tests, the verifier’s actions in the last round]\label{def:mqt:verifier_e_last}
Let $ c,e\in \{0, 1\}$. Let $ t\in\F_2^{5N}$. Let $ l_i\in [m]$ for $i\in [N]$. 
Let $E$ be the encoding circuit of the codes in Theorem \ref{thm:ltc}, $D$ be the decoding circuit of the codes in Theorem \ref{thm:ltc}, $T$ be the  code tester 
circuit of the codes in Theorem \ref{thm:ltc}.

\begin{mdframed}
\begin{itemize}
\item The verifier sends $l_i\in[m]$ for $i\in[N]$ to the prover. 
\item The prover sends the verifier $M, o, w$: the implicit input of the PCPP, the value supposed to be $C(D(M))$ (which is also used to build the explicit input of the PCPP), and the witness of the PCPP.
\item The verifier sets $a_1=1$ if and only if $o=0$ and the verifier in Theorem \ref{thm:pcpp}, when run a constant number of times, accepts each time with the following parameters:
explicit input $x = C(D(\cdot)) + o + 1 \wedge T(\cdot)$, implicit input $ y = M$, PCPP proof  $\pi= w$.
\item The verifier sends $t$ to the prover.
\item The prover sends the verifier $o',w' $.  The verifier sets $a_2=1$ if and only if $o'=e$ and the verifier in Theorem \ref{thm:pcpp}, when run a constant number of times, accepts each time with the following parameters:
$x = C'(D(\cdot)) + o' + 1 \wedge T(\cdot)$, $y = M$, $\pi= w'$. 
\item If $c=0$, the verifier accepts iff $a_1=1$. If $ c=1$, the verifier accepts iff $a_2=1$.
\end{itemize}
\end{mdframed}

\end{definition}

\begin{definition}[Energy tests and energy consistency tests, verifier]\label{def:mqt:verifier_energy} 
Let $p,q$ be as in Theorem \ref{thm:clifford_ham}, $s,N$ be as in Lemma \ref{lem:c_ham_amp}. Let $H$, $H_i$, and $O_{i,j}$ be as in Definition~\ref{def:mqt:hp_last_energy}.

The verifier's actions are defined as follows: 

\begin{mdframed}
\begin{enumerate}
\item The verifier samples $ t\in_R\F_2^{5N}$ and samples $l_i\in_R[m]$ for $i\in [N]$. Let
\begin{align*}
O = \prod_{i\in [N], j\in [5]} O_{l_i,j}^{t_{i,j}}\;, 
\end{align*}
with Pauli representation
\begin{align*}
O =  (-1)^{d} X(a) Z(b) Y(c)\;.  
\end{align*}
\item \textnormal{ (Qubits tests:)} With probability $1/3$, the verifier acts as in Definition \ref{def:mqt:verifier} and sends the prover a symbol ``QubitsTests'' after the end of the first two rounds. 
\item \textnormal{ (Energy tests:)} With probability $1/3$, the verifier repeats the verifier’s action in Definition \ref{def:mqt:verifier:block} for $2$ times. Each time, the verifier acts as in Definition \ref{def:mqt:verifier:block} 
with $M = 0, F = 0, m = 0, f = 0$. In the end, the verifier sends the prover a symbol ``EnergyTests'' and acts as in Definition \ref{def:mqt:verifier_e_last} with $ c=0,e=0$. 
\item \textnormal{ (Energy consistency tests:)} With probability $1/3$, the verifier acts as in Definition \ref{def:mqt:verifier:measure} with $a=a, b=b,c=c$. Let $e$ be the output. In the end, the verifier sends the prover a symbol ``EnergyTests'' and acts as in Definition \ref{def:mqt:verifier_e_last} with $ c=1, e=e+d, t=t, l_i=l_i$ for $i\in [N]$. 
\end{enumerate}
\end{mdframed}

\end{definition}

\subsection{General prover's actions}
\label{sec:siop:formula}

To state the main theorem of this section, we first introduce the behavior of a malicious prover.
\begin{definition}[Energy tests, the canonical prover's actions]
\label{def:mqt:canonical_prover_e}
Let $\mathsf{P}$ be a register. 
Let $h = |E(1^{5N})|$ be
the output length of the encoding circuit of the codes in Theorem \ref{thm:ltc}. 
Let $g$ denote the length of the witness used by the verifier in Theorem \ref{thm:pcpp}, where the circuit has input size $h$.
For $l\in [m]^N$, let $P_{l}=\{P_{l,M,o,w}\}_{M\in \F_2^h,o\in\F_2,w\in\F_2^{g}}$ be a projective measurement acting on $\mathsf{P}$. 
\begin{mdframed}
\hfill 

The canonical prover's actions in the last round  of the energy tests are defined as follows: 
\begin{itemize}
\item \textnormal{The verifier's first step:} 
The verifier sends $l_i\in [m]$ for $i\in[N]$ to the prover.
\item \textnormal{The prover's first step:} 
The prover measures according to $ {P}_{l,\mathsf{P}}$. 
The prover sends the verifier the measurement result $M$ and $o,w$.
\item \textnormal{The verifier's second step:} 
The verifier sends the prover $t\in \F_2^{5N}$. 
\item \textnormal{The prover's second step:} 
The prover sends verifier $o',w'$. 
\end{itemize}
\end{mdframed}
\end{definition}

\subsection{Main theorem}
\label{sec:siop:main}

\begin{lemma}[Implication of the energy consistency tests]
\label{lem:mqt:ec}

Let 
\begin{align*}
P'_{l,t}=\sum_{M,o,w} 1_{D(M)\cdot t=0} P_{l,M,o,w} - \sum_{M,o,w} 1_{D(M)\cdot t=1} P_{l,M,o,w}.
\end{align*}

If for any $ M,F,D,u,q'$,
\begin{align*}
\Pr[\text{ consistency check passes }~|~M,F, D]\geq&~ 1- \eps'\\
\Pr[\text{ validity check passes }~|~F, q',u, D~] \geq&~ 1-\eps'\\
 \Pr[\text{ anti-commuting test passes }~|~D]\geq&~ 1-\eps'.
\end{align*}

Then 
\begin{align*}
&~ \E_{l\in [m]^N, t\in \F_2^{5N}}
\|(
P'_{l,t}-
V^\dagger ((-1)^{d}\cdot  X(a)Z(b)Y(c)\otimes I) V
)U_{0\rightarrow 2}\ket{0} \|^2\\
\leq &~\Theta(\eps+\delta+\eps'),
\end{align*}
where $a,b,c,d$ are defined as in Definition \ref{def:mqt:verifier_energy}, $\eps$ is the self-correction parameter from Theorem~\ref{thm:3c_code} and $\delta$ is the PCPP proximity parameter from Theorem~\ref{thm:pcpp}. 
\end{lemma}

\begin{proof}
By Lemma \ref{lem:mqt:c},
\begin{align*}
\E \|{O}({a+c}) U_2 {O}({b+c}) U_1\ket{0} - {O}'_2({a+c}) U_2 {O}({b+c}) U_1\ket{0}\|^2 \leq \Theta(\eps+\eps'),
\end{align*}
and
\begin{align*}
\E \| {O}'_2({a+c}) U_2 {O}({b+c}) U_1\ket{0} -  {O}'_2({a+c}) {O}'_1({b+c})  U_2 U_1\ket{0} \|^2 \leq \Theta(\eps+\eps').
\end{align*}
By Lemma \ref{lem:mqt:v1}, 
\begin{align*}
\E\|{O}({a+c}) U_2 {O}({b+c}) U_1\ket{0}-{\tilde O}_{a+c} U_2 {\tilde O}_{b+c} U_1\ket{0} \|^2 \leq \Theta(\eps+\eps'). 
\end{align*}
Therefore, by combining the three inequalities above and the triangle inequality
\begin{align}
\E\|  {O}'_2({a+c}) {O}'_1({b+c})  U_2 U_1\ket{0}  -{\tilde O}_{a+c} U_2 {\tilde O}_{b+c} U_1\ket{0} \|^2 \leq \Theta(\eps+\eps')\;.\label{eq:mqt-1}
\end{align}
To analyze the energy consistency test, note that according to Definition \ref{def:mqt:verifier:measure} and Definition \ref{def:mqt:verifier_energy}, during the energy consistency test, the verifier first measure $X(a)Z(b)Y(c)$ by leveraging the first two rounds. In the rest of the proof, instead of $e$, we use $s$ to denote the measurement result.
\begin{align*}
\frac{1}{\sqrt{2}}(U_2  U_1\ket{+}_{\mathsf{M}}\ket{0}+\i^{c} \cdot & {\tilde O}_{a+c} U_2 {\tilde O}_{b+c} U_1\ket{-}_{\mathsf{M}}\ket{0} )\\
=&~ \frac{1}{{2}}(U_2  U_1\ket{0}+\i^{c} \cdot {\tilde O}_{a+c} U_2 {\tilde O}_{b+c} U_1\ket{0} )\ket{0}_{\mathsf{M}}\\
&~\qquad +\frac{1}{{2}}(U_2  U_1\ket{0}-\i^{c} \cdot {\tilde O}_{a+c} U_2 {\tilde O}_{b+c} U_1\ket{0} )\ket{1}_{\mathsf{M}}.
\end{align*}
We have that
\begin{align}
\Pr[\text{energy}&\text{ consistency test passes}] \notag\\
=& \E_{\substack{l\in [m]^N\\ t\in \F_2^{5N}}}\sum_{\substack{s\in \F_2\\M,o,w}}  \Pr[\text{verifier accepts}|l,t,M,o',w'] \cdot \Big\| P_{l,M,o,w} \frac{U_2  U_1\ket{0}+(-1)^{s}\i^{c} \cdot {\tilde O}_{a+c} U_2 {\tilde O}_{b+c} U_1\ket{0}  }{2}\Big\|^2\notag\\
\leq &\E_{\substack{l\in [m]^N\\ t\in \F_2^{5N}}}\sum_{\substack{s\in \F_2\\M,o,w}} \Pr[\text{verifier accepts}|l,t,M,o',w'] \notag\\
&\qquad\cdot \Big\| P_{l,M,o,w} \frac{U_2  U_1\ket{0}+(-1)^{s}\i^{c} \cdot   {O}'_2({a+c}) {O}'_1({b+c})  U_2 U_1\ket{0} }{2}\Big\|^2+\Theta(\eps+\delta)\;,\label{eq:mqt-2}
\end{align}
where the last step follows from~\eqref{eq:mqt-1}.

Let $S = \{x | C'(D(x)) = s+d \wedge T(x) = 1\}$. 
If $ dist(M,S) > \delta$,
\begin{align*}
\Pr[a_2=1|l,t,M,o',w'] \leq  \Pr[\text{verifier accepts}|l,t,M,o'=s+d,w'] \leq  \eps. 
\end{align*}
If $ dist(M,S) \leq \delta$,
\begin{align*}
\Pr[a_2=1|l,t,M,o',w'] \leq  1 = 1_{C'(D(M))=s+d}. 
\end{align*}
As a result,
\begin{align*}
\Pr[\text{verifier accepts}|l,t,M,o',w'] \leq 1_{D(M)\cdot t=s+d} +\eps.
\end{align*}
Plugging back into~\eqref{eq:mqt-2}, 
\begin{align}
\Pr[\text{energy}&\text{ consistency test passes}]\leq  \Theta(\eps+\delta)\notag\\
 & + E_{\substack{l\in [m]^N\\ t\in \F_2^{5N}}}\sum_{\substack{s\in \F_2\\M,o,w}} \Big\|(\sum_{M,o,w} 1_{D(M)\cdot t=s} P_{l,M,o,w}) \frac{I + (-1)^{d+s}\cdot \i^{c} \cdot O'_2(a+c) O'_1(b+c) }{2}
U_{0\rightarrow 2}\ket{0} \Big\|^2 \;.
\label{eq:mqt-3}
\end{align}
Let 
\begin{align*}
P'_{l,t}=\sum_{M,o,w} 1_{D(M)\cdot t=0} P_{l,M,o,w} - \sum_{M,o,w} 1_{D(M)\cdot t=1} P_{l,M,o,w}
\end{align*}
For any matrix $ R$,
\begin{align*}
\sum_{s\in\F_2}\Big\| \frac{I+(-1)^s P'_{l,t}}{2}\cdot \frac{I + (-1)^{s} R}{2}
U_{0\rightarrow 2}\ket{0} \Big\|^2 = 1- \frac{1}{4}\big\|P'_{l,t}U_{0\rightarrow 2}\ket{0} - R U_{0\rightarrow 2}\ket{0} \big\|^2,
\end{align*}

Thus by taking 
\begin{align*}
R = (-1)^{d}\cdot \i^{c} \cdot O'_2(a+c) O'_1(b+c),
\end{align*}
and Eq.~\eqref{eq:mqt-3}, we have that
\begin{align}
&~ \E_{l\in [m]^N, t\in \F_2^{5N}}
\big\|(
P'_{l,t}-
{(-1)^{d}\cdot \i^{c} \cdot O'_2(a+c) O'_1(b+c) }
)U_{0\rightarrow 2}\ket{0} \big\|^2\notag\\
\leq &~\Theta(1-\Pr[\text{energy consistency test passes}] )+\Theta(\eps+\delta). 
\label{eq:mqt-5}
\end{align}

Furthermore, according to Lemma \ref{lem:mqt:anti} and Definition \ref{def:hermitian_prover}, we can apply Theorem \ref{thm:ourqubits}, 
\begin{align}
\E \|{(-1)^{d}\cdot \i^{c} \cdot O'_2(a+c) O'_1(b+c) }
U_{0\rightarrow 2}\ket{0} -V^\dagger ((-1)^{d}\cdot  X(a)Z(b)Y(c)\otimes I) VU_{0\rightarrow 2}\ket{0}\|^2 \leq \Theta(\delta).
\label{eq:mqt-4}
\end{align}
As a result of Eq.~\eqref{eq:mqt-5} and Eq.~\eqref{eq:mqt-4},
\begin{align*}
&~ \E_{l\in [m]^N, t\in \F_2^{5N}}
\|(
P'_{l,t}-
V^\dagger ((-1)^{d}\cdot  X(a)Z(b)Y(c)\otimes I) V
)U_{0\rightarrow 2}\ket{0} \|^2\\
\leq &~\Theta(\eps+\delta+\eps'). 
\end{align*}

\end{proof}

The following lemma is a quantum analogue of the classical equality test via inner product with a random mask.

\begin{lemma}[Implications of using a random mask]
\label{lem:q_sound_had}
Let $ d' > d$. Let $ \{P_s\}_{s\in \F_2^n}$ be an orthogonal projective measurement on $ \C^{d}$.
Let $ O_i$ be observables 
on $\C^{d'}$ for $i\in [n]$. Let $ V:\C^{d}\rightarrow\C^{d'}$ be an isometry. 
For $t\in \F_2^n$, let
\begin{align*}
O^{(t)} = \prod O^{t_i}_i.
\end{align*}
Moreover, let 
\begin{align*}
P^{(t)}=\sum_s (-1)^{st}\cdot P_s.
\end{align*}
Then
\begin{align*}
\sum_s \Big\|P_s - V^\dagger \prod_{i\in [n]} \frac{I+(-1)^{s_i}O_i}{2} V\Big\|_\rho^2 \leq 
\E_t \|P^{(t)}- V^\dagger O^{(t)} V\|_\rho^2 +1 - \E_t \|V^\dagger O^{(t)} V\|_\rho^2. 
\end{align*}
\end{lemma}

\begin{proof}
Note that
\begin{align}
\E_t \|P^{(t)}- V^\dagger O^{(t)} V\|_\rho^2 =&~ 1 + \E_t \tr(V^\dagger O^{(t)} VV^\dagger O^{(t)} V\rho) - 2\E_t \Re(\tr((P^{(t)})^\dagger V^\dagger O^{(t)} V\rho))\notag\\
=&~ 1 + \E_t \|V^\dagger O^{(t)} V\|_\rho^2 - 2\E_t \sum_s \Re(\tr(((-1)^{st}\cdot P_s)^\dagger V^\dagger O^{(t)} V\rho))\;.\label{eq:prodt-1}
\end{align}
Moreover,
\begin{align}
\sum_s \Big\|P_s -& V^\dagger \prod_{i\in [n]} \frac{I+(-1)^{s_i}O_i}{2} V\Big\|_\rho^2 \notag\\
=& \sum_s \tr(P_s\rho) + \sum_s \tr( V^\dagger \prod_{i\in [n]} \frac{I+(-1)^{s_i}O_i}{2} V V^\dagger \prod_{i\in [n]} \frac{I+(-1)^{s_i}O_i}{2} V \rho)\notag \\
&\qquad - 2\sum_s \Re(\tr( P_s^\dagger V^\dagger \prod_{i\in [n]} \frac{I+(-1)^{s_i}O_i}{2} V \rho))\notag\\
\leq &~ 2 - 2\sum_s \Re(\tr( P_s^\dagger V^\dagger \prod_{i\in [n]} \frac{I+(-1)^{s_i}O_i}{2} V \rho))\;,\label{eq:prodt-2}
\end{align}
where the middle term was bounded using that $V$ is an isometry. 
Moreover,
\begin{align*}
\prod_{i\in [n]} \frac{I+(-1)^{s_i}O_i}{2} = \frac{1}{2^n}\sum_{t\in\F_2^n} \prod_i O_i^{t_i} (-1)^{s_i t_i} = \E_t (-1)^{st} O^{(t)}.
\end{align*}
Therefore, plugging back into~\eqref{eq:prodt-2},
\begin{align*}
&~\sum_s \|P_s - V^\dagger \prod_{i\in [n]} \frac{I+(-1)^{s_i}O_i}{2} V\|_\rho^2 \\
\leq &~ 2 - 2\E_t \sum_s \Re(\tr( P_s^\dagger  (-1)^{st} V^\dagger O^{(t)} V \rho))\\
\leq &~ \E_t \|P^{(t)}- V^\dagger O^{(t)} V\|_\rho^2 +1 - \E_t \|V^\dagger O^{(t)} V\|_\rho^2\;,
\end{align*}
where the second inequality is by~\eqref{eq:prodt-1}.
\end{proof}

\begin{lemma}[Further implication of the energy consistency tests]\label{lem:8.13}
\begin{align*}
\E_l\sum_s \Big\|\sum_{M,o,w} 1_{D(M)=s} P_{l,M,o,w}- V^\dagger \prod_{i\in [n]} \frac{I+(-1)^{s_i}O_{i,j}}{2} V\Big\|_\rho^2 
\leq 
\Theta(\sqrt{\eps+\eps'+\delta})\;. 
\end{align*}
\end{lemma}

\begin{proof}
Let 
\begin{align*}
P'_{l,t}=\sum_{M,o,w} 1_{D(M)\cdot t=0} P_{l,M,o,w} - \sum_{M,o,w} 1_{D(M)\cdot t=1} P_{l,M,o,w}.
\end{align*}
By Lemma \ref{lem:q_sound_had} with 
\begin{align*}
n=&~5N,\\
P_s =&~ \sum_{M,o,w} 1_{D(M)=s} P_{l,M,o,w} ,\\
O_{i,j} =&~ (-1)^{d_{{l_i},j}}X(a_{{l_i},j})Z(b_{{l_i},j})Y(c_{{l_i},j}),
\end{align*}
we have that
\begin{align}
\sum_s \Big\|\sum_{M,o,w} 1_{D(M)=s} P_{l,M,o,w}&- V^\dagger \prod_{i\in [N], j\in [5]} \frac{I+(-1)^{s_{i,j}}O_{l_i,j}}{2} V\Big\|_\rho^2 \notag\\
\leq &~
\E_t \|P'_{l,t}- V^\dagger O(t,l) V\|_\rho^2 +1 - \E_t \|V^\dagger O(t,l) V\|_\rho^2\;. \label{eq:mt-4}
\end{align}
By Lemma \ref{lem:mqt:ec}, 
\begin{align*}
\E_l \E_t \|P'_{l,t}- V^\dagger O(t,l) V\|_\rho^2 \leq \Theta (\eps + \delta+\eps'). 
\end{align*}

Note that, by Lemma \ref{lem:tri_sigma_norm},
\begin{align*}
&~ \E_l \E_t \| V^\dagger O(t,l) V\|_\rho \\
\geq &~ \E_l \E_t \|P'_{l,t}\|_\rho - \E_l \E_t \|P'_{l,t}- V^\dagger O(t,l) V\|_\rho \\
= &~ 1 - \E_l \E_t \|P'_{l,t}- V^\dagger O(t,l) V\|_\rho \\
\geq &~ 1 - \sqrt{\E_l \E_t \|P'_{l,t}- V^\dagger O(t,l) V\|_\rho^2} \\
\geq &~ 1 - \Theta(\sqrt{\eps+\eps'+\delta}) .
\end{align*}

Therefore, 
\begin{align*}
&~\E_l \E_t \| V^\dagger O(t,l) V\|_\rho^2  \\
\geq &~ (\E_l \E_t \| V^\dagger O(t,l) V\|_\rho )^2 \\
\geq &~ (1 - \Theta(\sqrt{\eps+\eps'+\delta}) )^2\\
\geq &~ 1 - \Theta(\sqrt{\eps+\eps'+\delta}) .
\end{align*}

Plugging back into~\eqref{eq:mt-4},
\begin{align*}
\E_l\sum_s \Big\|\sum_{M,o,w} 1_{D(M)=s} P_{l,M,o,w}- V^\dagger \prod_{i\in [N], j\in [5]} \frac{I+(-1)^{s_{i,j}}O_{l_i,j}}{2} V\Big\|_\rho^2 
\leq 
\Theta(\sqrt{\eps+\eps'+\delta})\;. 
\end{align*}
\end{proof}

\begin{lemma}[Implication of the energy tests]
\label{lem:e_imply}

\begin{align*}
\Pr[\text{energy test passes}]\leq  \eps + \E_{l\in_R [m]^N} 
 \Big\|  \sum_{M,o,w} 
1_{C(D(M))=0} 
\cdot P_{l,M,o,w} U_{0\rightarrow 2}\ket{0}\Big\|^2.
\end{align*}
\end{lemma}

\begin{proof}
Let $S = \{t | C(D(t)) = 0 \wedge T(t) = 1\}$. 
If $ dist(M,S) > \delta$,
\begin{align*}
\Pr[a_1=1|l,M,o,w] \leq  \Pr[\text{verifier accepts}|l,M,o=0,w] \leq  \eps. 
\end{align*}
If $ dist(M,S) \leq \delta$,
\begin{align*}
\Pr[a_1=1|l,M,o,w] \leq  1 = 1_{C(D(M))=0}. 
\end{align*}
Then,
\begin{align*}
\Pr[\text{energy test passes}] 
=&~ \E_{l\in_R [m]^N} 
\sum_{M,o,w} 
\Pr[a_1=1|l,M,o,w] 
\cdot \| P_{l,M,o,w} U_{0\rightarrow 2}\ket{0}\|^2\\
\leq &~ \eps + \E_{l\in_R [m]^N} 
\sum_{M,o,w} 
1_{C(D(M))=0} 
\cdot \| P_{l,M,o,w} U_{0\rightarrow 2}\ket{0}\|^2\\
= &~ \eps + \E_{l\in_R [m]^N} 
 \Big\| \sum_{M,o,w} 
1_{C(D(M))=0} 
\cdot P_{l,M,o,w} U_{0\rightarrow 2}\ket{0}\Big\|^2.
\end{align*}
\end{proof}

\begin{lemma}[The role of the energy consistency test in the energy test]
\label{lem:ec_in_e}
\begin{align*}
&~ \E_{l\in_R [m]^N} 
 \Big\| \sum_{M,o,w} 
1_{C(D(M))=0} 
\cdot P_{l,M,o,w} U_{0\rightarrow 2}\ket{0}\Big\|^2 \\
\leq  &~
\Theta(\sqrt{\eps+\eps'+\delta})+ 1.01\cdot \E_{l\in_R [m]^N} 
 \sum_{s} 1_{C(s)=0}\Big\| V^\dagger \prod_{i\in [N], j\in [5]} \frac{I+(-1)^{s_{i,j}}O_{l_i,j}}{2} V  U_{0\rightarrow 2}\ket{0}\Big\|^2\;.
\end{align*}
\end{lemma}

\begin{proof}
For any $n>0$,
\begin{align*}
\|a+b\|^2 \leq \frac{1+n^2}{n^2} \|a\|^2 + (1+n^2) \|b\|^2 ,
\end{align*}
because
\begin{align*}
0 \leq \frac{1}{n^2} \|a\|^2 + n^2 \|b\|^2 - a^\dagger b - b^\dagger a 
 = \Big\| \frac{1}{n} a - n b \Big\|^2.
\end{align*}
Therefore
\begin{align*}
 &~\E_{l\in_R [m]^N} 
 \Big\| \sum_{M,o,w} 
1_{C(D(M))=0} 
\cdot P_{l,M,o,w} U_{0\rightarrow 2}\ket{0}\Big\|^2 \\
=&~ \E_{l\in_R [m]^N} 
 \sum_{s} 1_{C(s)=0} \Big\| \sum_{M,o,w} 
1_{D(M)=s} 
\cdot P_{l,M,o,w} U_{0\rightarrow 2}\ket{0}\Big\|^2   \\
\leq &~ 101 \E_{l\in_R [m]^N} 
 \sum_{s} 1_{C(s)=0} \Big\| (\sum_{M,o,w} 
1_{D(M)=s} \cdot P_{l,M,o,w} -V^\dagger \prod_{i\in [N], j\in [5]} \frac{I+(-1)^{s_{i,j}}O_{l_i,j}}{2} V)U_{0\rightarrow 2}\ket{0}\Big\|^2 \\
&~\quad+  \frac{101}{100} \E_{l\in_R [m]^N} 
 \sum_{s} 1_{C(s)=0}\Big\| V^\dagger \prod_{i\in [N], j\in [5]} \frac{I+(-1)^{s_{i,j}}O_{l_i,j}}{2} V  U_{0\rightarrow 2}\ket{0}\Big\|^2 \\
\leq &~
\Theta(\sqrt{\eps+\eps'+\delta})+ 1.01\cdot \E_{l\in_R [m]^N} 
 \sum_{s} 1_{C(s)=0}\Big\| V^\dagger \prod_{i\in [N], j\in [5]} \frac{I+(-1)^{s_{i,j}}O_{l_i,j}}{2} V  U_{0\rightarrow 2}\ket{0}\Big\|^2\;,
\end{align*}
Where the last inequality is by Lemma~\ref{lem:8.13}.
\end{proof}

\begin{lemma}[Estimating the energy of Clifford Hamiltonians via Pauli measurements]
\label{lem:measP2e}
Let $H$, $H_i$, $O_{i,j}$ and $C$ be as in Definition~\ref{def:mqt:hp_last_energy}.
Then,
\begin{align*}
&~\E_{l\in_R [m]^N} 
 \sum_{s} 1_{C(s)=0}\Big\| \prod_{i\in [N], j\in [5]}\underbrace{I\otimes I \cdots \otimes I}_{i-1} \otimes\frac{I+(-1)^{s_{i,j}}O_{l_i,j}}{2}\otimes \underbrace{I\otimes I \cdots \otimes I}_{N-i} \ket{\phi}\Big\|^2\\
=&~ \bra{\phi} \frac{1+(1-2H)^{\otimes N} }{2}\ket{\phi} \;.
\end{align*}
\end{lemma}

\begin{proof}
Note that 
\begin{align*}
H_i = \prod_{j\in[5]}\frac{I+(-1)^0 O_{i,j}}{2}= \sum_{1 \oplus \vee_{j\in [5]} x_{j} =1}\prod_{j\in[5]}\frac{I+(-1)^{x_j} O_{i,j}}{2}.
\end{align*}
and 
\begin{align*}
I-H_i = I-\prod_{j\in[5]}\frac{I+O_{i,j}}{2}= \sum_{1 \oplus \vee_{j\in [5]} x_{j} =0 }\prod_{j\in[5]}\frac{I+(-1)^{x_j} O_{i,j}}{2}.
\end{align*}
Thus,
\begin{align*}
\frac{1}{2}+\frac{1}{2}(1-2H)^{\otimes N}
=&~\frac{1}{2}+\frac{1}{2}\cdot \frac{1}{m^N} \sum_{l\in [m]^N} \bigotimes_{i\in [N]} (1-2H_{l_i})\\
=&~\frac{1}{m^N} \sum_{l\in [m]^N}(\frac{1}{2}+\frac{1}{2} \bigotimes_{i\in [N]} (1-2H_{l_i}))\\
=&~  \E_{l\in_R [m]^N}\sum_{ \oplus_{k\in[N]} r_k = 0}  \prod_{i=1}^N (\frac{1}{2}+(-1)^{r_i}(\frac{1}{2}-H_{l_i})) \\
=&~  \E_{l\in_R [m]^N}\sum_{ \oplus_{k\in[N]} r_k = 0}  \prod_{i=1}^N  \sum_{1 \oplus \vee_{j\in [5]} s_{i,j} =r_i }\prod_{j\in[5]}\frac{I+(-1)^{s_{i,j}} O_{l_i,j}}{2}\\
=&~ \E_{l\in_R [m]^N} 
 \sum_{s} 1_{C(s)=0}\prod_{i\in [N], j\in [5]}\frac{I+(-1)^{s_{i,j}}O_{l_i,j}}{2}\;.
\end{align*}
\end{proof}

\begin{theorem}\label{thm:sqiop}
$\QMA\subseteq \SQIOP(O(1), O(\exp\circ \poly), O(1))$   .
\end{theorem}

\begin{proof}
Let $ \texttt V$ be the verifier in Definition \ref{def:mqt:verifier_energy}. We will show that $\texttt V $ is $O(1)$-message non-adaptive strong quantum interactive oracle proof system for $\mathcal{LCH}(5,2^{-p(n)},1/q(n))$ and thus $$ \mathcal{LCH}(5,2^{-p(n)},1/q(n))\in \SQIOP(O(1), O(\exp\circ\poly), O(1)).$$ 
By Theorem \ref{thm:clifford_ham}, $\mathcal{LCH}(5,2^{-p(n)},1/q(n))$ is $ \QMA$-complete. 
Therefore, $$\QMA\subseteq \SQIOP(O(1), O(\exp\circ\poly), O(1)).$$

By  Definition \ref{def:mqt:hp_energy}, Definition \ref{def:mqt:verifier_energy}, and  Definition \ref{def:qiop}, $\texttt V $ is consistent, has polynomial time complexity, is efficient, has proof length $ O(\exp\circ\poly)$, has query complexity $ O(1) $. Moreover, the completeness and soundness results are as follows. Let $$p= \Pr[\text{The verifier opts to execute the ``EnergyTests" subroutine}].$$

{\bf Completeness:} 
By Theorem \ref{thm:mqt}, ``Qubits tests" passes with probability $1$. By Definition \ref{def:mqt:hp_energy} and Definition \ref{def:mqt:verifier_energy}, the energy consistency tests pass with probability $1$ and the energy tests pass with probability 
\begin{align*}
\bra{\phi} \frac{1+(1-2H)^{\otimes N} }{2}\ket{\phi} .
\end{align*}
As a result,
\begin{align*}
\Pr[\text{\texttt{V} accepts}] = &~ 1-p
 + p\cdot \bra{\phi} \frac{1+(1-2H)^{\otimes N} }{2}\ket{\phi}\\
 \geq&~ 1-p\cdot \Theta(2^{-s(n)}).
\end{align*}

{\bf Soundness:} 
Assume that 
\begin{align*}
\Pr[\text{energy consistency test passes}] )\geq &~ 1-\Theta(\eps')\\
\Pr[\text{ consistency check passes }])\geq &~ 1-\Theta(\eps')\\
\Pr[\text{ validity check passes }]\geq &~ 1-\Theta(\eps')\\
\Pr[\text{ anti-commuting test passes }]\geq &~ 1-\Theta(\eps')\;.
\end{align*}

Then by combining Lemma \ref{lem:e_imply}, Lemma \ref{lem:ec_in_e},
and Lemma \ref{lem:measP2e},
\begin{align*}
\Pr[\text{energy test passes}] \leq \Theta(\sqrt{\eps+\eps'+\delta})+ 1.01\cdot  \bra{\phi} \frac{1+(1-2H)^{\otimes N} }{2}\ket{\phi} .
\end{align*}

Moreover, by Definition \ref{def:mqt:verifier_energy}, the acceptance probability of the verifier have two part, the acceptance probability of the energy test and the the acceptance probability of the other tests. According to the inequality above,
\begin{align}
\Pr[\text{\texttt{V} accepts}] \leq &~1-p
 + p\cdot (\Theta(\sqrt{\eps+\eps'+\delta})+ 1.01\cdot  \bra{\phi} \frac{1+(1-2H)^{\otimes N} }{2}\ket{\phi})\notag\\
\leq &~ 1-p
 + p\cdot (\Theta(\sqrt{\eps+\eps'+\delta})+ 0.505 +\Theta(2^{-s(n)})) \notag\\
 \leq &~ 1- 0.4 p ,\label{eq:f1}
\end{align}
where the last step follows by taking  ${\eps,\eps',\delta}$ small enough.

If one of the followings is not true
\begin{align*}
\Pr[\text{energy consistency test passes}] )\geq &~ 1-\Theta(\eps')\\
\Pr[\text{ consistency check passes }])\geq &~ 1-\Theta(\eps')\\
\Pr[\text{ validity check passes }]\geq &~ 1-\Theta(\eps')\\
\Pr[\text{ anti-commuting test passes }]\geq &~ 1-\Theta(\eps')\;,
\end{align*}
then
\begin{align}
\Pr[\text{\texttt{V} accepts}] \leq&~ p+ (1-p)\cdot(1- \Theta(\eps'))\notag\\
= &~ 1- (1-p)\cdot \Theta(\eps')\;. \label{eq:f2}
\end{align}

Note that $\eps$ is the self-correction parameter from Theorem~\ref{thm:3c_code} and $\delta$ is the PCPP proximity parameter from Theorem~\ref{thm:pcpp}. 
As a result, we can take $p=1/3$ and choose $\varepsilon=\delta=\varepsilon'$ sufficiently small. We can set $\varepsilon$ to be any fixed small constant, at the cost of increasing the query complexity of the self-correction procedure by a larger constant factor. By repeating the PCPP verifier a constant number of times, we can obtain a PCPP with an arbitrarily small constant proximity parameter and an arbitrarily small constant soundness error.
By Eq.~\eqref{eq:f1} and Eq.~\eqref{eq:f2}, we get that  
\begin{align*}
\Pr[\text{\texttt{V} accepts}] \leq&~ \max\{1- (1-p)\cdot \Theta(\eps'),1- 0.4 p  \} \leq  1-c\;,
\end{align*}
where $c$ is a small constant. 
\end{proof}

%% file: test1qubit.tex
\section{Single Qubit Tests}
\label{sec:sqt}

Our single qubit test starts with two rounds corresponding to X and Z measurement. The game is as follows: in each round, the honest prover sends the purified measurement of X for the first round, Z for the second round. The verifier either measures under the $Z$ bases or applies the $Z$ gates on some of the qubits, then the verifier sends all the registers from the prover's message back, the honest prover undoes the purified measurement. 

In the end, there is an additional round. The verifier sends a symbol either ``X" or ``Z". The honest prover returns the result of the measurement. 

The verifier performs the anti-commuting tests and the consistency checks. More specifically, the verifier checks the consistency corresponding to the first Z measurement and the Z measurement in the last round with an X applying in the middle. Moreover, the verifier checks the first two rounds consisting with the last rounds. We will show that the anti-commuting relationship will appear at the last round.

\subsection{Honest prover's strategy}
Now we define the exact way the honest prover acts in each of the first two rounds. The honest prover uses different $ O$ in different rounds. Later, we will describe the details of the implementation of the template.
\begin{definition}[Template of the honest prover's actions in the first two rounds]
\label{def:s:honest_prover_sqt:block}
Let $\mathsf{P}$ be a register of $1$ qubit, $\mathsf{R}$ be a register of $1$ qubit. 
$\mathsf{P, R}$ are initially hold by the prover. 
Let $ O_1,O_2$ be observables on register $\mathsf{P}$. 
\begin{mdframed}
\begin{itemize}
\item \textnormal{The prover's first step:} 
The prover applies the following unitary 
\begin{align*}
P(O_{\mathsf{P}}, \mathsf{R}). 
\end{align*}
The prover sends the verifier the register $\mathsf{R}$.
\item \textnormal{The verifier's first step:} 
The verifier receives the register $\mathsf{R}$ from the prover. The verifier sends the prover the register $\mathsf{R}$. 
\item \textnormal{The prover's second step:} 
The prover receives the register $\mathsf{R}$ from the verifier. The prover applies the following unitary 
\begin{align*}
(P(O_{\mathsf{P}}, \mathsf{R}))^{-1}. 
\end{align*}
\end{itemize}

\end{mdframed}
\end{definition}

Then, we describe the honest prover's actions in the last round:
\begin{definition}[The honest prover's actions in the last round]
\label{def:s:honest_prover_sqt:last}
Let $\mathsf{P}$ be a register of $1$ qubits. 
$\mathsf{P}$ is initially hold by the prover.  
\begin{mdframed}
\begin{itemize}
\item \textnormal{The verifier's first step:} The verifier sends the prover $q'\in [2]$.
\item \textnormal{The prover's first step:} 
The prover measures $\mathsf{P}$ under observable $ Z$ if $q'=1$, under observable $X$ if $q'=2$ and returns the measurement result. 
\end{itemize}
\end{mdframed}
\end{definition}

Now we summarize the honest prover's actions:
\begin{definition}[Single qubit tests, honest prover]\label{def:s:sqt:honest_prover} 
Let $\mathsf{P}$ be a register of $1$ qubits, $\mathsf{R}$ be a register of $1$ qubits. $\mathsf{P, R}$ are initially hold by the prover and initialized by $\ket{0}$. 

The honest prover's actions are defined as follows: 
\begin{mdframed}
\begin{itemize}
\item The honest prover acts as in Definition \ref{def:s:honest_prover_sqt:block} with $ O=Z$. 
\item The honest prover acts as in Definition \ref{def:s:honest_prover_sqt:block} with $ O=X$. 
\item In the end, the honest prover acts as in Definition \ref{def:s:honest_prover_sqt:last}.  
\end{itemize}
\end{mdframed}

\end{definition}

\subsection{Verifier's strategy}
Then we present the verifier's action. For the ease of the presentation, we first abstract the process of applying a specific unitary operation on a single qubit within the prover's message.

In the first two rounds, the verifier can respond in various ways to the prover's messages. Later, we will show how the verifier chooses to respond in the whole protocol of the single qubit tests. 
In the following definition, $M, F$ can be viewed as the input to the template. 
$M$ indicates whether or not to measure. 
$F$ indicates whether or not to flip. 

\begin{definition}[Template of the verifier's action in the first two rounds]\label{def:s:verifier_sqt:block}
Let $\mathsf{R} $ be a register. Let $M,F\in \{0,1\}$. 
\begin{mdframed}
\begin{itemize}
\item \textnormal{The prover's first step:} 
The prover sends the verifier the register $\mathsf{R}$.
\item \textnormal{The verifier's first step:} 
The verifier receives the register $\mathsf{R}$ from the prover. 
If $F=1$, verifier applies $Z_{\mathsf{R}}$. If $M=1$, verifier
measures $Z_{\mathsf{R}}$ and outputs the measurement result. 
The verifier sends the prover the register $\mathsf{R}$. 
\item \textnormal{The prover's second step:} 
The prover receives the register $\mathsf{R}$ from the verifier. 
\end{itemize}
\end{mdframed}
\end{definition}

Then, we describe the verifier's action in the last round, where $q'=1$ corresponding to a symbol ``X", $q'=2$ corresponding to a symbol ``Z":
\begin{definition}[The verifier's action in the last round]
\label{def:s:verifier_sqt:last}
Let $q'\in [2]$. 
\begin{mdframed}
\begin{itemize}
\item \textnormal{The verifier's first step:} 
The verifier sends $q'$ to the prover.
\item \textnormal{The prover's first step:} 
The prover returns a single bit $r'$.  
\end{itemize}
\end{mdframed}
\end{definition}

Now, we introduce the verifier's action in the single qubit tests. The verifier has two types of checks: consistency checks and anti-commuting checks. We first introduce the consistency check. In this check, the verifier checks the consistency between the first two rounds and the last round.
More specifically, the verifier checks that the answer in the final round matches the answer from the previous round.
\begin{definition}[Single qubit tests, consistency check]\label{def:s:sqt:verifier:c} 
Let $\mathsf{R} $ be a register.  
The verifier's action are defined as follows: 
\begin{mdframed}
\begin{itemize}
\item The verifier samples $u\in_R [2]$, sets $M_u=1$, and sets $ M_i =0$ for $i \neq u, i\in [2]$. 
\item If $M_1=1$, the verifier sets $F_i = 0$ for $i\in [2]$. If $M_2=1$, the verifier sets $F_1 \in_R \F_2$, $F_2 = 0$. 
\item The verifier repeats the verifier's action in Definition \ref{def:s:verifier_sqt:block} for $2$ times. In the $i$-th time, the verifier acts as in Definition \ref{def:s:verifier_sqt:block} with $ M = M_i, F = F_i$, if $M_i=1$, sets $r_i$ be the output. 
\item In the end, the verifier acts as in Definition \ref{def:s:verifier_sqt:last} with $q'= u$. Let $r'$ denote the bit that the verifier received from the prover.  
\item The verifier accepts iff $ r_u=r'$.
\end{itemize}
\end{mdframed}
\end{definition}

Then, we introduce the anti-commuting tests, where the verifier checks the consistency corresponding to the first Z measurement and the Z measurement in the last round with an X applies in the middle.
\begin{definition}[Single qubit tests, anti-commuting tests]\label{def:s:sqt:verifier:a} 
Let $\mathsf{R} $ be a register.  
The verifier's action are defined as follows: 
\begin{mdframed}
\begin{itemize}
\item The verifier sets $M_1=1, F_1=0, M_2=0, F_2=1$.  
\item The verifier repeats the verifier's action in Definition \ref{def:s:verifier_sqt:block} for $2$ times. In the $i$-th time, the verifier acts as in Definition \ref{def:s:verifier_sqt:block} with $ M = M_i, F = F_i$, if $M_i=1$, sets $r_i$ be the output. 
\item In the end, the verifier acts as in Definition \ref{def:s:verifier_sqt:last} with $q'= 1$ (corresponding to the symbol ``Z"). Let $r'$ denote the bit that the verifier received from the prover.  
\item The verifier accepts iff $ r_1=r'\oplus 1$.
\end{itemize}
\end{mdframed}
\end{definition}

Now we put everything together and define the action of the verifier in the single qubit tests. 
\begin{definition}[Single qubit tests, verifier]\label{def:s:sqt:verifier} 
Let $\mathsf{R} $ be a register.  

The verifier's action are defined as follows: 
\begin{mdframed}
\begin{enumerate}
    \item With probability $1/2$, the verifier acts as in Definition \ref{def:s:sqt:verifier:c},
    \item With probability $1/2$, the verifier acts as in Definition \ref{def:s:sqt:verifier:a}.
\end{enumerate}
\end{mdframed}
\end{definition}

\subsection{Main theorem}
Before we state our main theorem in this section, we first introduce a general prover's actions. 
\begin{definition}[The canonical prover's actions]
\label{def:s:sqt:canonical_prover}
Let $\mathsf{P}$ be a register, $\mathsf{R}$ be a register of $1$ qubits. 
$\mathsf{P, R}$ are initially hold by the prover and initialized by $\ket{0}$. For $q'\in [2]$, let $O'_{q'}$ be an observable such that $ (O'_{q'})^2=I$. 
For $i\in [2]$, let $ U_i$ be a unitary. 
\begin{mdframed}
\hfill 

The canonical prover's actions in the first two rounds are defined as follows: At round $i$, 
\begin{itemize}
\item \textnormal{The prover's first step:} 
The prover applies a unitary $ {U}_{i,\mathsf{PR}}$. 
The prover sends the verifier the register $\mathsf{R}$.
\item \textnormal{The verifier's first step:} 
The verifier receives the register $\mathsf{R}$ from the prover. The verifier sends the prover the register $\mathsf{R}$. 
\item \textnormal{The prover's second step:} 
The prover receives the register $\mathsf{R}$ from the verifier. 
\end{itemize}

The canonical prover's actions in the last round  are defined as follows: 
\begin{itemize}
\item \textnormal{The verifier's first step:} 
The verifier sends $q'\in [2]$ to the prover.
\item \textnormal{The prover's first step:} 
The prover measures the observable $ O'_{q', \mathsf{PR}}$ and returns the measurement result.  
\end{itemize}
\end{mdframed}
\end{definition}
\begin{remark}
Note that any prover's actions can be purified into a canonical prover's actions described above. 
\end{remark}

For the ease of presentation, we introduce transition actions.
\begin{definition}[Transition actions]\label{def:s:trans}
Let $ {U}_{i}$ for $i\in[2]$ be defined as in Definition \ref{def:s:sqt:canonical_prover}. 

For $i<j$, $i,j\in \{0, 1,2\}$, define the transition action from $i$ to $j$ as:
\begin{align*}
{U}_{i\rightarrow j} ={U}_{j} \cdots {U}_{i+2}{U}_{i+1}.
\end{align*}

\end{definition}

Now, we state our main theorem in this section. 
\begin{theorem}[Single qubit tests]
\label{thm:s:sqt}
Let $ \texttt{V                 
      }$ be the verifier defined in Definition \ref{def:s:sqt:verifier}. Then, 

\begin{enumerate}
\item \textnormal{Completeness:} $ \texttt{V}$ accepts the prover in Definition \ref{def:s:sqt:honest_prover}  with probability $1$.
\item \textnormal{Soundness:} Let $ \texttt{P}$ be the prover defined in Definition \ref{def:s:sqt:canonical_prover}.  Let $ O'_{1},O'_{2}$ be defined as in Definition \ref{def:s:sqt:canonical_prover}. 
If $ \texttt{V}$ accepts $ \texttt{P}$ with probability greater than $ 1- \delta$, then the approximate anti-commuting relationship of $ O'_{1},O'_{2}$ hold under the state 
$$\rho = U_{0\rightarrow 2}\ket{00}_{\mathsf{PR}}\bra{00}_{\mathsf{PR}} U_{0\rightarrow 2}^\dagger. $$
More specifically,
$$\tr( (O'_{1}O'_{2}+O'_{2}O'_{1})^2\rho) <  \Theta(\delta).$$
\end{enumerate}
\end{theorem}

Before proving the main theorem, we first introduce the necessary notations. Following this, we will establish several lemmas that will aid in the proof of the main theorem.
\begin{definition}[Transition actions with flips]\label{def:s:trans_f}
Let $ {U}_{i}$ for $i\in[2]$ be defined as in Definition \ref{def:s:sqt:canonical_prover}.
For $u\in [2], F\in \{0, 1\}^2$, define the transition action with flips as:
\begin{align*}
{U}_{1, F} =&~ Z_{\mathsf{R}}^{F_1}{U}_{1}, \\
{U}_{2, F} =&~ Z_{\mathsf{R}}^{F_2}{U}_{2} Z_{\mathsf{R}}^{F_1}{U}_{1}.
\end{align*}

\end{definition}

Next, we examine the relationship between the overall success probability and the success probabilities in each configuration.
\begin{lemma}\label{lem:s:sqt:prob}
Let $\texttt{V},\texttt{P}$ be defined as in Theorem \ref{thm:s:sqt}. 
Let $D$ be the distribution of $ M,F$ in Definition \ref{def:s:sqt:verifier:c}.
Then
$$ \Pr[\text{ \texttt{V} accepts \texttt{P} } ] \geq 1-\Theta(\delta)\;,$$  
if and only if:
\begin{itemize}
\item for any $M,F\in\{0, 1\}^2$ in the support of $D$,
\begin{align*}
\Pr[\text{ consistency check passes }~|~M,F] \geq 1- \Theta(\delta)\;;
\end{align*}
\item and
\begin{align*}
\Pr[\text{ anti-commuting test passes }] \geq 1- \Theta(\delta)\;.
\end{align*}
\end{itemize}
\end{lemma}
\begin{proof}
This lemma follows from a straightforward calculation.
\end{proof}

\begin{lemma}
Let $O_1, O_2 \in \mathbb{C}^{n \times n}$ be observables, $ O_1^2 = O_2^2 = I, O_1^\dagger = O_1, O_2^\dagger = O_2$. Let $\ket{\phi} \in \mathbb{C}^n$ be a quantum state. Then
\[
\sum_{a \in \{\pm 1\}} \Big\| \frac{I + a \cdot O_1}{2} \cdot U \cdot \frac{I + a \cdot O_2}{2} \ket{\phi} \Big\|^2 = \bra{\phi} \frac{2I + U^\dagger O_1 U O_2 + O_2^\dagger U^\dagger O_1 U}{4} \ket{\phi}.
\]
\end{lemma}
\begin{proof}
We have that 
\begin{align*}
&~ \sum_{a \in \{\pm 1\}} \Big\| \frac{I + a \cdot O_1}{2} \cdot U \cdot \frac{I + a \cdot O_2}{2} \ket{\phi} \Big\|^2 \\
=&~ \sum_{a \in \{\pm 1\}} \bra{\phi}\frac{I + a \cdot O_2^\dagger}{2} \cdot U^\dagger \cdot \frac{I + a \cdot O_1}{2} \cdot U \cdot \frac{I + a \cdot O_2}{2} \ket{\phi}  \\
=&~ \sum_{a \in \{\pm 1\}} 
\bra{\phi}\frac{2I + U^\dagger O_1 U O_2 + O_2^\dagger U^\dagger O_1 U + a(O_2+ U^\dagger O_1 U + O_2^\dagger +O_2^\dagger U^\dagger O_1 U O_2)}{8} \ket{\phi}
\\
=&~ 
\bra{\phi}\frac{2I + U^\dagger O_1 U O_2 + O_2^\dagger U^\dagger O_1 U }{4} \ket{\phi}
.
\end{align*}

\end{proof}

\begin{lemma}[Implication of the consistency check]\label{lem:s:sqt:c}
Let $ {U}_{i}$ for $i\in[2]$ be defined as in Definition \ref{def:s:sqt:canonical_prover}.
Let $\mathsf{P},\mathsf{R}$ be defined as in Definition \ref{def:s:sqt:canonical_prover}. 
Let $D$ be the distribution of $ M,F$ in Definition \ref{def:s:sqt:verifier:c}. For any $M,F\in\{0, 1\}^2$ in the support of $D$, let $u\in\{0,1\}$ be such that $ M_u=1$. 
Then
\begin{align*}
\|(O'_{u} U_{2, F} -U_{u\rightarrow 2} Z_{\mathsf{R}} U_{u, F} )\ket{00}_{\mathsf{PR}}\|^2 =  4\cdot(1-\Pr[\text{ consistency check passes }~|~M,F])\;. 
\end{align*}
\end{lemma}

\begin{proof}
By a straightforward calculation,
\begin{align*}
&~\Pr[\text{ consistency check passes }~|~M,F]\\
=&~\sum_{a\in \{\pm 1\}}
\|  \frac{I + a\cdot  O'_{u}}{2} \cdot U_{u\rightarrow 2} \cdot \frac{I + a\cdot  Z}{2} \cdot U_{u, F} \ket{0}  \|^2\\
=&~\bra{0} \frac{2I + U_{u, F}^\dagger   Z  U_{u\rightarrow 2}^\dagger O'_{u} U_{2, F}+U_{2, F}^\dagger O'_{u} U_{u\rightarrow 2} Z U_{u, F} }{4} \ket{0} . 
\end{align*} 
Moreover, we have that
\begin{align*}
&~\|O'_{u} U_{2, F} \ket{0} - U_{u\rightarrow 2} Z_{\mathsf{R}} U_{u, F} \ket{0}\|^2 \\
=&~ 2 - \bra{0} U_{u, F}^\dagger   Z  U_{u\rightarrow 2}^\dagger O'_{u} U_{2, F}+U_{2, F}^\dagger O'_{u} U_{u\rightarrow 2} Z U_{u, F} \ket{0}.
\end{align*}

\end{proof}

\begin{lemma}[Implication of the anti-commuting tests]\label{lem:s:sqt:a}
Let $ {U}_{i}$ for $i\in[2]$ be defined as in Definition \ref{def:s:sqt:canonical_prover}.
Let $\mathsf{P},\mathsf{R}$ be defined as in Definition \ref{def:s:sqt:canonical_prover}.

Then,
\begin{align*}
\|(O'_{1} Z U_{0\rightarrow 2} + Z U_2 Z U_1)\ket{00}_{\mathsf{PR}}\|^2= 4\cdot(1-\Pr[\text{ anti-commuting test passes }]). 
\end{align*}
\end{lemma}
\begin{proof}
By a straightforward calculation,
\begin{align*}
&~\Pr[\text{ anti-commuting test passes }]\\
=&~\sum_{a\in \{\pm 1\}}
\|  \frac{I - a\cdot  O'_{1}}{2} \cdot Z \cdot U_{2} \cdot \frac{I + a\cdot  Z}{2} \cdot U_{1} \ket{0}  \|^2\\
=&~\bra{0} \frac{2I -U_{0\rightarrow 2}^\dagger Z O'_{1}  Z  U_{2} Z U_{1}-U_{1}^\dagger Z U_{2}^\dagger Z O'_{1} Z U_{0 \rightarrow 2} }{4} \ket{0} . 
\end{align*} 
Moreover, we have that
\begin{align*}
&~\|O'_{1} Z U_{0\rightarrow 2}\ket{0} + Z U_2 Z U_1\ket{0}\|^2 \\
=&~ 2 + \bra{0} U_{0\rightarrow 2}^\dagger Z O'_{1}  Z  U_{2} Z U_{1}+U_{1}^\dagger Z U_{2}^\dagger Z O'_{1} Z U_{0 \rightarrow 2} \ket{0}.
\end{align*}

\end{proof}

The following lemmas will be instrumental in proving the completeness of Theorem \ref{thm:s:sqt}.

\begin{lemma}[The honest prover in the canonical prover’s form]
\label{lem:s:sqt:hp_cp}
Let 
$$ R_1 = P(Z_{\mathsf{P}}, \mathsf{R}), R_2 = P(X_{\mathsf{P}}, \mathsf{R}). $$

The prover in Definition \ref{def:s:sqt:honest_prover} is the prover in Definition \ref{def:s:sqt:canonical_prover} with:
\begin{align*}
    U_1 =&~ R_1,\\
    U_2 =&~ R_2\cdot R_{1}^{-1},
\end{align*}
and
\begin{align*}
O'_{1}=&~ R_2 \cdot Z_{\mathsf{P}} \cdot R_2^{-1}, \\
O'_{2}=&~ R_2 \cdot X_{\mathsf{P}} \cdot R_2^{-1}.
\end{align*}
\end{lemma}
\begin{proof}
This lemma follows from a straightforward calculation.
\end{proof}

\begin{lemma}[Properties of the honest prover]\label{lem:s:hp_properties}
For $i\in [2]$, let $ U_i$ be defined as in Lemma \ref{lem:s:sqt:hp_cp}. For $q'\in [2]$, let $ O'_{q'}$ be defined as in Lemma \ref{lem:s:sqt:hp_cp}. 
Then:

\begin{itemize}
\item For any $u\in[2]$,
\begin{align*}
O'_{u} U_{u\rightarrow 2} =U_{u\rightarrow 2} Z_{\mathsf{R}}.
\end{align*}
\item Moreover, 
\begin{align*}
O'_{1} Z_{\mathsf{R}} U_{2} = -Z_{\mathsf{R}} U_{2} Z_{\mathsf{R}}
\end{align*}
\end{itemize}
\end{lemma}
\begin{proof}
\hfill

Let $R_i$ be defined as in Lemma \ref{lem:s:sqt:hp_cp}. 

{\bf Part I:}
It is sufficient to show that
\begin{align*}
R_2^{-1} O'_{u} R_2 = R_u^{-1} Z_{\mathsf{R}} R_u,
\end{align*}
which follows from 
\begin{align*}
R_2^{-1} O'_{1} R_2= Z_{\mathsf{P}},\\
R_2^{-1} O'_{2} R_2= X_{\mathsf{P}},
\end{align*}
and
\begin{align*}
R_1^{-1} Z_{\mathsf{R}} R_1 =&~ Z_{\mathsf{P}}, \\
R_2^{-1} Z_{\mathsf{R}} R_2 =&~ X_{\mathsf{P}}.
\end{align*}

{\bf Part II:}
It is sufficient to show that
\begin{align*}
Z_{\mathsf{P}} R_2^{-1} Z_{\mathsf{R}} R_2 = -R_2^{-1} Z_{\mathsf{R}} R_2 R_1^{-1} Z_{\mathsf{R}} R_1,
\end{align*}
which follows from 
\begin{align*}
Z_{\mathsf{P}} X_{\mathsf{P}} = -X_{\mathsf{P}} Z_{\mathsf{P}}.
\end{align*}
\end{proof}

\begin{proof}[Proof of Theorem \ref{thm:s:sqt}]

\hfill 

{\bf Completeness:}
For $i\in [2]$, let $ U_i$ be defined as in Lemma \ref{lem:s:sqt:hp_cp}. For $q'\in [2]$, let $ O'_{q'}$ be defined as in Lemma \ref{lem:s:sqt:hp_cp}. 
Let $D$ be the sampling distribution of $ M,F$ in Definition \ref{def:s:sqt:verifier:c}.

By Lemma \ref{lem:s:sqt:prob}, Lemma \ref{lem:s:sqt:c}, and Lemma \ref{lem:s:sqt:a}, we only need to prove that for any $M,F\in\{0, 1\}^2$ in the support of $D$, letting $u$ be such that $ M_u=1$,
\begin{align*}
O'_{u} U_{u\rightarrow 2} =U_{u\rightarrow 2} Z_{\mathsf{R}},
\end{align*}
and 
\begin{align*}
O'_{1} Z_{\mathsf{R}} U_{0 \rightarrow 2} = -Z_{\mathsf{R}} U_{2} Z_{\mathsf{R}} U_{1}.
\end{align*}
Both parts follow directly from Lemma \ref{lem:s:hp_properties}.

{\bf Soundness:}

We use $ \ket \alpha \approx_{\delta} \ket \beta$ to denote $ \|\ket \alpha- \ket \beta\|^2 \leq \Theta (\delta)$.
We have that
\begin{align*}
&~\tr( (O'_{1}O'_{2}+O'_{2}O'_{1})^2 U_{0\rightarrow 2}\ket{00}_{\mathsf{PR}}\bra{00}_{\mathsf{PR}} U_{0\rightarrow 2}^\dagger)\\
=&~\bra{00}_{\mathsf{PR}}U_{0\rightarrow 2}^\dagger(O'_{1}O'_{2}+O'_{2}O'_{1})^2U_{0\rightarrow 2}\ket{00}_{\mathsf{PR}} \\
=&~\|O'_{1}O'_{2}U_{0\rightarrow 2}\ket{00}_{\mathsf{PR}}+O'_{2}O'_{1}U_{0\rightarrow 2}\ket{00}_{\mathsf{PR}}\|^2 .
\end{align*}
By Lemma \ref{lem:s:sqt:prob}, Lemma \ref{lem:s:sqt:c}, and Lemma \ref{lem:s:sqt:a}, we establish the following relationship:
\begin{align*}
&~ O'_{1}O'_{2}U_{0\rightarrow 2}\ket{0}\\
\approx_\delta &~ O'_{1}Z U_{2} U_1\ket{0}\\
\approx_\delta &~ - Z U_{2} Z U_1\ket{0}\\
\approx_\delta &~ - O'_2 U_{2} Z U_1\ket{0}\\
\approx_\delta &~ - O'_2 O'_1 U_{0\rightarrow 2}\ket{0}.
\end{align*}
Therefore,  
\begin{align*}
\|O'_{1}O'_{2}U_{0\rightarrow 2}\ket{00}_{\mathsf{PR}}+O'_{2}O'_{1}U_{0\rightarrow 2}\ket{00}_{\mathsf{PR}}\|^2\leq \Theta(\delta).
\end{align*}
\end{proof}

%% file: ref.bib
@article{gheorghiu2019verification,
  title={Verification of quantum computation: An overview of existing approaches},
  author={Gheorghiu, Alexandru and Kapourniotis, Theodoros and Kashefi, Elham},
  journal={Theory of computing systems},
  volume={63},
  pages={715--808},
  year={2019},
  publisher={Springer}
}

@phdthesis{vidick2011complexity,
  title={The complexity of entangled games},
  author={Vidick, Thomas},
  year={2011},
  school={UC Berkeley}
}

@article{fitzsimons2017unconditionally,
  title={Unconditionally verifiable blind quantum computation},
  author={Fitzsimons, Joseph F and Kashefi, Elham},
  journal={Physical Review A},
  volume={96},
  number={1},
  pages={012303},
  year={2017},
  publisher={APS}
}

@inproceedings{natarajan2019neexp,
  title={{NEEXP} is contained in {MIP}$^*$},
  author={Natarajan, Anand and Wright, John},
  booktitle={2019 IEEE 60th annual symposium on foundations of computer science (FOCS)},
  pages={510--518},
  year={2019},
  organization={IEEE}
}

@article{fitzsimons2018post,
  title={Post hoc verification of quantum computation},
  author={Fitzsimons, Joseph F and Hajdu{\v{s}}ek, Michal and Morimae, Tomoyuki},
  journal={Physical review letters},
  volume={120},
  number={4},
  pages={040501},
  year={2018},
  publisher={APS}
}

@ARTICLE{556667,

  author={Sipser, M. and Spielman, D.A.},

  journal={IEEE Transactions on Information Theory}, 

  title={Expander codes}, 

  year={1996},

  volume={42},

  number={6},

  pages={1710-1722},

  doi={10.1109/18.556667}}

@article{buhrman2024quantum,
  title={Quantum {PCP}s: on adaptivity, multiple provers and reductions to local hamiltonians},
  author={Buhrman, Harry and Helsen, Jonas and Weggemans, Jordi},
  journal={arXiv preprint arXiv:2403.04841},
  year={2024}
}

@article{aharonov2017interactive,
  title={Interactive proofs for quantum computations},
  author={Aharonov, Dorit and Ben-Or, Michael and Eban, Elad and Mahadev, Urmila},
  journal={arXiv preprint arXiv:1704.04487},
  year={2017}
}

@inproceedings{bk97,
  author       = {L{\'{a}}szl{\'{o}} Babai and
                  Peter G. Kimmel},
  title        = {Randomized Simultaneous Messages: Solution of a Problem of Yao in
                  Communication Complexity},
  booktitle    = {Proceedings of the Twelfth Annual {IEEE} Conference on Computational
                  Complexity, Ulm, Germany, June 24-27, 1997},
  pages        = {239--246},
  publisher    = {{IEEE} Computer Society},
  year         = {1997},
  url          = {https://doi.org/10.1109/CCC.1997.612319},
  doi          = {10.1109/CCC.1997.612319},
  bibsource    = {dblp computer science bibliography, https://dblp.org}
}

@inproceedings{bfl90,
  author       = {L{\'{a}}szl{\'{o}} Babai and
                  Lance Fortnow and
                  Carsten Lund},
  title        = {Non-Deterministic Exponential Time Has Two-Prover Interactive Protocols},
  booktitle    = {31st Annual Symposium on Foundations of Computer Science, St. Louis,
                  Missouri, USA, October 22-24, 1990, Volume {I}},
  pages        = {16--25},
  publisher    = {{IEEE} Computer Society},
  year         = {1990},
  url          = {https://doi.org/10.1109/FSCS.1990.89520},
  doi          = {10.1109/FSCS.1990.89520},
  bibsource    = {dblp computer science bibliography, https://dblp.org}
}

@inproceedings{abcy22,
  author       = {Gal Arnon and
                  Amey Bhangale and
                  Alessandro Chiesa and
                  Eylon Yogev},
  editor       = {Eike Kiltz and
                  Vinod Vaikuntanathan},
  title        = {A Toolbox for Barriers on Interactive Oracle Proofs},
  booktitle    = {Theory of Cryptography - 20th International Conference, {TCC} 2022,
                  Chicago, IL, USA, November 7-10, 2022, Proceedings, Part {I}},
  series       = {Lecture Notes in Computer Science},
  volume       = {13747},
  pages        = {447--466},
  publisher    = {Springer},
  year         = {2022},
  url          = {https://doi.org/10.1007/978-3-031-22318-1\_16},
  doi          = {10.1007/978-3-031-22318-1\_16},
  bibsource    = {dblp computer science bibliography, https://dblp.org}
}

@inproceedings{cgjv19,
  author       = {Andrea Coladangelo and
                  Alex Bredariol Grilo and
                  Stacey Jeffery and
                  Thomas Vidick},
  editor       = {Yuval Ishai and
                  Vincent Rijmen},
  title        = {Verifier-on-a-Leash: New Schemes for Verifiable Delegated Quantum
                  Computation, with Quasilinear Resources},
  booktitle    = {Advances in Cryptology - {EUROCRYPT} 2019 - 38th Annual International
                  Conference on the Theory and Applications of Cryptographic Techniques,
                  Darmstadt, Germany, May 19-23, 2019, Proceedings, Part {III}},
  series       = {Lecture Notes in Computer Science},
  volume       = {11478},
  pages        = {247--277},
  publisher    = {Springer},
  year         = {2019},
  url          = {https://doi.org/10.1007/978-3-030-17659-4\_9},
  doi          = {10.1007/978-3-030-17659-4\_9},
  bibsource    = {dblp computer science bibliography, https://dblp.org}
}

@inproceedings{crsv17,
  author       = {Rui Chao and
                  Ben W. Reichardt and
                  Chris Sutherland and
                  Thomas Vidick},
  editor       = {Christos H. Papadimitriou},
  title        = {Overlapping Qubits},
  booktitle    = {8th Innovations in Theoretical Computer Science Conference, {ITCS}
                  2017, January 9-11, 2017, Berkeley, CA, {USA}},
  series       = {LIPIcs},
  volume       = {67},
  pages        = {48:1--48:21},
  publisher    = {Schloss Dagstuhl - Leibniz-Zentrum f{\"{u}}r Informatik},
  year         = {2017},
  url          = {https://doi.org/10.4230/LIPIcs.ITCS.2017.48},
  doi          = {10.4230/LIPICS.ITCS.2017.48},
  bibsource    = {dblp computer science bibliography, https://dblp.org}
}

@article{dls22,
  title={Spectral gap and stability for groups and non-local games (2022)},
  author={de la Salle, Mikael},
  journal={arXiv preprint arXiv:2204.07084},
  year={2022}
}

@inproceedings{lms22,
  author       = {Alex Lombardi and
                  Fermi Ma and
                  Nicholas Spooner},
  title        = {Post-Quantum Zero Knowledge, Revisited or: How to Do Quantum Rewinding
                  Undetectably},
  booktitle    = {63rd {IEEE} Annual Symposium on Foundations of Computer Science, {FOCS}
                  2022, Denver, CO, USA, October 31 - November 3, 2022},
  pages        = {851--859},
  publisher    = {{IEEE}},
  year         = {2022},
  url          = {https://doi.org/10.1109/FOCS54457.2022.00086},
  doi          = {10.1109/FOCS54457.2022.00086},
  bibsource    = {dblp computer science bibliography, https://dblp.org}
}

@inproceedings{gri19,
  author       = {Alex B. Grilo},
  editor       = {Christel Baier and
                  Ioannis Chatzigiannakis and
                  Paola Flocchini and
                  Stefano Leonardi},
  title        = {A Simple Protocol for Verifiable Delegation of Quantum Computation
                  in One Round},
  booktitle    = {46th International Colloquium on Automata, Languages, and Programming,
                  {ICALP} 2019, July 9-12, 2019, Patras, Greece},
  series       = {LIPIcs},
  volume       = {132},
  pages        = {28:1--28:13},
  publisher    = {Schloss Dagstuhl - Leibniz-Zentrum f{\"{u}}r Informatik},
  year         = {2019},
  url          = {https://doi.org/10.4230/LIPIcs.ICALP.2019.28},
  doi          = {10.4230/LIPICS.ICALP.2019.28},
  bibsource    = {dblp computer science bibliography, https://dblp.org}
}

@article{gh17,
  title={Inverse and stability theorems for approximate representations of finite groups},
  author={Gowers, WT and Hatami, O},
  journal={Sbornik: Mathematics},
  volume={208},
  number={12},
  pages={1787--1817},
  year={2017},
  publisher={Turpion Limited}
}

@misc{vid20,
  title={Course {FSMP}, {F}all’20: Interactions with quantum devices},
  author={Vidick, Thomas},
  year={2020}
}

@inproceedings{bgh+04,
  author       = {Eli Ben{-}Sasson and
                  Oded Goldreich and
                  Prahladh Harsha and
                  Madhu Sudan and
                  Salil P. Vadhan},
  editor       = {L{\'{a}}szl{\'{o}} Babai},
  title        = {Robust {PCP}s of proximity, shorter {PCP}s and applications to coding},
  booktitle    = {Proceedings of the 36th Annual {ACM} Symposium on Theory of Computing,
                  Chicago, IL, USA, June 13-16, 2004},
  pages        = {1--10},
  publisher    = {{ACM}},
  year         = {2004},
  url          = {https://doi.org/10.1145/1007352.1007361},
  doi          = {10.1145/1007352.1007361},
  bibsource    = {dblp computer science bibliography, https://dblp.org}
}

@inproceedings{wehner06,
  author       = {Stephanie Wehner},
  editor       = {Bruno Durand and
                  Wolfgang Thomas},
  title        = {Entanglement in Interactive Proof Systems with Binary Answers},
  booktitle    = {{STACS} 2006, 23rd Annual Symposium on Theoretical Aspects of Computer
                  Science, Marseille, France, February 23-25, 2006, Proceedings},
  series       = {Lecture Notes in Computer Science},
  volume       = {3884},
  pages        = {162--171},
  publisher    = {Springer},
  year         = {2006},
  url          = {https://doi.org/10.1007/11672142\_12},
  doi          = {10.1007/11672142\_12},
  bibsource    = {dblp computer science bibliography, https://dblp.org}
}

@inproceedings{dinur06,
  author       = {Irit Dinur},
  editor       = {Jon M. Kleinberg},
  title        = {The {PCP} theorem by gap amplification},
  booktitle    = {Proceedings of the 38th Annual {ACM} Symposium on Theory of Computing,
                  Seattle, WA, USA, May 21-23, 2006},
  pages        = {241--250},
  publisher    = {{ACM}},
  year         = {2006},
  url          = {https://doi.org/10.1145/1132516.1132553},
  doi          = {10.1145/1132516.1132553},
  bibsource    = {dblp computer science bibliography, https://dblp.org}
}

@inproceedings{bkl+22,
  author       = {James Bartusek and
                  Yael Tauman Kalai and
                  Alex Lombardi and
                  Fermi Ma and
                  Giulio Malavolta and
                  Vinod Vaikuntanathan and
                  Thomas Vidick and
                  Lisa Yang},
  editor       = {Yevgeniy Dodis and
                  Thomas Shrimpton},
  title        = {Succinct Classical Verification of Quantum Computation},
  booktitle    = {Advances in Cryptology - {CRYPTO} 2022 - 42nd Annual International
                  Cryptology Conference, {CRYPTO} 2022, Santa Barbara, CA, USA, August
                  15-18, 2022, Proceedings, Part {II}},
  series       = {Lecture Notes in Computer Science},
  volume       = {13508},
  pages        = {195--211},
  publisher    = {Springer},
  year         = {2022},
  url          = {https://doi.org/10.1007/978-3-031-15979-4\_7},
  doi          = {10.1007/978-3-031-15979-4\_7},
  bibsource    = {dblp computer science bibliography, https://dblp.org}
}

@article{mnz24,
  author       = {Tony Metger and
                  Anand Natarajan and
                  Tina Zhang},
  title        = {Succinct arguments for {QMA} from standard assumptions via compiled
                  nonlocal games},
  journal      = {CoRR},
  volume       = {abs/2404.19754},
  year         = {2024},
  url          = {https://doi.org/10.48550/arXiv.2404.19754},
  doi          = {10.48550/ARXIV.2404.19754},
  eprinttype    = {arXiv},
  eprint       = {2404.19754},
  bibsource    = {dblp computer science bibliography, https://dblp.org}
}

@inproceedings{gjmz23,
  author       = {Sam Gunn and
                  Nathan Ju and
                  Fermi Ma and
                  Mark Zhandry},
  editor       = {Barna Saha and
                  Rocco A. Servedio},
  title        = {Commitments to Quantum States},
  booktitle    = {Proceedings of the 55th Annual {ACM} Symposium on Theory of Computing,
                  {STOC} 2023, Orlando, FL, USA, June 20-23, 2023},
  pages        = {1579--1588},
  publisher    = {{ACM}},
  year         = {2023},
  url          = {https://doi.org/10.1145/3564246.3585198},
  doi          = {10.1145/3564246.3585198},
  bibsource    = {dblp computer science bibliography, https://dblp.org}
}

@article{aav13,
  author       = {Dorit Aharonov and
                  Itai Arad and
                  Thomas Vidick},
  title        = {The Quantum {PCP} Conjecture},
  journal      = {CoRR},
  volume       = {abs/1309.7495},
  year         = {2013},
  url          = {http://arxiv.org/abs/1309.7495},
  eprinttype    = {arXiv},
  eprint       = {1309.7495},
  bibsource    = {dblp computer science bibliography, https://dblp.org}
}

@inproceedings{bh13,
  author       = {Fernando G. S. L. Brand{\~{a}}o and
                  Aram W. Harrow},
  editor       = {Dan Boneh and
                  Tim Roughgarden and
                  Joan Feigenbaum},
  title        = {Product-state approximations to quantum ground states},
  booktitle    = {Symposium on Theory of Computing Conference, STOC'13, Palo Alto, CA,
                  USA, June 1-4, 2013},
  pages        = {871--880},
  publisher    = {{ACM}},
  year         = {2013},
  url          = {https://doi.org/10.1145/2488608.2488719},
  doi          = {10.1145/2488608.2488719},
  bibsource    = {dblp computer science bibliography, https://dblp.org}
}

@article{fh14,
  title={Quantum systems on non-k-hyperfinite complexes: a generalization of classical statistical mechanics on expander graphs},
  author={Freedman, Michael H and Hastings, Matthew B},
  journal={Quantum Information \& Computation},
  volume={14},
  number={1-2},
  pages={144--180},
  year={2014},
  publisher={Rinton Press, Incorporated Paramus, NJ}
}

@article{nn24,
  author       = {Anand Natarajan and
                  Chinmay Nirkhe},
  title        = {The status of the quantum {PCP} conjecture (games version)},
  journal      = {CoRR},
  volume       = {abs/2403.13084},
  year         = {2024},
  url          = {https://doi.org/10.48550/arXiv.2403.13084},
  doi          = {10.48550/ARXIV.2403.13084},
  eprinttype    = {arXiv},
  eprint       = {2403.13084},
  bibsource    = {dblp computer science bibliography, https://dblp.org}
}

@inproceedings{fgl+91,
  author       = {Uriel Feige and
                  Shafi Goldwasser and
                  L{\'{a}}szl{\'{o}} Lov{\'{a}}sz and
                  Shmuel Safra and
                  Mario Szegedy},
  title        = {Approximating Clique is Almost NP-Complete (Preliminary Version)},
  booktitle    = {32nd Annual Symposium on Foundations of Computer Science, San Juan,
                  Puerto Rico, 1-4 October 1991},
  pages        = {2--12},
  publisher    = {{IEEE} Computer Society},
  year         = {1991},
  url          = {https://doi.org/10.1109/SFCS.1991.185341},
  doi          = {10.1109/SFCS.1991.185341},
  bibsource    = {dblp computer science bibliography, https://dblp.org}
}

@inproceedings{bjsw16,
  author       = {Anne Broadbent and
                  Zhengfeng Ji and
                  Fang Song and
                  John Watrous},
  editor       = {Irit Dinur},
  title        = {Zero-Knowledge Proof Systems for {QMA}},
  booktitle    = {{IEEE} 57th Annual Symposium on Foundations of Computer Science, {FOCS}
                  2016, 9-11 October 2016, Hyatt Regency, New Brunswick, New Jersey,
                  {USA}},
  pages        = {31--40},
  publisher    = {{IEEE} Computer Society},
  year         = {2016},
  url          = {https://doi.org/10.1109/FOCS.2016.13},
  doi          = {10.1109/FOCS.2016.13},
  bibsource    = {dblp computer science bibliography, https://dblp.org}
}

@inproceedings{nz23,
  author       = {Anand Natarajan and
                  Tina Zhang},
  title        = {Bounding the Quantum Value of Compiled Nonlocal Games: From {CHSH}
                  to {BQP} Verification},
  booktitle    = {64th {IEEE} Annual Symposium on Foundations of Computer Science, {FOCS}
                  2023, Santa Cruz, CA, USA, November 6-9, 2023},
  pages        = {1342--1348},
  publisher    = {{IEEE}},
  year         = {2023},
  url          = {https://doi.org/10.1109/FOCS57990.2023.00081},
  doi          = {10.1109/FOCS57990.2023.00081},
  bibsource    = {dblp computer science bibliography, https://dblp.org}
}

@inproceedings{bcs16,
  author       = {Eli Ben{-}Sasson and
                  Alessandro Chiesa and
                  Nicholas Spooner},
  editor       = {Martin Hirt and
                  Adam D. Smith},
  title        = {Interactive Oracle Proofs},
  booktitle    = {Theory of Cryptography - 14th International Conference, {TCC} 2016-B,
                  Beijing, China, October 31 - November 3, 2016, Proceedings, Part {II}},
  series       = {Lecture Notes in Computer Science},
  volume       = {9986},
  pages        = {31--60},
  year         = {2016},
  url          = {https://doi.org/10.1007/978-3-662-53644-5\_2},
  doi          = {10.1007/978-3-662-53644-5\_2},
  bibsource    = {dblp computer science bibliography, https://dblp.org}
}

@inproceedings{rrr16,
  author       = {Omer Reingold and
                  Guy N. Rothblum and
                  Ron D. Rothblum},
  editor       = {Daniel Wichs and
                  Yishay Mansour},
  title        = {Constant-round interactive proofs for delegating computation},
  booktitle    = {Proceedings of the 48th Annual {ACM} {SIGACT} Symposium on Theory
                  of Computing, {STOC} 2016, Cambridge, MA, USA, June 18-21, 2016},
  pages        = {49--62},
  publisher    = {{ACM}},
  year         = {2016},
  url          = {https://doi.org/10.1145/2897518.2897652},
  doi          = {10.1145/2897518.2897652},
  bibsource    = {dblp computer science bibliography, https://dblp.org}
}

@inproceedings{aalv09,
  title={The detectability lemma and quantum gap amplification},
  author={Aharonov, Dorit and Arad, Itai and Landau, Zeph and Vazirani, Umesh},
  booktitle={Proceedings of the forty-first annual ACM symposium on Theory of computing},
  pages={417--426},
  year={2009}
}

@article{an02,
  title={Quantum NP-a survey},
  author={Aharonov, Dorit and Naveh, Tomer},
  journal={arXiv preprint quant-ph/0210077},
  year={2002}
}

@inproceedings{abn23,
  author       = {Anurag Anshu and
                  Nikolas P. Breuckmann and
                  Chinmay Nirkhe},
  editor       = {Barna Saha and
                  Rocco A. Servedio},
  title        = {{NLTS} Hamiltonians from Good Quantum Codes},
  booktitle    = {Proceedings of the 55th Annual {ACM} Symposium on Theory of Computing,
                  {STOC} 2023, Orlando, FL, USA, June 20-23, 2023},
  pages        = {1090--1096},
  publisher    = {{ACM}},
  year         = {2023},
  url          = {https://doi.org/10.1145/3564246.3585114},
  doi          = {10.1145/3564246.3585114},
  bibsource    = {dblp computer science bibliography, https://dblp.org}
}

@article{as98,
  author       = {Sanjeev Arora and
                  Shmuel Safra},
  title        = {Probabilistic Checking of Proofs: {A} New Characterization of {NP}},
  journal      = {J. {ACM}},
  volume       = {45},
  number       = {1},
  pages        = {70--122},
  year         = {1998},
  url          = {https://doi.org/10.1145/273865.273901},
  doi          = {10.1145/273865.273901},
  bibsource    = {dblp computer science bibliography, https://dblp.org}
}

@article{alm+98,
  author       = {Sanjeev Arora and
                  Carsten Lund and
                  Rajeev Motwani and
                  Madhu Sudan and
                  Mario Szegedy},
  title        = {Proof Verification and the Hardness of Approximation Problems},
  journal      = {J. {ACM}},
  volume       = {45},
  number       = {3},
  pages        = {501--555},
  year         = {1998},
  url          = {https://doi.org/10.1145/278298.278306},
  doi          = {10.1145/278298.278306},
  bibsource    = {dblp computer science bibliography, https://dblp.org}
}

@article{nv16,
  author       = {Anand Natarajan and
                  Thomas Vidick},
  title        = {Robust self-testing of many-qubit states},
  journal      = {CoRR},
  volume       = {abs/1610.03574},
  year         = {2016},
  url          = {http://arxiv.org/abs/1610.03574},
  eprinttype    = {arXiv},
  eprint       = {1610.03574},
  bibsource    = {dblp computer science bibliography, https://dblp.org}
}

@inproceedings{klvy22,
  author       = {Yael Kalai and
                  Alex Lombardi and
                  Vinod Vaikuntanathan and
                  Lisa Yang},
  editor       = {Barna Saha and
                  Rocco A. Servedio},
  title        = {Quantum Advantage from Any Non-local Game},
  booktitle    = {Proceedings of the 55th Annual {ACM} Symposium on Theory of Computing,
                  {STOC} 2023, Orlando, FL, USA, June 20-23, 2023},
  pages        = {1617--1628},
  publisher    = {{ACM}},
  year         = {2023},
  url          = {https://doi.org/10.1145/3564246.3585164},
  doi          = {10.1145/3564246.3585164},
  bibsource    = {dblp computer science bibliography, https://dblp.org}
}
